\documentclass[a4paper,11pt]{article}
\pdfoutput=1 % if your are submitting a pdflatex (i.e. if you have
\usepackage{jheppub} % for details on the use of the package, please
\usepackage{graphicx}
\usepackage{amsmath}
\usepackage{amssymb}
\usepackage{amsthm}
\usepackage{physics}
\usepackage{bbm}
\usepackage{url}

\newtheorem{theorem}{Theorem}
\newtheorem{lemma}{Lemma}

\newtheorem{corollary}{Corollary}

\newtheorem{definition}{Definition}

\title{Recovering the QNEC from the ANEC}

\author{Fikret Ceyhan and}
\author{Thomas Faulkner}

\affiliation[a]{Department of Physics, University of Illinois, Urbana-Champaign \\ 1110 W. Green St., Urbana IL 61801, USA.}

\abstract{
We study the relative entropy in QFT comparing the vacuum state to a special family of purifications determined by an input state and constructed using relative modular flow. We use this to prove a conjecture by Wall that relates the shape derivative of relative entropy to a variational expression over the averaged null energy (ANE) of possible purifications. This variational expression can be used to easily prove the quantum null energy condition (QNEC). 
We formulate 
Wall's conjecture as a theorem pertaining to operator algebras satisfying the properties of a half-sided modular inclusion, with the additional assumption that the input state has  finite averaged null energy. 
We also give a new derivation of the strong superadditivity property of relative entropy in this context.  We speculate about possible connections to the recent methods used to strengthen monotonicity of relative entropy with recovery maps. }

\begin{document} 
\maketitle
\flushbottom

%\title{\boldmath Recovering the QNEC from the ANEC}

%\begin{document} 
%\maketitle
%\flushbottom

\section{Introduction}

The main goal of this paper is to present a mathematically rigorous proof of the quantum null energy condition (QNEC) in the context of algebraic QFT. The QNEC is a local bound on the expectation value of the null energy density \cite{Bousso:2015wca}. In certain situations it can be related to the positivity of the second derivative of relative entropy thought of as a function of the shape of an entangling surface that cuts the generators of a killing horizon. This convexity constraint is in turn related to the so called quantum focusing conjecture (QFC) \cite{bousso2016quantum} who's subject is the generalized area $A_{\rm gen}/(4G_N)$ \cite{Bekenstein:1973ur,Casini:2008cr,wall2012proof}. For the horizon cuts considered here, and in a semi-classical limit, the generalized area reduces to  $- S_{\rm rel} +$ constant.
Since entanglement entropy is not well defined in the continuum limit where we work \cite{Witten:2018lha}, the bound in terms of relative entropy will be our goal. We specialize here to relativistic QFT in $d$-dimensional Minkowski space with $d\geq 2$ and with cuts along a Rindler horizon.

Previous proofs \cite{Bousso:2015wca,Balakrishnan:2017bjg} used ideas that are hard to make mathematically rigorous in general, such as path integrals and the replica trick. 
These path integral/replica methods \cite{Holzhey:1994we} are one of the more powerful tools that we have for uncovering properties of entanglement in QFT \cite{Calabrese:2004eu} and in AdS/CFT \cite{Lewkowycz:2013nqa}.
However it is worth spelling out more rigorous approaches, if they are available, since they can lead to their own insights. See \cite{Hollands:2017dov,Longo:2018obd,Longo:2017mbg,Xu:2018fsv,Xu:2018uxc,Kang:2018xqy} for some recent progress along these lines. 
In this paper we will take inspiration from the previous QNEC proof for interacting theories \cite{Balakrishnan:2017bjg} as well as some ideas laid out by Wall \cite{Wall:2017blw}. In this way we unify these two seemingly disparate approaches and ``explain'' the somewhat mysterious correlators in \cite{Balakrishnan:2017bjg} that were used to extract the QNEC. 

The main lesson can be summed up as follows. \emph{The QNEC reduces to the ANEC in a new state constructed from the original state with relative modular flow.} The ANEC has been proven now in various ways \cite{Klinkhammer:1991ki,Kelly:2014mra,Faulkner:2016mzt,Hartman:2016lgu,Kravchuk:2018htv}. %The purifications are constructed using relative modular flow. 
We start, in Section~\ref{sec:ant}, by describing the relative entropies of these new states which can be almost completely understood.  The missing ingredient being the averaged null energy (ANE) which the bulk of this paper is dedicated to finding;  we do so with two lemmas: Lemma~\ref{lemma2} is proven in Section~\ref{sec:nullant} and Lemma~\ref{lemma3} is proven in \ref{sec:nullcone}. 
The relative entropies satisfy an important constraint, Lemma~\ref{lemma1}, that is well known but non-trivial to derive in the algebraic context - we do this in Section~\ref{sec:rel}.
Our main mathematical tool will be the algebraic structure of half-sided modular inclusions \cite{borchers1992cpt,wiesbrock1993half,borchers1996half,araki2005extension}, the relative modular operators which we summarize in Appendix~\ref{app:mod}, and some elementary theorems on holomorphic functions (including holomorphic functions of two variables.) For example these later theorems allow us to  give a rigorous example of the saturation of a (modular) chaos bound \cite{Maldacena:2015waa}, a delicate phenomenon that occurs for an analagous CFT four point function \cite{Hartman:2015lfa} expanded using the light cone OPE and continued to a Lorentzian regime.

One new result that we would like to advertise is an expression for the shape variation of the relative entropy, comparing some vector state $\psi$ with the vacuum, and for null cuts with some shape $x^+(y)$:
\begin{align}
-\frac{1}{2\pi} \frac{ \delta S_{\rm rel}(\psi|\Omega;x^+(\cdot ))}{\delta x^+(y)} &= \inf_{\phi}   \left< \phi \right| \mathcal{E}_+(y)  \left| \phi \right> =  \inf_{s \in \mathbb{R}} \left<\psi_s \right| \mathcal{E}_+(y)  \left| \psi_s \right> \, \\ \qquad \mathcal{E}_+(y) &\equiv 2\pi \int_{-\infty}^{\infty} dx^+T_{++}(x^+,y)
\end{align}
where in the first expression $\phi = u' \psi$ for some unitary acting in the complement to the entangling region. The second expression gives the explicit minimal value where $\psi_s = u_s' \psi$ is simply constructed with relative modular flow (more precisely with the Connes cocycle.) 
These formulas assume the averaged null energy (ANE) for the input $\psi$ is finite\footnote{Actually we only need that there is at least one state $u' \psi$ with finite ANE.}. The minimization is  over the set of states that also have finite complementary relative entropy and we will  show that there is always at least one such state. 

Our work was initiated as an attempt to apply some recent results in quantum information \cite{fawzi2015quantum,wilde2015recoverability,junge2018universal,swingle2018recovery}  that give useful strengthening of the monotonicity of relative entropy inequalities.  Since the ANEC is tightly linked to monotonicity and the QNEC seems like a strengthening of the ANEC it is natural to guess that there is an interesting connection here. More specifically the strengthened inequalities improve monotonicity using certain recovered states that attempt to optimally invert a given quantum channel, which in the case at hand is simply related to an inclusion of algebras.  It is interesting to speculate that there might be a relation between the universal recovered state in \cite{junge2018universal} and the various purifications that we discuss in this paper. In particular they are both constructed with modular flow.  This might of course just be a coincidence and we would like to know if there is more to it than this.  In the discussion section we give some ideas about how this connection might work.

\section{The Ant's Best Guess}

\label{sec:ant}

We start in $d$-dimensional Minkowski space with null coordinates associated to a Rinlder horizon $v=0$:
\begin{equation}
ds^2 = - du dv + dy_{d-2}^2
\end{equation}
The Rindler wedge $R$ is the right region $\{ u > 0, v < 0\}$ with the associated algebra of operators $\mathcal{A}_R$.  
We now define a generalization of the Rinder wedge and the associated algebra that we will collectively refer to as \emph{null cuts}.
Consider a null cut $N_C= \{v=0,u>C(y))\}$ where $C(y)$ is a continuous function of the coordinates $y$ along the entangling surface.
Define $N_C'$ as the maximal open subset spacelike separated from $N_C$  and so forth for $N_C'' = (N_C')'$. Then $N_C''$ is an open space-time region for which we can associate a von Neumann algebra. As a short hand we will label this as $\mathcal{A}_C
\equiv \mathcal{A}_{N_C''} $. This algebra can heuristically be thought of as the double commutant of the local operators on $N_C$ \cite{wittenpitp}. In this notation $\mathcal{A}_R = \mathcal{A}_0$.

The vacuum state $\Omega$ is cyclic and separating for all the algebra's that we consider here, a property which follows from the Reeh-Schlieder theorem \cite{reeh1961bemerkungen}. Applying Tomita-Takesaki theory we can define the associated modular operators in the usual way. 
For the Rindler cut, the Bisognano-Wichmann theorem \cite{Bisognano:1976za} shows that the modular operator $\Delta_{\Omega;R}^{is}$ is simply the boost that fixes the entangling surface $u=v=0$. For other null cuts the vacuum modular Hamiltonian's have been the subject of recent investigation \cite{wall2012proof,Faulkner:2016mzt,koeller2018local,Casini:2017roe}. The modular Hamiltonian is defined as $K_A \equiv - \ln \Delta_{\Omega;A}$ and the results of \cite{Casini:2017roe} showed that:
\begin{equation}
K_A = 2\pi \int_{v=0} (u -A(y)) T_{uu}
\end{equation}
If we consider two null cuts then the modular Hamiltonian's satisfy the following algebra:
\begin{equation}
\left[ K_{A}, K_{B} \right] =  2\pi i (K_A - K_B) \equiv (2\pi)^2 i P
\end{equation}
where $P$ is a modular translation operator who's action on the null lines of the Rindler horizon is a translation by $B(y)-A(y)$. Furthermore these operators can be related to the averaged null energy (ANE):
\begin{equation}
\label{nullstress}
P =  \int_{v=0} (B(y)-A(y)) T_{uu} 
\end{equation}
If $\mathcal{A}_B \subset \mathcal{A}_A$ then one has a situation that is referred to as a half sided modular inclusion (HSMI) or translation \cite{borchers1992cpt,wiesbrock1993half,borchers1996half,araki2005extension,borchers1995use}.  In this case $P\geq 0$, since such an inclusion implies that the difference in modular Hamiltonians is a positive semi-definite operator \cite{buchholz1990nuclear,Witten:2018lha}.  Considering the relationship to the null energy, \eqref{nullstress}, this then proves the ANEC \cite{Faulkner:2016mzt}. 

We will mostly work in this context where we have algebras satisfying the properties of a HSMI, and where we have in mind applications to QFT for null cuts. In particular a HSMI is a well studied  algebraic structure in which we need not make any mention of the stress tensor $T_{uu}$. Note that the stress tensor is a local operator that is often not included in the basic axioms of algebraic QFT. 
For a recent application of HSMI to black hole physics see \cite{Jefferson:2018ksk}.

Some of the properties of HSMI are given in the following definition:

\begin{definition}[Half-sided modular inclusion]
\label{def1}

An inclusion of von Neumann algebra's $\mathcal{A}_B \subset \mathcal{A}_A$ is called a half-sided modular inclusion if there is a common cyclic and separating vector $\Omega \in \mathcal{H}$ such that:
\begin{equation}
\Delta_{\Omega;A}^{-is} \mathcal{A}_B \Delta_{\Omega;A}^{is} \subset \mathcal{A}_B \qquad s \geq 0
\end{equation}
From this minimal starting assumption one can derive the following. Let $P$ be the closure of:
\begin{equation}
\frac{1}{2\pi} \left( \ln \Delta_{\Omega;B} - \ln \Delta_{\Omega;A} \right)
\end{equation}
Then $P$ is a self-adjoint positive semi-definite operator. One can also derive the following results \cite{borchers1992cpt,wiesbrock1993half,borchers1996half,araki2005extension,borchers1995use}:
\begin{enumerate}
\item[(a)] The modular translations $U_b \equiv e^{ - i b P}$ act as:
\begin{equation}
\label{hsmi:translation}
U_{-b} \mathcal{A}_{A} U_{b}  \equiv \mathcal{A}_{A_b} \left( \subset \mathcal{A}_{A_{b_2}}\,, \quad b \geq b_2 \right)
\end{equation}
where $\mathcal{A}_{A_{1/2\pi}} =\mathcal{A}_{B}$ and $U_b$ leaves invariant the vacuum $\Omega$. This new one parameter ``translated'' family of algebra's has $\Omega$ as a common cyclic and separating vector. The translations apply for all real $b$ and for the commutant's which satisfy:
\begin{equation}
 \mathcal{A}_{A_{b}}' \subset \mathcal{A}_{A_{b_2}}' \qquad b \leq b_2
\end{equation}
\item[(b)] We can also``boost'' these translated algebras:
\begin{equation}
\label{hsmi:boost}
\Delta_{\Omega;A}^{-is} \mathcal{A}_{A_b} \Delta_{\Omega;A}^{is} = \mathcal{A}_{A_{b e^{2\pi s} }}
\end{equation}
and similarly for the commutant.
\item[(c)] The modular operators satisfy:
\begin{align}
\label{hsmi:algebra}
\Delta_{\Omega;A}^{is} U_b \Delta_{\Omega;A}^{-is} &=   U_{e^{-2\pi s}b}  \,,
 \qquad J_{\Omega;A} U_b J_{\Omega;A}  =U_{-b} \\
 % \Delta^{is}_{\Omega;A} \Delta_{\Omega;B}^{-is} = U_{(1-e^{-2\pi s})/2\pi} & \qquad  
  \Delta^{is} _{\Omega;A_{b_1}} \Delta_{\Omega;A_{b_2}}^{-is} &= U_{ (b_1-b_2) (e^{-2\pi s} -1)} \nonumber
\end{align}
furthermore $U_b$ varies continuously in the strong operator topology (sot) and can be analytically continued into the complex $b$ plane where it is bounded by $1$ and (sot) continuous for ${\rm Im b} \leq 0$.
\end{enumerate}

\end{definition}
This definition applies abstractly to von Neumann algebra's and as already mentioned one can work entirely from this point of view. At the same time, however, our notation is uniform with the application to null cuts of a Rindler horizon. For example in this later notation $A_b = A +2\pi b(B-A)$ where $B(y) \geq A(y)$.  We will also sometimes use the following notation:
\begin{equation}
C \equiv A_{c} \qquad C_a \equiv A_{c+a}
\end{equation}
where the $C$ cut with $b=c$ plays a distinguished role. 

We now consider an excited state $\psi$ which we take to be a vector in the QFT Hilbert space. For now we will assume that this state has the following finite quantities:
\begin{equation}
\label{assumpt}
P_\psi = \left< \psi \right| P \left| \psi \right> < \infty\,,  \qquad S_{\rm rel}( \psi | \Omega; A_{c}) < \infty \,, 
\qquad  S_{\rm rel}( \psi | \Omega; A'_{c}) < \infty
\end{equation}
for some $c$, and where $S_{\rm rel}$ is the relative entropy discussed by Araki \cite{araki1976relative} and which is defined for general states \cite{araki1977relative}. In particular we \emph{do not} assume that $\psi$ is cyclic and separating. This definition uses the relative modular operator which is defined in general via the Tomita operator $S$:
\begin{equation} 
\label{tom}
S_{\psi|\Omega;A_b}\left( \alpha \left| \psi \right> + \left| \chi' \right> \right) = \pi_{A_b}(\psi) \alpha^\dagger \left| \Omega \right>  \qquad \forall \,\, \alpha \in \mathcal{A}_{A_b}\,,\,\, \chi' \in (1-\pi_{A_b'}(\psi)) \mathcal{H} 
\end{equation}
In the above definition $\pi_{A_b}(\psi)$ is the support projector, which is the smallest projector in $\mathcal{A}_{A_b}$ satisfying $\pi_{A_b}(\psi) \left| \psi \right> = \left| \psi \right>$. For a cyclic and separating vector both $\pi_{A_b,A_b'}(\psi)$ are the unit operator. See Appendix~\ref{app:mod} for further discussion of these. Note that \eqref{tom} only really defines $S_{\psi|\Omega;A_b}$ for a dense set of states in $\mathcal{H}$, however one can show that this operator is closeable \cite{araki1982positive} and we will use the same symbol for its closure. The modular operator is defined as:\footnote{Our conventions are not standard. The state labels on the relative modular operators are switched. We follow the conventions in \cite{Witten:2018lha} where the relative entropy and relative modular operators are labelled in the same way. Our labelling on the Connes cocycle are standard. }
\begin{equation}
\Delta_{\psi|\Omega;A_b} = S_{\psi|\Omega;A_b}^\dagger S_{\psi|\Omega;A_b}
\end{equation}
with support $\pi'_{A_b}(\psi)$. This then leads to Araki's definition of relative entropy (where $\Omega$ is cyclic and separating):
\begin{equation}
\label{araki-srel}
S_{\rm rel}(\psi|\Omega;A_b) =  - \left< \psi \right| \log \Delta_{\psi|\Omega} \left| \psi \right>
\equiv - \int_0^{\infty} \log \lambda d\left< \psi \right| E_{\lambda}(\Delta_{\psi|\Omega}) \left| \psi \right>
\end{equation}
where $E_{\lambda}(\Delta)$ are the spectral projections of $\Delta$.  The relative entropy could be infinite if this later integral diverges. Note that since $ \left| \psi \right>$ is in the domain of $\Delta_{\psi|\Omega}^{1/2}$ the following integral always converges:
\begin{equation}
\int_0^\infty \lambda  d\left< \psi \right| E_{\lambda}(\Delta_{\psi|\Omega})\left| \psi \right> < \infty
\end{equation}
which implies that any divergence in \eqref{araki-srel} comes from the lower end as $\lambda \rightarrow 0$.

One expression for relative entropy that we will find useful is due to Uhlmann \cite{uhlmann1977relative}
\begin{equation}
\label{reltheta}
S_{\rm rel}(\psi|\Omega;A_b) = \lim_{\theta \rightarrow 0^+} \frac{1- \left< \psi \right|  \Delta_{\psi|\Omega}^\theta \left| \psi \right>}{\theta}
\end{equation}
and this definition is equivalent to \eqref{araki-srel} since $(1-\lambda^\theta)/\theta$ is a decreasing (increasinng) function of $\theta$ for all $0< \lambda < 1$ ($1 < \lambda < \infty$), so we can use the monotone convergence theorem for the integral in the spectral representation \cite{ohya2004quantum}.  

Now consider the following functions:
\begin{equation}
S(b) \equiv S_{\rm rel}( \psi | \Omega; A_b) \qquad \bar{S}(b) \equiv S_{\rm rel}( \psi | \Omega; A'_b) 
\end{equation}
Under the conditions specified in \eqref{assumpt} for $\psi$ one can show that $S(b)$  $(\bar{S}(b))$ is a continuous monotonically decreasing (increasing) function for all $b \in \mathbb{R}$. Monotonicity is a classic result for relative entropy \cite{lieb1973proof,uhlmann1977relative,araki1977relative}. Continuity follows from the  following relation:

\begin{lemma}
\label{lem:srane}
\label{lemma1}
Under the assumptions of \eqref{assumpt}:
\begin{equation}
\label{dsdsp}
-\left( S(b_2)-S(b_1) \right)  +\left( \bar{S}(b_2)-\bar{S}(b_1) \right) = (b_2-b_1) 2\pi P_\psi  
\end{equation}
and this, combined with monotonicity, implies that $S(b),\bar{S}(b)$ are everywhere finite and Lipschitz continuous.
\end{lemma}
The proof of this Lemma~\ref{lemma1} is the subject of Section~\ref{sec:rel}. 
In previous works this relationship was essentially taken to be an obvious consequence of the form of relative entropy written in terms of the (half) modular energy and entanglement entropy \cite{bousso2015entropy,blanco2013localization,Faulkner:2016mzt}.
These arguments are based on assuming a tensor factorization and working with density matrices  (or a regularization consistent with this). For example in \cite{Faulkner:2016mzt} equation \eqref{dsdsp} was used to motivate the ANEC.\footnote{From the algebraic point of view $P\geq 0$ follows more directly from properties of modular Hamiltonians under inclusion \cite{Witten:2018lha}.} So it might come as a surprise that we have to devote a whole section to proving this. It turns out that this relation is simple to derive if one assumes that all relative entropies in \eqref{dsdsp} are finite to begin with. We would like to not assume this, and in fact we would like to use this equation as a tool to derive when some relative entropies are finite given some other ones are finite. This is a non-trivial task but we managed to get it to work with the assumptions in \eqref{assumpt} in which case we learn that $\bar{S}(b_2)$ is finite for $b_2 > c$ and this finitness does not follow from monotonicity. 
It then follows that all relative entropies are finite. These considerations are fundamentally important for proceeding to compute the relative entropies of the various purifications that we discuss next.

For some of this discussion we will be interested in $\psi$ restricted to $\mathcal{A}_{A_c} \equiv \mathcal{A}_{C}  $ and purifications thereof. Since $\psi$ is a vector in the Hilbert space, this represents one such purification. Any other purification can be constructed from $\psi$ with the action of a unitary from the commutant algebra $\mathcal{A}_{C}'$.

In the following discussion we will often drop the $C$ label on the modular operators for the $\mathcal{A}_C$ algebra, since this is the most common algebra we write. 
Consider the Connes cocycle, which is defined as:
\begin{equation}
u_s = (D \Omega:D\psi )_s = \Delta_{\Omega}^{is} \Delta_{\Omega|\psi}^{-is}  \in \mathcal{A}_C 
\end{equation}
for real $s$, where $\Delta_{\Omega} \equiv \Delta_{\Omega|\Omega}$. The fact that this is an operator in the algebra $\mathcal{A}_C$ is a non-trivial result of Tomita-Takesaki theory applied to an enlarged Hilbert space (by a few qudits) using the doubling trick  that we review in Appendix~\ref{app:mod}. We define powers of the modular operator, for example $\Delta_{\Omega|\psi}^{is}$, on the subspace of the Hilbert space with non-zero support for the operator: $\pi(\psi) \mathcal{H}$ in this case. We also define such powers to annihilate the kernel, $(1-\pi(\psi)) \mathcal{H}$. For example
this means that $\lim_{s \rightarrow 0} \Delta_{\Omega|\psi}^{is} = \pi(\psi)$.
An alternative expression for the cocycle is:
\begin{equation}
 \Delta_{\psi|\Omega}^{is} \Delta_{\psi}^{-is} = u_s \pi'(\psi)
\end{equation}
which requires the additional support projector in $\mathcal{A}_C'$ and is sometimes less convenient, however for us this will often not matter since we will take the cocycle to act on $\left| \psi \right>$ where we can drop the support projector.

Similarly we have a cocycle for the complement:
\begin{equation}
u_s' = (D \Omega:D\psi )'_s = (\Delta_{\Omega}')^{is} (\Delta_{\Omega|\psi}')^{-is}  \in \mathcal{A}'_C 
\end{equation}
Note that $u_s'$ is in general not unitary. Instead it is a partial isometry satisfying:
\begin{equation}
(u_s')^\dagger u_s' =  \pi'(\psi) \qquad  u_s' (u_s')^\dagger = \Delta_{\Omega}^{-is}  \pi'(\psi) \Delta_{\Omega}^{is}
\end{equation}

The following states are interesting purifications of $\psi$ restricted to $\mathcal{A}_C$: 
\begin{equation}
\left| \psi_s \right> \equiv u_s' \left| \psi \right>  \qquad s\in\mathbb{R}
\label{p:flow}
\end{equation}
This state preserves all expectation values of operators in $\mathcal{A}_C$:
\begin{equation}
\label{sameexp}
\left< \psi \right| (u_s')^\dagger \gamma  u_s' \left| \psi \right> 
= \left< \psi \right| \gamma (u_s')^\dagger u_s'   \left| \psi \right>  = \left< \psi \right| \gamma \pi'(\psi) \left| \psi \right> =  \left< \psi \right| \gamma \left| \psi \right> \qquad \gamma \in \mathcal{A}_C
\end{equation}

We would like to compute the relative entropy of this purification. This state also preserves expectation values of operators in $\mathcal{A}_{C_a}   \equiv \mathcal{A}_{A_{c+a}} \subset \mathcal{A}_{C}$ for $a>0$ so we conclude that:
\begin{equation}
S_{\rm rel}(\psi_s |\Omega ; C_a) = S_{\rm rel}(\psi |\Omega ; C_a)  \qquad a \geq 0
\end{equation}
since the relative entropy can be shown to be independent of the vector representation, only depending on the linear functional that the state induces on operators \cite{araki1976relative}.
Using the relationship:
\begin{equation}
(\Delta_{\Omega|\psi}')^{-is} = \Delta_{\psi|\Omega}^{is} 
\end{equation}
discussed in Appendix~\ref{app:mod}, this purification can be written as:
\begin{equation}
\left| \psi_s \right> = u_s' \left| \psi \right> = \Delta_{\Omega}^{-is} u_s \left| \psi \right>
\end{equation}
Such that:
\begin{equation}
\left< \psi_s \right| \gamma' \left| \psi_s \right>
= \left< \psi \right|  \Delta_{\Omega}^{is} \gamma' \Delta_{\Omega}^{-is}  \left| \psi \right>\qquad \gamma' \in \mathcal{A}_C'
\end{equation}
Thus the complement relative entropy matches the complement relative entropy of the state $\Delta_{\Omega}^{-is}  \left| \psi \right>$.
We can compute the relative entropy by constructing the relative modular operator for the cuts $\mathcal{A}_{C'_a}$ for $a<0$. Using the algebra of half-sided modular inclusions \eqref{hsmi:translation}-\eqref{hsmi:boost} we find:
\begin{equation}
\label{sflow}
S_{\Delta_{\Omega;C}^{-is} \psi| \Omega; C_a'} = \Delta_{\Omega;C}^{-is} S_{\psi|\Omega;C'_{a e^{-2\pi s}}}  \Delta_{\Omega;C}^{is}
\end{equation}
and where the support projectors satisfy:
\begin{equation}
\pi_{C_a}(\Delta_{\Omega;C}^{-is} \psi) = \Delta_{\Omega;C}^{-is} \pi_{C_{ae^{-2\pi s}}} (\psi) \Delta_{\Omega;C}^{is}
\end{equation}
and similarly for the complement support projector. We can then construct the relative modular operator and use this to compute the relative entropy. The answer is simply:
\begin{equation}
S_{\rm rel}(\psi_s |\Omega ; C'_a) =S_{\rm rel}(\Delta_{\Omega}^{-is} \psi |\Omega ; C'_a)   = S_{\rm rel}(\psi |\Omega ; C'_{a e^{-2\pi s}})  \qquad a \leq 0
\end{equation}
Thus it is easy to compute the relative entropy for $\mathcal{A}_{C_a} \subset \mathcal{A}_{C}$ or $\mathcal{A}_{C_a}' \subset \mathcal{A}'_{C}$ in terms of the input relative entropy for $\psi$.  This is because the state $\left|\psi_s\right>$ is roughly a half sided boost of $\left| \psi \right>$, leaving one side invariant as above.\footnote{This should not be confused with the boosted states discussed in \cite{Jafferis:2014lza,Faulkner:2018faa}. These states are more singular since they involve modular flow with only half the $\psi$-modular Hamiltonian. The Connes cocycle is one way to deal with issues related to divergences that arise in that case, and some of the resulting physics is related to that discussed in \cite{Jafferis:2014lza,Faulkner:2018faa}. In particular we expect the bulk description of these states, in the context of AdS/CFT, to be the same for the part of the bulk spacetime that is (bulk) causally separated from the boundary entangling surface. } The other cases, such as $\mathcal{A}_{C_a}'$ for $a \geq 0$ are harder, since the cocycle acts simply as a half sided boost only on some of the operators in this algebra but not all. 

It turns out however that all we need to know to complete the full picture of relative entropies is the averaged null energy of this purification. In fact all we need is the following lemma:
\begin{lemma}
\label{lem:ant}
\label{lemma2}
For a vector state $\psi$ that has finite $P_\psi < \infty$ then:
\begin{equation}
\label{aps}
P_{s} \equiv \left< \psi_s \right| P \left| \psi_s\right>  = R + e^{-2\pi s} \left( P_\psi -R \right)
\end{equation} 
with $0 \leq R \leq P_\psi$ independent of $s$, and 
where the state $\psi_s$ was defined in \eqref{p:flow}.
\end{lemma}
We will delay the proof of this Lemma to Section~\ref{sec:nullant}, although we should stress that we think that this is the most interesting part of this paper. The rough sketch of how this goes is that we prove that $P_{s}$ is an entire function of $e^{-2\pi s}$ satisfying a growth condition that fixes the answer as above. The growth in \eqref{aps} as $s \rightarrow -\infty$ can be interpreted as resulting from similar mathematics to the chaos bound discussed in \cite{Maldacena:2015waa}. For example in Section~\ref{sec:nullant} we consider a function:
\begin{equation}
g = \left< \psi_s \right| e^{ -  \epsilon P } \left| \psi_s \right>
\end{equation}
which we will show has a magnitude bounded by $1$ for $-1/4\leq {\rm Im}s \leq 1/4$ and is analytic in that strip. These are the same properties as the out of time order (OTO) four point functions used to study the chaos/scrambling phenomenon \cite{Shenker:2013pqa}. In particular we find: 
\begin{equation}
g \approx 1 - \epsilon e^{-2\pi s} (P_\psi - R) + \ldots
\end{equation}
where we are imagining sending $s$ large and negative (but not too large). Generally one might have expected $1- \epsilon e^{- \lambda_L s}$ where then the same arguments as in \cite{Maldacena:2015waa} would have fixed $\lambda_L \leq 2\pi$ the maximal Lyapunov exponent. Here we prove that this bound is actually always saturated and this arrises from a shift in the ``spectral weight'' (the discontinuity across a certain branch cut) of $g$ towards large $s$ as one sends $\epsilon \rightarrow 0$. This same phenomenon happens in the light-cone limit of the (Rindler) thermal  OTO correlator \cite{Hartman:2015lfa} and also in CFTs with a holographic gravitational dual as one sends $\epsilon = G_N \rightarrow 0$ where in these cases the ``spectral weight'' is governed by the double discontinuity defined in \cite{Caron-Huot:2017vep}. We think that this is more than just a mathematical analogy since in holographic theories both effects will be governed by some kind of gravitational time delay.

Now define the following function:
\begin{equation}
S_{s}(b) \equiv S_{\rm rel}(\psi_s  |\Omega ; B_b) \qquad \bar{S}_{s}(b) \equiv S_{\rm rel}(\psi_s |\Omega ; B'_b)
\end{equation}
Given Lemma~\ref{lem:ant} we can apply the results of Lemma~\ref{lemma1} to the state $\left| \psi_s \right>$ since we also know that $S_{s}(c) = S(c) < \infty$ and $\bar{S}_{s}(c) = \bar{S}(c) < \infty$ which means that $S_{s}(b), \bar{S}_{s}(b)$ are finite and continuous functions of $b$ for all $-\infty< b< \infty$.  

For now let us simply use the fact that $P_s$ is finite and \emph{not} the explicit form in \eqref{aps}. Applying the equation in Lemma~\ref{lemma1}
\begin{equation}
\label{dsdsp2}
-\left( S_{s}(b)-S_{s}(c) \right)  +\left( \bar{S}_{s}(b)-\bar{S}_{s}(c) \right) = (b-c) 2\pi P_{s}
\end{equation}
we can use this to construct the relative entropies everywhere:
\begin{align} 
\label{ssnew}
S_{s}(b) &= \begin{cases} S(b) & b\geq c \\
S\left( e^{-2\pi s} (b-c)+c\right) + 2\pi ( P_\psi e^{-2\pi s} - P_{s} )  (b-c) & b \leq c \end{cases} 
\end{align}
and for the complement:
\begin{align} 
\label{ssbnew}
\bar{S}_{s}(b) &= \begin{cases} \bar{S}(b) + 2\pi  (P_{s} - P_\psi) (b-c) & b\geq c \\
\bar{S}\left( e^{-2\pi s}(b-c)+c\right) & b \leq c \end{cases} 
\end{align}
These functions are clearly still continuous. 
It is more convenient to track the derivative of these functions. Since the input functions are monotonic their derivatives exists almost everywhere. 
For now we will take $c$ to be a point where the derivative of the input relative entropies $S(b),\bar{S}(b)$ exist. Depending on the value of $P_s$ the flowed state might then have a discontinuity in the derivative at $b=c$, however we can still consider the half sided derivatives $\partial^\pm$ taking limits from $c\pm a$ as $a \rightarrow 0$. For example:
\begin{equation}
\label{boundm}
e^{2\pi s} \partial^-S_s(c) = \partial S(c) + 2\pi ( P_\psi - P_s e^{2\pi s}) \leq 0
\end{equation}
where the later inequality is simply monotonicity of the flowed relative entropy. For the complement region we have:
\begin{equation}
\label{boundp}
\partial^+ \bar{S}_s(c) = \partial \bar{S}(c) + 2\pi (P_s - P_\psi)
= \partial S(c) + 2 \pi P_s  \geq 0
\end{equation}
where we used $ - \partial S(c) + \partial \bar{S}(c) = 2\pi P_\psi$ which follows from \eqref{dsdsp}.
We thus derive the following bound by extremizing over $s$:
\begin{equation}
  \sup_s 2\pi ( P_\psi - e^{2\pi s} P_{s} )
 \leq - \partial S(c)  \leq
 \inf_s 2\pi P_{s}
 \end{equation}
 This is an interesting formula in itself, only relying on the finiteness of $P_s$. If we plug in the form of $P_s= R+ (P_\psi - R) e^{-2\pi s}$ given in Lemma~\ref{lemma2} we find:
 \begin{equation}
 \label{equal}
 2\pi R \leq - \partial S(c) \leq 2\pi R \, \implies \, 2\pi R = - \partial S(c) \,.
 \end{equation}
We thus have the following corollary to Lemma~\ref{lemma2}: 
\begin{corollary}[to Lemma~\ref{lemma2}]
\label{lemma2c}
For $\psi$ satisfying the assumptions in \eqref{assumpt} the averaged null energy of the flowed state can be written as:
\begin{equation}
2\pi P_{s} \equiv 2\pi \left< \psi_s \right| P \left| \psi_s \right>  =  - \partial S(c) + e^{-2\pi s} \partial \bar{S}(c)
\end{equation} 
almost everywhere for $c \in \mathbb{R}$. In particular the above equation holds when $S(b)$ is differentiable at $b=c$. 
\end{corollary}
\begin{proof}
See above.
\end{proof}
We can now give a more complete description of the relative entropies. The derivatives satisfy:
\begin{align} 
\partial S_{s}(b) &= \begin{cases} \partial S(b) & b \geq c \\
e^{-2\pi s} \partial S\left( e^{-2\pi s} (b-c)+c\right) +(1-e^{-2\pi s}) \partial S(c) & b \leq c \end{cases} 
\end{align}
almost everywhere in $b$ and for $S$ differentiable at $c$. And for the complement:
\begin{align} 
\label{dssbar}
\partial \bar{S}_{s}(b) &= \begin{cases}  \partial\bar{S}(b) - (1- e^{-2\pi s}) \partial \bar{S}(c) & b\geq c \\
e^{-2\pi s} \partial \bar{S}\left( e^{-2\pi s}(b-c)+c\right) & b \leq c \end{cases} 
\end{align}
Note that the above resulting relative entropies are actually still differentiable at $b=c$ as a result of the form in Corollary~\ref{lemma2c}. We give an example plot of the derivative of relative entropy under modular flow in Figure~\ref{fig:rel-flow}

\begin{figure}[h!]
\centering % \begin{center}/\end{center} takes some additional vertical space
\includegraphics[width=.45 \textwidth]{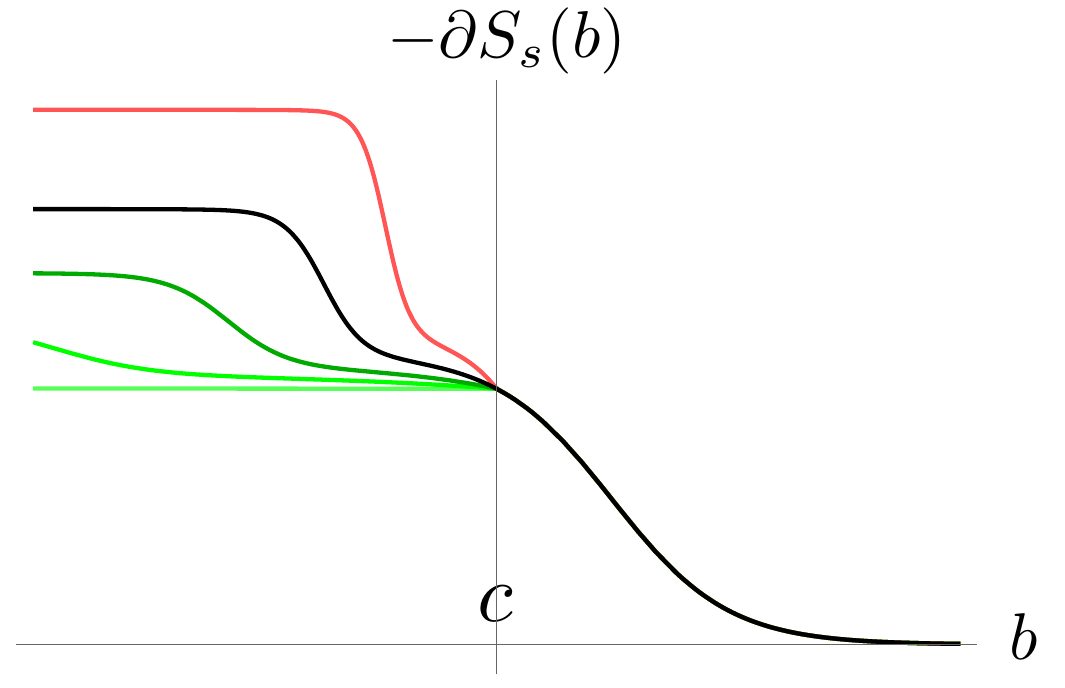} 
\hspace{.5cm}
\includegraphics[width=.45 \textwidth]{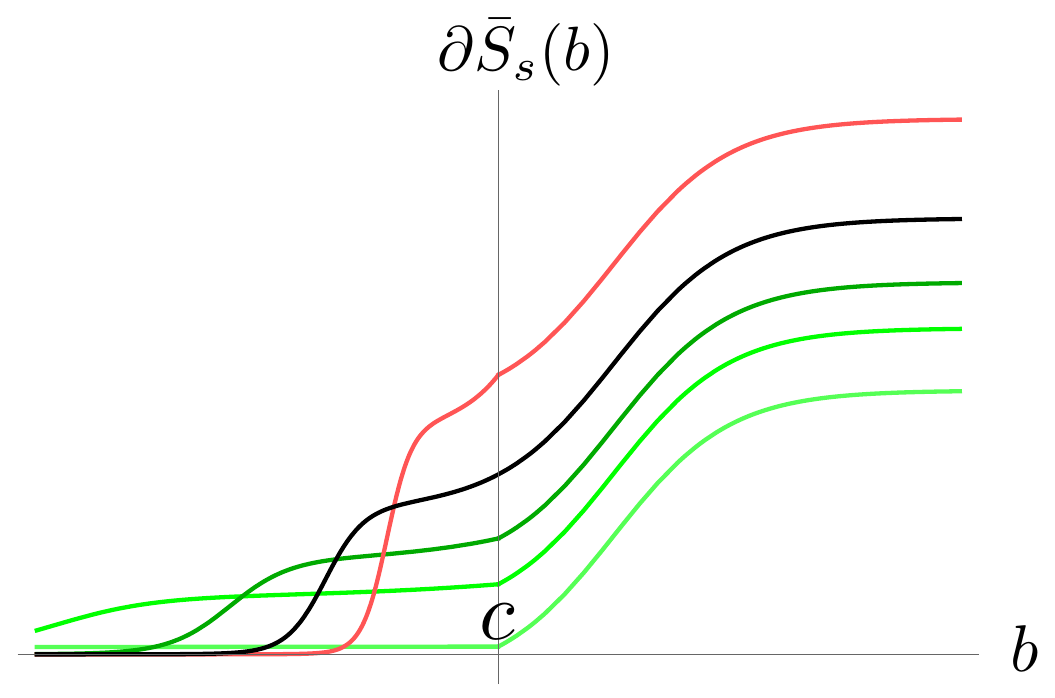}
\caption{ \label{fig:rel-flow} Relative entropy of $\psi_s$ for various $s$ as a function of $b$. The input relative entropy function, shown in black, is a cartoon.  It has the additional property that $\partial S(b), \partial \bar{S}(-b) \rightarrow 0$  as $b \rightarrow \infty$ which is the case if the state approaches vacuum in that limit (this may not actually be the case in QFT since the wiggly cut functions, representing deformations from the Rindler cut $R$, might have bounded support.)
Positive $s$ are the green curves and negative $s$ is the red curve. The relative modular flow is defined with respect to the cut at the point $b=c$.
 } 
\end{figure}

We will give a slightly more refined discussion of these relative entropy functions after we prove the QNEC. For a preview, we will actually find that $S(b)$ is a convex functions which means that the one sided derivatives of these functions exist everywhere and this allows us to constrain $P_{s}$ for all $c$ without the restriction of ``almost everywhere''. 

Another natural class of purifications that the ant might consider are associated to states in the natural self-dual cone $V_{\Omega;C} \subset \mathcal{H}$ for $\Omega$\cite{araki1974some}. 
In general this cone is defined via the closure of:
\begin{equation}
\Delta_{\Omega;C}^{1/4} \mathcal{A}_C^+ \left| \Omega \right>
\end{equation}
where $\mathcal{A}_C^+$ are the positive elements of that algebra.
For a given $\psi$ there is a unique representative $\widehat{\psi} \in V_{\Omega;C}$ that gives the same expectation values of operators in $\mathcal{A}_C$ but generally differs by the action of a unitary from $\mathcal{A}_C'$.
These states have the special property that $J_{\Omega;C} \big| \widehat{\psi} \big> = \big| \widehat{\psi} \big>$. They can be constructed using the ``conjugation cocycles'':
\begin{equation}
\label{def:natcone}
 \big| \widehat{\psi} \big> =( \Theta')^\dagger \left| \psi \right>
 \qquad (\Theta')^\dagger \equiv (\Theta'_{\psi|\Omega})^\dagger = J_{\Omega} J_{\psi|\Omega}  \in \mathcal{A}_C'
\end{equation}
that we review in Appendix~\ref{app:mod}. See in particular \eqref{defconj}. Following the same strategy as above we can compute the relative entropy as follows. Define the following functions:
\begin{equation}
S_{\widehat{\psi}}(b)  = S_{\rm rel}(\widehat{\psi}|\Omega;A_b)\,, \qquad \bar{S}_{\widehat{\psi}}(b)  = S_{\rm rel}(\widehat{\psi}|\Omega;A_b')
\end{equation}
Expectation values of operators in $\mathcal{A}_C$ are unaffected:
\begin{equation}
\left< \psi \right| \Theta' \gamma (\Theta')^\dagger \left| \psi \right>
= \left< \psi \right| \gamma \left| \psi \right> 
\end{equation}
which implies that:
\begin{equation}
S_{\widehat{\psi}}(b) = S(b) \qquad b \geq c
\end{equation}
And for operators in $\mathcal{A}_C'$:
\begin{equation}
\left< \psi \right| \Theta' \gamma' (\Theta')^\dagger \left| \psi \right>
= \left< \psi \right| J_{\psi|\Omega}' \gamma' J_{\Omega|\psi}'  \left| \psi \right> 
=  \left< \psi \right| J_{\Omega} \gamma' J_{\Omega}  \pi'(\psi) \left| \psi \right> 
= \left< \psi \right| J_{\Omega} \gamma' J_{\Omega} \left| \psi \right> 
\end{equation}
So for cuts $\mathcal{A}_{C_{-a}'}$ for $a>0$ the relative entropy is the same as that of the state $ J_{\Omega;C} \left| \psi \right>$. We can then use the following result for the modular operator of such a state:
\begin{equation}
J_\Omega \Delta_{\psi|\Omega;C_a} J_\Omega = \Delta_{ J_\Omega \psi|\Omega;C_{-a}'}
\end{equation}
using similar arguments to those that arrived at \eqref{sflow}. This implies that:
\begin{equation}
\bar{S}_{\widehat{\psi}}(b) =  S_{\rm rel}( J_\Omega \psi|\Omega; A_{b}') = S_{\rm rel}( \psi |\Omega;A_{2c-b})
= S(2c-b) \qquad b\leq c
\end{equation}

For the other relative entropies we need the ANE of this new state. In fact it is not obvious this is finite. However we have the following result:

\begin{lemma}
\label{lemma3}
%In the context of algebras associated to a half sided modular inclusion, see Definition~\ref{} we have the following result. 
Given a state $\psi$ with finite averaged null energy, then the state in the natural self-dual cone $V_{\Omega;C}$ associated to $\Omega$ has finite averaged null energy:
\begin{equation}
\widehat{P} \equiv \big< \widehat{\psi} \big| P \big| \widehat{\psi} \big> \leq 2 R \leq 2 P_\psi 
\end{equation}
where $R$ is the same quantity appearing in Lemma~\ref{lemma2} (not necessarily subject to Corollary~\ref{lemma2c}).
Additionally assuming the state $\psi$ has finite relative entropy for some cut $\mathcal{A}_C$ then the state in the natural self-dual cone for $\Omega$ has finite relative entropy and finite complementary relative entropy:
\begin{equation}
S_{\rm rel}(\widehat{\psi}|\Omega;C') =S_{\rm rel}(\widehat{\psi}|\Omega;C)  < \infty 
\end{equation}
\end{lemma}
Note the later fact about relative entropy follows from our discussion just before the statement of Lemma~\ref{lemma3}.
 We will delay the rest of the proof of this to Section~\ref{sec:nullcone}. This result will be useful for us in the next section since we now need only assume that a particular purification has finite null energy before we can then conclude that there is a state also with finite complementary relative entropy, so we may relax one of the assumptions in \eqref{assumpt}.

We can now, using Lemma~\ref{lemma1} for  $\widehat{\psi}$, give a more complete discussion of the relative entropy for $\widehat{\psi}$ associated to a state $\psi$ with finite ANE:
\begin{align}
S_{\widehat{\psi}}(b) &= \begin{cases} 
S(b) & b \geq c \\
S(2c-b) - 2\pi \widehat{P} (b-c) & b \leq c 
\end{cases} \\
\bar{S}_{\widehat{\psi}}(b) &= \begin{cases} 
S(b)+ 2\pi \widehat{P} (b-c) & b \geq c \\
S(2c-b) & b \leq c 
\end{cases} 
\end{align}
Notice that the relative entropy here only depends on $S(b)$ for $b\geq c$. Indeed none of these manipulations assumed that the complement relative entropy of $\psi$ is finite. See Figure~\ref{fig:rel-cone} for an example of these functions.

\begin{figure}[h!]
\centering % \begin{center}/\end{center} takes some additional vertical space
\includegraphics[width=.45\textwidth]{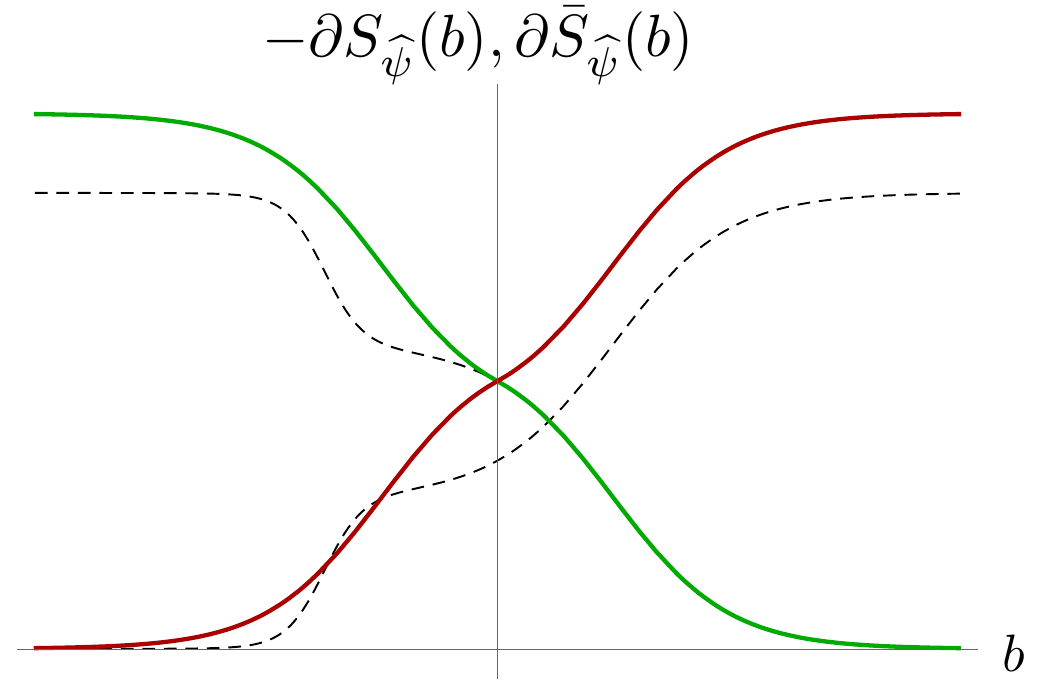} \hspace{.5cm}
\includegraphics[width=.45\textwidth]{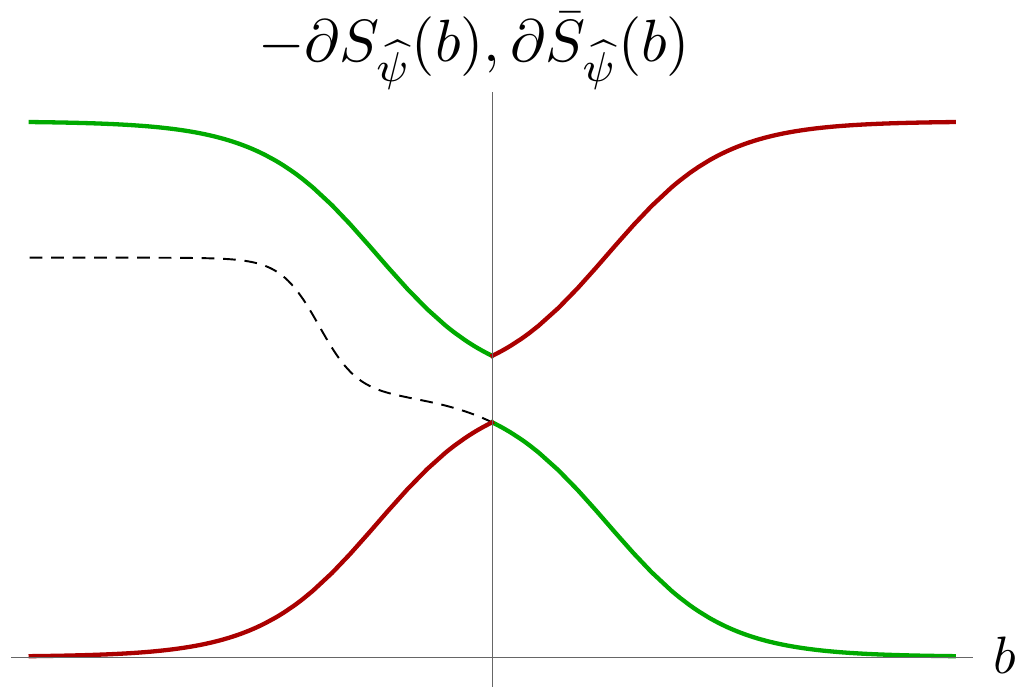}
\caption{ \label{fig:rel-cone} Relative entropy of $\widehat{\psi}$ as a function of $b$, where the natural self-dual cone is with respect to the algebra $\mathcal{A}_{A_c}$ (at the the origin of the $b$-axis). The dashed curves show the input relative entropy and the complement relative entropy (only in the left figure). If the complementary relative entropy for the input state $\psi$ is not finite then we have not been able to discount the possibility of a bounded jump discontinuity in the derivative of relative entropy at $b=c$ which is shown in the right figure. 
 } 
\end{figure}

We turn now to our main results that can be derived from the behavior of the flowed state and the state in the natural self-dual cone. 

\section{Main Results}
\label{sec:main}

Our main goal is to prove the following conjecture by Wall \cite{Wall:2017blw}:
\begin{theorem}[Wall's conjecture]\footnote{Wall wrote down this conjecture in a different form involving entanglement entropy and the half integrated ANE. It is essentially equivalent to what is stated here, although the original form involves quantities that are not obviously well defined in algebraic quantum field theory.  He also for the most part had in mind 2d QFTs. The original conjecture came from arguing that there was no other quantity that he could imagine except $- \partial S$ that satisfies all the properties of the right hand side of \eqref{mc}.}
\label{thm:wall} In the context of algebras satisfying the property of a half-sided modular inclusion $\mathcal{A}_B \subset \mathcal{A}_A$ in Definition~\ref{def1}, consider the relative entropy $S(b)$ of some vector state $\psi \in \mathcal{H}$ compared to the vacuum state $\Omega$, thought of as a function of the null cuts labeled by $b$. Consider states such that:
\begin{equation}
\label{ass22}
\left< \psi \right| P \left| \psi \right> < \infty  \,\, \,\, {\rm and} \,\,\,\, S(b) = S_{\rm rel}(\psi|\Omega;A_b)  < \infty\,,\,\, b \geq b_0
\end{equation}
for some $b_0$.  The derivative of $S(b)$ exists almost everywhere for $c>b_0$ and can be calculated using the following variational expression:
%\begin{equation}
%-\partial_a S(a) = \inf_{U \in \mathcal{A}_{B_{a}}' } \left< \psi \right| U^\dagger P U \left| \psi \right>   \equiv  M(\psi; B_a) 
%\end{equation}
\begin{equation}
-\partial S(c) = M(c) \equiv  \inf_{%\substack{ 
\phi \in  \mathcal{H}: \left\{ \substack{\left< \beta \right>_\phi =   \left< \beta \right>_\psi \, \forall \beta \in \mathcal{A}_{A_c} \\% \qquad
S_{\rm rel}(\phi|\Omega;A_c') < \infty } \right. } 2\pi \left< \phi \right|  P \left| \phi \right>   \,.
\label{mc}
\end{equation}
%whenever the infimum is finite.
%and where the  minimization is over all unitary elements $U \in \mathcal{A}_{B_{a}}'$. 
%(of which $\psi$ need not be one).
%We need the additional technical assumption that for the input vector $\psi$, $S_{\rm rel}$ is finite for at least one of the cuts $B_{b}$.
\end{theorem}
\begin{proof}
From Lemma~\ref{lemma3} we can pass to a state in the natural self-dual cone $V_{\Omega;A_{b_0}}$ at $b=b_0$.
That is consider:
\begin{equation}
\big| \widehat{\psi}\big> = \left(\Theta_{\psi|\Omega;A_{b_0}'}\right)^\dagger \left| \psi \right> \end{equation}
We know $\widehat{\psi}$ has finite ANE, relative entropy and complementary relative entropy at $b=b_0$. We can then apply Lemma~\ref{lemma1} such that the relative entropy functions $S_{\widehat{\psi}}(b)$ and $\bar{S}_{\widehat{\psi}}(b)$ for this new state are finite for all $b \in \mathbb{R}$ and $S_{\widehat{\psi}}(b)$ agrees with the relative entropy of $\psi$ for $b \geq b_0$.  We thus redefine $\widehat{\psi} \rightarrow \psi$ and work with this state. 

For any $c> b_0$ we know there exists at least one state $\phi$ satisfying the assumptions that go into the infimum in \eqref{mc} ($\psi$ itself).
%\begin{equation}
%\label{bs}
%\left< \phi \right| P \left| \phi \right> < \infty \,, \qquad S_{\rm rel}(\phi|\Omega;A_c') < \infty \,, \qquad S_{\rm rel}(\phi|\Omega;A_c) < \infty
%\end{equation}
%By Lemma~\ref{} we know that, 
%Given the assumption that the infimum is finite, there is at least one such state.
%For example for a state with finite ANE we can then pass to the natural self-dual cone where it has both finite relative entropy  (due to \eqref{ass22})and complementary relative entropy.
Such a state $\phi$ also satisfies the assumptions \eqref{assumpt} that go into Lemma~\ref{lemma1}. 
So we have the following estimate: 
\begin{equation}
2\pi \left< \phi \right| P \left| \phi \right> = \lim_{a \rightarrow 0^+} \frac{\bar{S}_\phi(c+a)-\bar{S}_\phi(c) -(S_\phi(c+a) -S_\phi(c))}{a}
\geq - \partial S(c)
\end{equation}
where $S_\phi(b) = S_{\rm rel}(\phi|\Omega;A_b')$ and $\bar{S}_\phi(b) = S_{\rm rel}(\phi|\Omega;A_b)$, the later of which agrees with $S(b)$ for $b \geq c$. We have applied monotonicity to $\bar{S}_\phi(c+a) \geq \bar{S}_\phi(c)$.  Note that $\partial S(c)$ exists since by assumption we are working around a point where the derivative exists, and although the derivative of $S_\phi(b)$ might not exists at $b=c$ its half sided derivative $\partial^+$ does exist and equals $\partial S(c)$.  We have:
\begin{equation}
\label{b1}
M(c) \geq - \partial S(c)
\end{equation}
We next aim to show that the bound in \eqref{b1} is saturated. We use the flowed state discussed in the previous section. 
We apply Lemma~\ref{lemma2} and Corollary~\ref{lemma2c} to:
\begin{equation}
\left| \psi_s \right> = (D \Omega: D \psi\,; A_c' )_s \left| \psi \right>
\end{equation}
where this state also satisfies the properties that go into the infimum of \eqref{mc}, since the action of the co-cycle leaves invariant expectation values of operators in $\mathcal{A}_{A_c}$ \eqref{sameexp}. It also leaves the relative entropies  $A_c$ and $A_c'$ invariant as can be seen from \eqref{ssnew} and \eqref{ssbnew}. The null energy is of course finite and equal to:
\begin{equation}
2\pi \left< \psi_s \right| P \left| \psi_s \right> =  - \partial S(c) + e^{-2\pi s} (2\pi P_\psi + \partial S(c) )
\end{equation}
thanks to Corollary~\ref{lemma2c}. Thus we have the estimate:
\begin{equation}
M(c) \leq \inf_s 2\pi \left< \psi_s \right| P \left| \psi_s \right> = - \partial S(c)
\end{equation}
which finishes the proof.
 \end{proof}

Note that we can slightly loosen the assumptions of this theorem by demanding that for a cut $A_{b_0}$ there is at least one purification with finite ANE. Then we can apply the theorem to compute relative entropies for $c\geq b_0$.
If we were to also demand that the input state $\psi$ has finite complementary relative entropy then considering Lemma~\ref{lemma1} we can now apply this theorem for all $b_0 \rightarrow -\infty$ since all relative entropies are then finite.

The assumption on the ANE is physically sensible for QFT
yet it would still be nice to relax this assumption, for example by showing that the infimum above always exists as long as the relative entropy of one cut $\mathcal{A}_{A_{b_0}}$ is finite. 
It seems reasonable that we should be able to show this by using the state $\widehat{\psi}$ in the natural self-dual cone associated to $\Omega$. However we have so far been unsuccessful here since it is hard to make progress if we don't assume the initial $\psi$ has finite ANE to begin with.

Once we have this theorem it is easy to prove a limited version of the QNEC in the following form:

\begin{theorem}[The Quantum Null Energy Condition (lite) ]
\label{thm:qnec}
For all vector states $\left| \psi \right> \in \mathcal{H}$ with finite averaged null energy and finite relative entropy for cuts $c> b_0$,
%satisfying the assumptions in \eqref{assumpt}
then:
\begin{equation}
\label{tointqnec}
\partial S(c+a) - \partial S(c) \geq 0 
\end{equation}
almost everywhere in $(a,c)$ with $a \geq 0$ and $c>b_0$.
\end{theorem}
\begin{proof} 
This result follows simply because the minimization in \eqref{mc} for $M(c+a)$, compared to that for $M(c)$, is over a super set of states $\phi$ for the same averaged null energy quantity. In particular any state that is included in the minimization for $M(c)$, via Lemma~\ref{lemma1}, has finite relative entropy and complement relative entropy for the algebra $\mathcal{A}_{A_{c+a}}$ as well as finite null energy so it also goes into the minimization for  $M(c+a)$.
\end{proof}

\begin{corollary} 
For all vector states $\left| \psi \right> \in \mathcal{H}$ with finite null energy and relative entropy for $c>b_0$ %satisfying the assumptions in \eqref{assumpt}
the relative entropy $S(b)$ is a convex function - for $(b_1,b_2)> b_0$:
\begin{equation}\label{convex}
S( b_1 t + b_2 (1-t)) \leq t S(b_1) + (1-t) S(b_2) \qquad 0 \leq t \leq 1
\end{equation}
This implies that the half sided derivatives of $S(b)$ exists everywhere and we have a refined estimate to Theorem~\ref{thm:wall}:
\begin{equation}
- \partial^+S(c) \leq  M(c) \leq - \partial^-S(c) 
\end{equation}
where $M(c)$ is defined in \eqref{mc}
\end{corollary}
\begin{proof}
We integrate \eqref{tointqnec}, and since the original function was Lipschitz continuous the fundamental theorem of calculus applies for the Lebesgue integral and we have:
\begin{equation}
S(c +a+a') - S(c+a)  - S(c+a') + S(c) \geq 0
\end{equation}
for all $a,a'\geq 0$ and all $c > b_0$.  Setting $a' = a$ we have:
\begin{equation}
S(c+a) \leq \frac{S( c + 2a) + S(c)}{2}
\end{equation}
which becomes \eqref{convex} for $t=1/2$, $c=b_1$ and $c+2a = b_2$.
That is $S(b)$ is mid-point convex. Continuity plus mid point convexity for all points $b_1,b_2$ implies the more general convexity statement  \eqref{convex} for all $t$. 

To find the improved estimate for $M(c)$ we reconsider our discussion of the ANE of the flowed state - Lemma~\ref{lemma2} and Corollary~\ref{lemma2c}, with the new knowledge that the one sided derivatives always exist. 
Note that the QNEC also applies to $\psi_s$ (the assumptions are satisfied for this state) so we also know the half sided derivatives of $S_s$ exist. 
That is for all $c$ we can replace \eqref{boundm} and \eqref{boundp} with:
\begin{align}
0& \leq - e^{2\pi s} \partial^-S_s(c)=- \partial^-S( c) + 2\pi ( e^{2\pi s}-1)R \mathop{\rightarrow}_{s \rightarrow -\infty}  - \partial^-S( c) - 2\pi R   \\
0 &\leq \partial^+\bar{S}_s(c)  = \partial^+ \bar{S}(c) + 2\pi (R - P_\psi) (1-e^{-2\pi s}) \mathop{\rightarrow}_{s \rightarrow \infty} \partial^+ S(c) + 2\pi R
\end{align}
so we learn that the equality \eqref{equal} is replaced by the inequalities:
\begin{equation}
- \partial^+S(c) \leq  2\pi R \leq - \partial^- S(c)
\end{equation}
We should replace the bound in \eqref{b1} by:
\begin{equation}
M(c) \geq - \partial^+ S(c)
\end{equation}
since we must approach from the side where the relative entropy is unchanged. The flowed state still gives an estimate:
\begin{equation}
M(c) \leq 2\pi R \implies  - \partial^+ S(c) \leq  M(c) \leq  - \partial^- S(c)
\end{equation} 
which is the desired result. The QNEC that follows from this is:
\begin{equation}
M(c+a) \leq M(c) \implies  \partial^+ S(c+a) - \partial^- S(c) \geq 0
\end{equation} 

\end{proof}

\subsection{The QNEC from the ANEC}

There are several alternative routes to the QNEC. For example if we examine monotonicity of relative entropy for $b \geq c$ in \eqref{dssbar} and take the large $s$ limit:
\begin{equation}
0 \leq \partial \bar{S}_s(b) \mathop{\rightarrow}_{s \rightarrow \infty} \partial \bar{S}(b) - \partial \bar{S}(c)  = \partial S(b) - \partial S(c)
\end{equation}
where we are assuming all derivatives exists for simplicity.
Alternatively we can apply Lemma~\ref{lemma2} twice to two different flows. For example consider the state:
\begin{equation}
\big| \psi_{s_{c'}; -s_{b}} \big> \equiv (D \Omega: D \psi\,; A_b )_{-s}  (D \Omega: D \psi\,; A_c )'_s  \left| \psi \right>
\end{equation}
for $b>c$. We can use our results there to compute the ANE of this state:
\begin{equation}
\label{p2s}
2\pi P_{s_{c'}; -s_b} = 2\pi P_\psi e^{-2\pi s}+ (1 - e^{-2\pi s}) (\partial S(b) - \partial S(c) ) 
\end{equation}
which limits to the QNEC at large $s$.

Note that this later result connects with a previous proof of the QNEC. Consider the state:
\begin{equation}
\left| \psi_Q \right> = \Delta_{\Omega;A_{(b+c)/2}}^{is} \big| \psi_{s_{c'}; -s_{b}} \big> 
\end{equation}
where for symmetry we have added an extra boost around the mid point of the two cuts.
In this state, $\psi_Q$, consider a correlation function of two operators from $\mathcal{O}_L \in \mathcal{A}_{A_c}'$ and the other from  $\mathcal{O}_R \in \mathcal{A}_{A_b}$ for $b>c$.
One finds:
\begin{equation}
\left< \psi_Q \right| \mathcal{O}_L \mathcal{O}_R \left| \psi_Q \right>
 = \left< \psi \right| \mathcal{O}_L(s) \Delta_{\psi;A_c}^{is} \Delta_{\psi;A_b}^{-is} \mathcal{O}_R(s) \left| \psi \right> 
\end{equation}
where 
\begin{align}
\mathcal{O}_{L}(s) = V \mathcal {O}_L V^\dagger  \qquad \mathcal{O}_{R}(s) = V^\dagger \mathcal {O}_L V \qquad V = U_{ (b-c) (e^{-2\pi s} -1)/2}
\end{align}
 and we have used the algebra of half sided modular inclusions. 
This correlator was the starting point for the proof of the QNEC in \cite{Balakrishnan:2017bjg}.
We can easily compute the ANE in this state now. The boost simply amplifies the ANE in \eqref{p2s}
\begin{equation}
2\pi \left< \psi_Q \right| P \left| \psi_Q \right> = 2\pi  P_\psi + (e^{2\pi s}-1) (\partial S(b) - \partial S(c) ) 
\end{equation}
which reproduces the large $s$ results in \cite{Balakrishnan:2017bjg}. In other words the results in \cite{Balakrishnan:2017bjg} can simply be interpreted as extracting the ANEC, using the methods of \cite{Hartman:2016lgu}, but in the state $\left| \psi_Q \right>$.
Proving positivity seems to work slightly differently but we now know that it simply follows from the ANEC. 

\subsection{The QNEC and Strong Superadditivity of Relative Entropy}

We now give a more complete discussion of the QNEC for Rindler cuts in Minkowski space.
We consider states $\psi$ that have finite relative entropy for the undeformed cut $R$ and finite null energy for the generator of null translations:
\begin{equation}
S_{\rm rel}(\psi|\Omega; R) < \infty\,, \quad
\left< \psi \right| P_u \left| \psi \right> < \infty
\end{equation}
where $P_u = \int_{v=0} T_{uu}$. We then only ever consider wiggly cuts defined by continuous functions of the coordinates $y$ along the entangling surface that do not diverge.\footnote{For a CFT a more general discussion is possible, but we do not consider this here.} That is $A(y)$ is such that $ \sup_y A(y) <\infty, \inf_y A(y) > -\infty$ and similarly for $B(y)$ etc. 

Since these functions never diverge any null energy that we define with respect to these wiggly cuts  must be finite because:
\begin{align}
&P_{B-A} = \int_{v=0} (B(y) - A(y) ) T_{uu}\,, \qquad B(y) \geq A(y) \\ 
&\implies [P_{B-A},P_u ] = 0 \,, \,\, \& \,\, P_{B-A} \leq  P_u \sup_y(B(y) - A(y) ) 
\end{align}
which implies that:
\begin{equation}
\left< \psi \right| P_{B-A} \left| \psi \right> \leq \left< \psi \right| P_u \left| \psi \right>  \sup_y(B(y) - A(y) ) 
\end{equation}
so all the null energies that we can define are finite. Similarly by monotonicity of relative entropy, any cut that lies entirely inside the Rindler cut $R$: $A(y) \geq 0, B(y) \geq 0$ etc. implies that the relative entropies of these wiggly cuts are also finite. 

%and Lemma~\ref{} we are also guaranteed that any relative entropy or complementary relative entropy associated to these cuts is finite.
%\footnote{Using the Rindler cut we can uniformly translate this cut and maintain finite relative entropies. The bound on the wiggly cut functions then allows us to always move such a Rindler cut so that it contains the wiggly cut or conversely so that the complement Rindler cut contains the complement wiggly cut. Monotonicity guarantees finite relative entropy. }

We can now state the following more general QNEC. For $B(y) \geq  A(y)$ and $\Sigma(y) \geq 0$
\begin{equation}
\label{qnecfull}
\partial_\lambda^+ \left. S_{\rm rel}(\psi|\Omega;B+ \lambda \Sigma) \right|_{\lambda = 0}-\partial_\lambda^- \left. S_{\rm rel}(\psi|\Omega;A + \lambda \Sigma) \right|_{\lambda = 0} \geq 0
\end{equation}
This follows because the $\lambda$ shape variations can be computed using Theorem~\ref{thm:wall} which then involves the same positive null energy operator $P_\Sigma$ in both cases. The minimization is then over a superset for the $B$ cut and the bound in \eqref{qnecfull} follows. 

We can also prove the strong super-additivity of relative entropy that was first discussed in \cite{Casini:2017roe}. The advantage gained here is that we can derive this result without ever mentioning entanglement entropy which is UV divergent. Given two potentially intersecting cuts $A,B$ we can consider the following one parameter family of cuts $0 \leq \lambda \leq 1$:
\begin{equation}
\alpha_\lambda =  A \cup B + \lambda (B-A\cup B)  \qquad \beta_\lambda = A + \lambda (A \cap B - A)
\end{equation}
where $A \cap B(y) = \max\{ A(y),B(y) \}$ and $A \cup B(y) = \min\{ A(y),B(y) \}$.
Now we have:
\begin{equation}
 A\cap B - A = B - A \cup B = \Sigma \geq 0
\end{equation}
and $\beta_\lambda \geq \alpha_\lambda$ 
so we can apply the more general QNEC for  $0 \leq \lambda \leq 1$:
\begin{equation}
\partial_\lambda S_{\rm rel}(\psi|\Omega;  \beta_\lambda) - \partial_\lambda S_{\rm rel}(\psi|\Omega;  \alpha_\lambda) \geq 0
\end{equation} 
almost everywhere as a function of $\lambda$. Integrating this (which is allowed by Lipschitz continuity of the underlying relative entropy) we have:
\begin{equation}
 S_{\rm rel}(\psi|\Omega; A \cup B) + S_{\rm rel}(\psi|\Omega; A \cap B)  \geq S_{\rm rel}(\psi|\Omega; B) + S_{\rm rel}(\psi|\Omega; A) 
\end{equation}
which is the strong superadditivity statement for relative entropy. 

\section{Null Energy of the Flowed State}
\label{sec:nullant}

We aim to prove Lemma~\ref{lemma2} in this section. That is we would like to find the form of the null energy of the modular flowed state. This discussion takes inspiration from some of the general theorems that go into the theory of half sided modular inclusions, see in particular \cite{borchers1995use} and \cite{araki2005extension}.
We split this discussion into two parts. Firstly we consider a special dense set of states for which the translation operator acts analytically everywhere. We first prove Lemma~\ref{lemma2} for this set of states. We then use a continuity argument for a general state. 

\subsection{Entire states}

In this part we will work with the following nice set of vectors. %We will prove many assertions using these states. 
This set of states will be dense in the Hilbert space and a continuity argument will establish the more general assertions. 
They are entire vectors for the modular translation group $U_a$. Entirety is the statement that:
\begin{equation}
U_a \left| \psi_\Lambda \right>
\end{equation}
is a vector that varies analytically with $a$ in the entire complex-$a$ plane. 
It is clear that such vectors are dense since we can construct them from some arbitrary vector $\psi$ via:
\begin{equation}
\left| \psi_\Lambda \right> = \Pi_\Lambda \left| \psi \right>
\end{equation}
where $\Pi_\Lambda$ is a projection operator onto a compact part $0 \leq \lambda \leq \Lambda$ of the spectrum of $P$:
\begin{equation}
\Pi_\Lambda = \int_0^\Lambda  d E_\lambda(P)
\end{equation} 
where $P=\int_0^\infty \lambda d E_\lambda(P)$ is the spectral resolution of the unbounded postive operator $P$. Here $\Lambda$ can be arbitrarily large, but compactness means that these states are entire vectors. Indeed these vectors are of exponential type with:
\begin{equation}
\label{exptype}
|| U_a \left| \psi_\Lambda \right> || <  \exp(2 \Lambda |a| )
\end{equation}
They provide good approximations to $\psi$. Consider a sequence of increasing $\Lambda_n$ which limits to $\infty$. Let $\psi_n \equiv Z_n \Pi_{\Lambda_n} \psi$ where $Z_n =( \left< \psi \right| \Pi_{\Lambda_n} \left| \psi \right>)^{-1/2}$ . Then we have:
\begin{equation}
\lim_{n} ||  \left| \psi \right> - \left| \psi_n \right> || = 0
\end{equation}
Thus the set of states $\psi_\Lambda$ with $\Lambda < \infty$ is dense  in $ \mathcal{H}$.
Note that the null energy itself varies continuously for these states:
\begin{equation}
\label{pcont}
\lim_{n} \left< \psi_n \right| P \left| \psi_n \right>
= \lim_{n} Z_n   \int_0^{\Lambda_n}\lambda \left< \psi \right| dE_\lambda \left| \psi \right> = \left< \psi \right| P \left| \psi \right>
\end{equation}
which converges assuming that the initial $\psi$ has finite null energy. When we study properties of $\psi_n$ we will often drop the subscript on $n$ and simply assume that $\psi$ is entire. We will return to labeling these states correctly when we make the continuity argument.

In this section, to lighten the notation we will often use the following shorthand for the relative modular and associated operators:
\begin{align}
\Delta_{\psi|\Omega} &= \Delta \qquad \Delta_{\psi|\Omega}' = \Delta' \qquad \Delta_{\Omega|\psi} = (\Delta')^{-1} \\  J_{\psi|\Omega} &= J \qquad  J_{\Omega|\psi} = J^\dagger  \qquad
\Theta_{\psi|\Omega}' = \Theta'  \qquad \Theta_{\psi|\Omega} = \Theta 
\end{align}
where the modular conjugation operators $J$ and the corresponding ``conjugation cocycles'' $\Theta$ are discussed and defined in Appendix~\ref{app:mod}. As usual, we have dropped the region label $C=A_c$ since in this section the algebra will either be $\mathcal{A}_C$ or its complement and we denote the later with a prime on the modular operator. 

Let us consider the following ``structure function'':
\begin{equation}
\label{def:struct}
g(s,\eta) \equiv \left< \psi_{s} \right| U_{-a} \left| \psi_{s} \right>
= \left< \psi \right| (\Delta')^{-is} U_{-a}  (\Delta')^{is} \left| \psi \right>
= \left< \psi \right| \Delta^{-is} U_{-a_s} \Delta^{is} \left| \psi \right> \,
\end{equation}
where:
\begin{equation}
a = \epsilon e^{2\pi \eta} \,, \quad a_s = e^{-2\pi s } a = \epsilon e^{2\pi(\eta-s)}
\end{equation}
and we will consider the analytic properties of $g$ as a function of the two complex variables $\eta,s$ for fixed real positive $\epsilon$. The second expression in \eqref{def:struct} used the algebra of half-sided modular inclusions \eqref{hsmi:algebra}.
 %We will sometimes re-write $a = \epsilon e^{ 2\pi \eta}$ with $\epsilon \geq 0$ and $\eta$ complex with 
The eventual goal will be to send $\epsilon \rightarrow 0$.
 In particular if we can show that $P_s < \infty$ then we can extract the null energy via:
 %\footnote{We have made some arbitrary choices for the structure function that may not seem so natural. More natural choices are possible, we wrote this section before we had settled on the exact form of the flowed state and $g(s,\eta)$ reflects an earlier choice of modular flow state. We can still use it to extract the ANE for our new modular flow state.} 
\begin{equation}
P_{s}= \lim_{\epsilon \rightarrow 0}\frac{ g(s, 0)-1 }{ i \epsilon} \,, \qquad s \in \mathbb{R}
\end{equation}
and this is the reason we study this function. 
So far we have defined $g(s,\eta)$ for real $s,a$ but we will now explore analytic continuations of this function. The utility of complexifying $\eta$ will become clear later.

To find the analytic continuation of $g$ we define the following vector valued holomorphic/anti-holomorphic functions via their inner product with a dense set of states in the Hilbert space:
\begin{align}
\left| \Gamma_I(s,\eta) \right> &\,:\, \left( c' \left| \Omega \right>, \left| \Gamma_I(s,\eta) \right> \right)
\equiv \left( \Delta^{ i s^\star} c' \left| \Omega \right>, U_{-a_s} \Delta^{is} \left| \psi \right>  \right)\\
\left| \Gamma_{II}(s^\star,\eta^\star) \right> &\,:\, \left( \left| \Gamma_{II}(s^\star,\eta^\star) \right>, c' \left| \Omega \right>  \right)
\equiv \left( \Delta^{ i s^\star}  \left| \psi \right> , U_{-a_s} \Delta^{is}  c' \left| \Omega \right>\right) \\
\left| \Gamma_{III}(s^\star,\eta^\star) \right> &\,:\, \left( \left| \Gamma_{III}(s^\star,\eta^\star) \right>, c' \left| \Omega \right>  \right)
\equiv \left( (\Delta')^{  i s^\star} S'_{\Omega|\psi} ( c'\left| \Omega \right>), U_{-a} (\Delta')^{is} \left| \Omega \right> \right) \\
\left| \Gamma_{IV}(s,\eta) \right> &\,:\, \left( c' \left| \Omega \right>, \left| \Gamma_{IV}(s,\eta) \right> \right)
\equiv \left( (\Delta')^{  i s^\star} \left| \Omega \right>, U_{-a} (\Delta')^{is}S'_{\Omega|\psi}( c' \left|\Omega \right> ) \right)
\end{align}
where recall that $S_{\Omega|\psi}' = (\Delta')^{-1/2} J$. These vectors have the following properties:

\begin{lemma}
\label{lemma4}
The above vectors, can be extended to well defined vector valued holomorphic/anti-holomorphic functions on the Hilbert space in their respective open convex tube region $(s,\eta) \in T_{I,II,III,IV}$ defined via $-\infty <{\rm Re} (\eta,s) < \infty$ and the imaginary part living in their respective triangles with corners at:
\begin{align}
T_I &\,: \, \left( {\rm Im s} ,{\rm Im \eta} \right) = \{ (0,1/2), (0, 0), (-1/2, 0) \} \\
T_{II} &\,: \, \left( {\rm Im s} ,{\rm Im \eta} \right) = \{ (0,1/2), (0, 0), (1/2, 1/2) \} \\
T_{III} &\,: \, \left( {\rm Im s} ,{\rm Im \eta} \right) = \{ (1/2,1/2), (1/2,0), (0,0) \} \\
T_{IV} &\,: \, \left( {\rm Im s} ,{\rm Im \eta} \right) = \{ (-1/2,1/2), (-1/2, 0), (0, 1/2) \} 
\end{align}
The vectors are weakly continuous on the closure of these triangles and holomorphic/anti-holomorphic along the (one parameter) set of complex strip sub-regions based on each side of the triangle. The vectors are thus strongly continuous on the domain of holomorphy.\footnote{Recall that strong continuity of a vector uses the Hilbert space norm and weak continuity demands that the inner product with any fixed vector is continuous. A weakly continuous vector valued holomorphic function can be shown to be strongly continuous via the Cauchy integral formula.} The norms of the vectors are bounded by $1$ in the closure of their respective domains of holomorphy. See Figure~\ref{fig:triangles}.
\end{lemma}
\begin{figure}[h!]
\centering % \begin{center}/\end{center} takes some additional vertical space
\includegraphics[width=.8\textwidth]{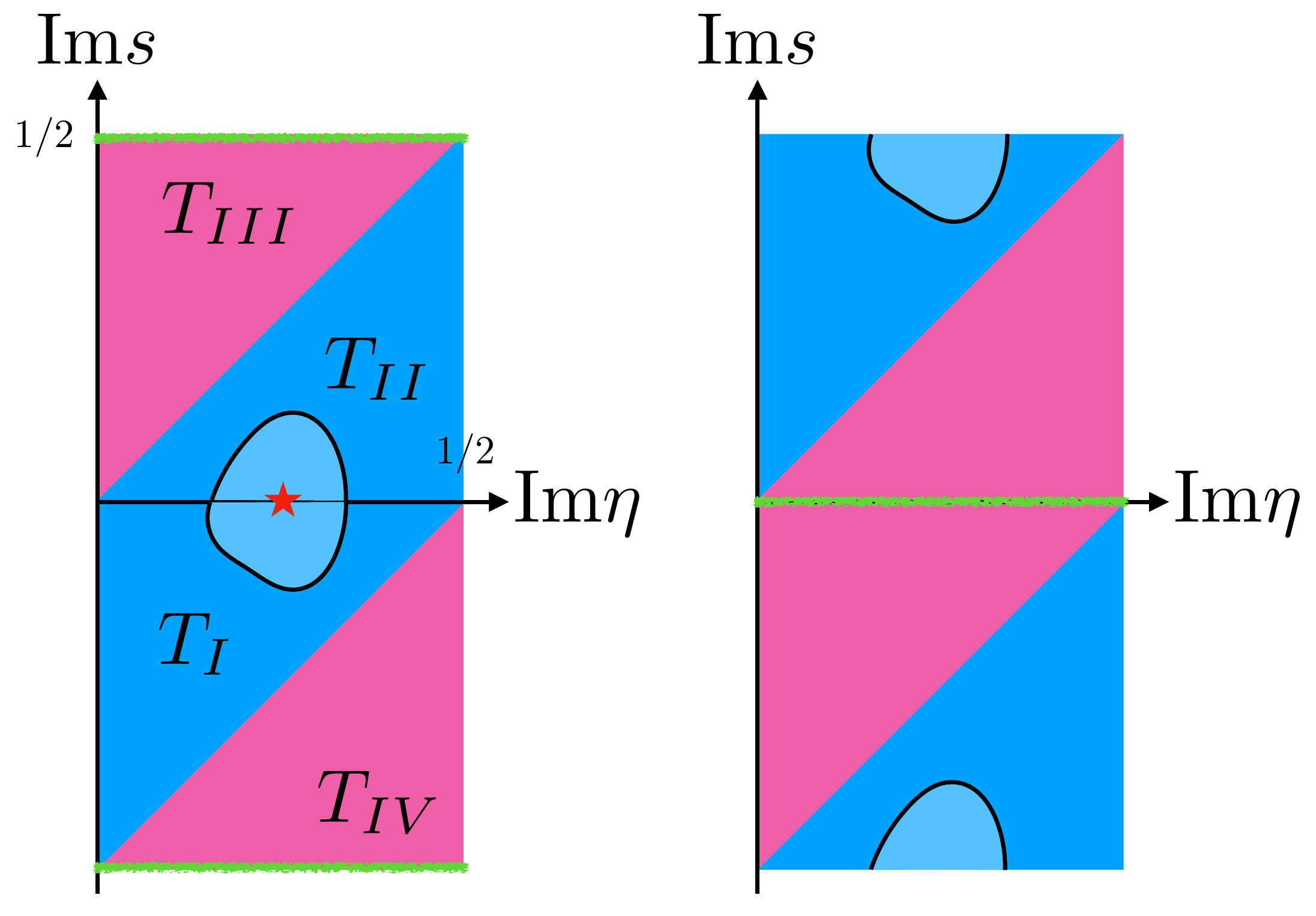}
\caption{ \label{fig:triangles} Domains of holomorphy for the various vectors $\Gamma_{I,II,III,IV}$ discussed in Lemma~\ref{lemma4}. Shown are the regions in the imaginary plane upon which the complex tube regions are based with $-\infty< {\rm Re} (\eta,s) < \infty$. We have also used different colors to show where the resulting function $g_\epsilon(s,\eta)$ defined via Lemma~\ref{lemma5} (see \eqref{gall} and \eqref{geps}) is bounded (blue) or not uniformly bounded (pink). The function is still bounded on compact subsets. The star marks a particularly well behaved point for $g_\epsilon$ since the function is real and monotonic as a function of $\epsilon$. If one identifies the top and bottom via $s \equiv s + i$ then the green lines become branch cuts. The right figure is for $q_\epsilon(\eta,s)$ defined in the proof of Lemma~\ref{lemma7}, see \eqref{qeps}. As $\epsilon \rightarrow 0$ the left and right functions are related in such away that they must become a periodic function under $s \rightarrow s + i$. }
\end{figure}

\begin{proof}
Firstly it is useful to note that the triangular regions are distinguished by the following condition:
\begin{equation}
T_{I,II}: {\rm Im}(a_s) > 0 \qquad T_{III,IV}: {\rm Im} (a_s) < 0
\end{equation}

We now work out explicitly the $\Gamma_{I}$ case and quickly sketch the other cases which follow the same procedure.  We follow the general strategy of Araki \cite{araki1973relative,araki1982positive}, see for example Appendix~A.2 of \cite{Witten:2018zxz}. Consider $(s,\eta) \in T_I$ and define:
\begin{equation}
G_I(s,\eta) = \left( \Delta^{i s^\star} c' \left| \Omega \right>, U_{-a_s} \Delta^{is} \left| \psi \right> \right)
\end{equation} 
This function is holomorphic in $T_I$ because $U_{-a_s}$ is a holomorphic and bounded operator there, $|| U_{ -a_s }|| \leq 1$, and also because the vectors $\Delta^{is^\star} c'  \Omega $/$\Delta^{i s} \psi$  vary anti-holomorphically/ holomorphically in the strip $-1/2 < {\rm Im}(s) <0$ which is a standard result of Tomita-Takesaki theory applied to the relative modular operators.\footnote{This works for a domain twice the size of $T_I$, however the necessary bound below, as far as we are aware, does not extend beyond $T_I$.} Note that $G_I$ is continuous on the closure of $T_I$ due to standard results in Tomita-Takesaki theory and the fact that $U_{-a_s}$ is continuous in the strong operator topology there.
We also have the bound
\begin{equation}
|G_I(t+ i \theta, \eta)| \leq\,\,  || \Delta^{\theta} c' \left| \Omega \right> ||  \,\,  || \Delta^{-\theta} \left| \psi \right> ||   \qquad -1/2 \leq \theta \leq 0
\end{equation}
which is uniform  as a function of $(s,\eta) \in \bar{\Gamma}_I$ (the closure of the trangle) but may not be uniform as a function of the state $ c' \left| \Omega \right>$ (after dividing by $|| c' \left| \Omega \right> ||$.) We will give an improved bound next where this is the case.

Note that $G_I$ is uniformly bounded at the following edge of the triangle:
\begin{equation}
\label{bside}
|G_I(t,\eta)| \leq   || \pi'(\psi) c' \left| \Omega \right> || \leq  || c' \left| \Omega \right> || \qquad t \in \mathbb{R}  \,, \, 
\eta \in \bar{S}(0,1/2)
%0 \leq {\rm Im }(\eta) \leq 1/2
\end{equation}
since  $|| U_{-a_s }|| \leq 1$. We have used the fact that $\Delta^{is-is} = \pi'(\psi)$ and $\pi'(\psi)\leq 1$. We have defined the complex strip:
\begin{equation}
z\in S(a,b): \quad a< {\rm Im} z< b\,, \quad -\infty < {\rm Re} z < \infty
\end{equation}
and where $\bar{S}$ includes the boundaries: $a \leq {\rm Im} z \leq b$.

Along $s= t- i/2$, for real $t$, we can compute:
\begin{align}
G_I(t-i/2,\eta) &= \left( \Delta^{it-1/2} c' \left| \Omega \right>, U_{a_t} \Delta^{it} J_{\Omega|\psi} \pi(\psi) \left| \Omega \right> \right) \\
& = \left( \Delta^{it-1/2} c' \left| \Omega \right>, U_{a_t} \Delta^{it}  \Theta'  \left| \Omega \right> \right) \\& = \left( \Delta^{it-1/2} c' \left| \Omega \right>, U_{a_t} \Delta^{it} \Theta' \Delta_\Omega^{-it} U_{-a_t} \left| \Omega \right> \right) \\ & = \left( \Delta^{it-1/2} c' \left| \Omega \right>, \alpha' \left| \Omega \right> \right)  \\
& = \left( \Delta^{it} c' \left| \Omega \right>, (\Delta_{\Omega|\psi}')^{1/2} \alpha' \left| \Omega \right> \right) \\
& = \left( \Delta^{it} c' \left| \Omega \right>, J' (\alpha')^\dagger \left| \psi \right> \right)  
= \left( (\alpha')^\dagger \left| \psi \right>, J \Delta^{it} c' \left| \Omega \right> \right)\\
&=\left< \psi \right| U_{a_t} \Delta^{it} \Theta' U_{-a} (\Delta_{\Omega}')^{it} (\Delta')^{-it}
J ( c' \left| \Omega \right> ) 
\label{lastmany}
\end{align}
where in the first line we used the fact that the support of $J_{\Omega|\psi}$ is $\pi(\psi)$ so we can just drop the projector. In the third line we inserted a translation and a boost which leave the vacuum invariant. In the fourth line we defined:
\begin{equation}
\alpha' = U_{a_t} \Delta^{it} \Theta' \Delta_{\Omega}^{-it} U_{-a_t} \in \mathcal{A}'_{C_{-a_t}} \subset \mathcal{A}_C'
\end{equation}
and we used the fact that $\alpha' \in \mathcal{A}_C'$ so we could pass the $\Delta^{-1/2}$ to the right as we did in line 5 above. In line 6 we used $J' = J^\dagger$ (see \eqref{compmod}) along with the definition of the hermitian conjugate of an anti-linear operator and in the last line we used \eqref{jimport}.

%we can deriive the following expression for $G_I$ taking $a > 0$ such that the inclusion $\mathcal{A}_{C_{-a_t}}' \subset \mathcal{A}_C'$ holds:
%\begin{equation}
%G_I(t-i/2,\eta) =\left< \psi \right| U_{a_t} \Delta^{it} \Theta' U_{-a} (\Delta_{\Omega}')^{it} (\Delta')^{-it}
%(\Delta')^{it} \Delta_{\psi}^{it}
%J ( c' \left| \Omega \right> ) 
%\end{equation}
%{\bf give details.}

The result is again bounded since these operators are partial isometries or anti-linear equivalents giving:
\begin{equation}
\label{aside}
|G_I(t-i/2,\eta)| \leq   || c' \left| \Omega \right> || \qquad t \in \mathbb{R} \,, \, a >0
\end{equation}
where we have droped various occurrences of support projectors using $\pi(\psi),\pi(\psi') \leq 1$.
The bound in \eqref{aside} applies to one of the corners of the triangle, and \eqref{bside} applies along the \emph{opposite} edge of the triangle. We can thus consider the following two real parameter set of complex strips which sit between one corner of the triangle and the opposite edge:
\begin{equation}
F^{(h,\kappa)}(s) \equiv G_I(s, h+\kappa(s+i/2) )
%a e^{2\pi (s-i/2)\kappa} ) 
\qquad s\in S(-1/2,0) \,,\,\, h,\kappa \in \mathbb{R}
\end{equation}
 where the resulting function of $s$ is holomorphic. We consider $0 \leq \kappa \leq 1$ and in particular $\kappa$ labels the point of intersection with the top edge of the triangle: at ${\rm Im} s  =0$ we have ${\rm Im} \eta = \kappa/2$.
% functions holomorphic and bounded uniformly (in the complex variable $s$ with $a>0$ fixed) in the strip:
%\begin{equation}
%g_{IA}(s) \equiv G_I(s,\eta)\,, \qquad g_{IB}(s) \equiv G_I(s,- a e^{2\pi s} ) \qquad s\in S(0,1/2) \, , \, a>0
%\end{equation}
The function $F^{(h,\kappa)}(s)$, for fixed $(h,\kappa)$, is bounded by $\leq   || c' \left| \Omega \right> ||$ at the edges of the strip. Then, via the Phragm\'{e}n-Lindel\"{o}f principle, this later bound applies inside the $s$-strip. %These represent the other sides of the triangle to \eqref{bside} such that we have establihsed $G_I \leq   || c \left| \Omega \right> ||$ on $\partial T_I$.  Then, via the multi-dimensional generalization of the Phragm\'{e}n-Lindel\"{o}f \cite{}, the bound can be extended to the complex region based on the interior of the triangle, which is the desired domain $ T_I$.  
This applies for all values of $h \in \mathbb{R},\, 0\leq \kappa \leq 1$ and so it applies everywhere in the tube region $\Gamma_I$ including at the boundaries. See Figure~\ref{fig:substrips}.

\begin{figure}[h!]
\centering % \begin{center}/\end{center} takes some additional vertical space
\includegraphics[width=.4\textwidth]{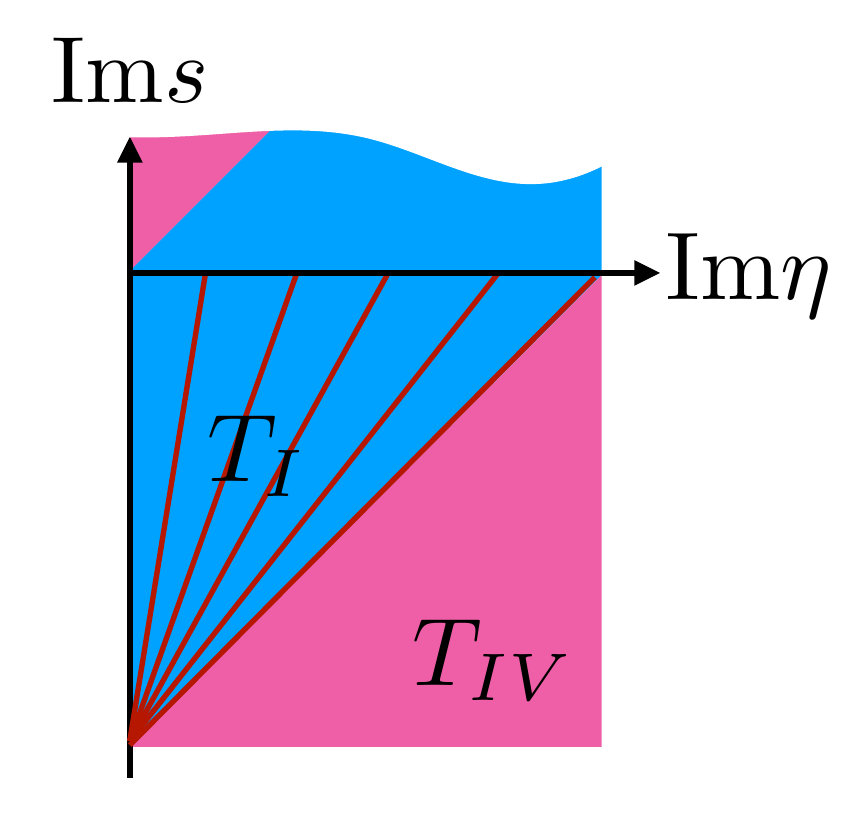}
\caption{ \label{fig:substrips} Substrips of $T_I$ parameterized by $\kappa$ and $h$ and defined via $\eta = h+\kappa(s+i/2)$ with $s\in S(-1/2,0)$. The function $F^{(h,\kappa)}(s)$ on this strip is uniformly bounded by $|| c'\left| \Omega \right> ||$ on the edges so it is uniformly bounded in the bulk.}
\end{figure}

Thus we may interpret $G_I(s,\eta)$ as a bounded anti-linear functional on a dense set of states in the Hilbert space. Such a functional can be extended to the full Hilbert space as follows. Consider a  sequence of states that converges $c_n' \left| \Omega \right> \rightarrow \left| \phi \right>$ in the Hilbert space norm. Then we have a convergent Cauchy sequence of functions: $G_{I}^{(c_n' \Omega)}(s,\eta)$ 
\begin{equation}
\label{csnm}
| G_{I}^{(c_n' \Omega)}(s,\eta) -  G_{I}^{(c_m' \Omega)}(s,\eta) | \leq || c_n' \left| \Omega\right> - c_m' \left| \Omega \right> || \rightarrow 0
\end{equation}
that then must converge to a unique $G_{I}^{(\phi)}(s,\eta)$ which is now an anti-linear functional on the entire Hilbert space. This then defines a vector  $\left|\Gamma_I(s,\eta) \right>$ in the Hilbert space
with the properties that it is norm bounded by $1$ and
\begin{equation}
G_I^{(c' \Omega)}(s,\eta) = \left( c' \left| \Omega \right>, \left|\Gamma_I(s,\eta) \right> \right)
\end{equation}
for all $c' \in \mathcal{A}_C'$. This vector is (weakly) continuous and holomorphic in the same regions as the bounded functions that define it -  since the sequence in \eqref{csnm} is uniformly convergent in $s,\eta$. 
Explicitly this means that %$\left( \left| \phi \right> , 
$\left|\Gamma_I(s,\eta) \right>$ 
is holomorphic (and thus strongly continuous) inside the two complex dimensional domain $T_I$, holomorphic inside the appropriate set of one complex dimensional strips that lie along the edges of $T_I$ and weakly continuous everywhere in the closure of the tube region $\bar{T}_I$.

A similar analysis works for $\Gamma_{II,III,IV}$, and we simply summarize some of the intermediate steps. We define:
\begin{equation}
G_{II}(s,\eta) =  \left( \Delta^{ i s^\star}  \left| \psi \right> , U_{-a_s} \Delta^{is}  c' \left| \Omega \right>\right) 
\end{equation}
which is bounded and holomorphic in $T_{II}$. The uniform bound by $ ||c' \left| \Omega \right> ||$ can be found by examining one of the corners of the triangle  as well as the opposite side of the triangle:
\begin{align}
G_{II}(t+i/2,\eta) &= \left(J( c' \left| \Omega \right>), (\Delta')^{it}\Delta_{\Omega}^{it} U_{-a} (\Theta')^\dagger \Delta^{-it} U_{a_t} \left| \psi \right> \right)  \,, \quad a < 0 \\
G_{II}(t,\eta) &=\left< \psi \right| \Delta^{-it} U_{-a_t} \Delta^{it}  c' \left| \Omega \right>\,, %\qquad 0 \leq {\rm Im} \eta \leq 1/2
\quad \eta \in \bar{S}(0,1/2)
\end{align}
which we use to prove uniform bounds in $\bar{T}_{II}$ which allows us to extend $G_{II}$ to the full Hilbert space.

Continuing we have:
\begin{equation}
G_{III}(s,\eta) =  \left( (\Delta')^{  i s^\star} S'_{\Omega|\psi} ( c'\left| \Omega \right>), U_{-a} (\Delta')^{is} \left| \Omega \right> \right)
\end{equation}
which is bounded and holomorphic in $T_{III}$. The uniform bound in the state again comes from examining one of the corners  of the triangle, and the opposite edge:
\begin{align}
\label{G3}
G_{III}(t+i/2,\eta) &=  \left(J (c' \left| \Omega \right>), (\Delta')^{-it} U_{-a} (\Delta')^{it} (\Theta')^\dagger \left| \psi \right> \right) \,,\quad %0\leq {\rm Im} (\eta) \leq 1/2 \\
\quad \eta \in \bar{S}(0,1/2) \\
G_{III}(t,\eta) &=\left< \psi \right| (\Delta')^{-it} U_{-a} (\Delta')^{it}
% (\Delta')^{-it} \Delta_\psi^{-it} 
U_{a_t} c' \left| \Omega \right>\,, \qquad a>0
\label{g32}
\end{align}
Since this works a bit differently to $G_I$ we give the details of the last equality above:
\begin{align}
G_{III}(t,\eta) &= \left< \psi \right| c' (\Delta')^{-it} U_{-a} (\Delta')^{it} \left| \Omega \right> \\
& = \left< \psi \right| c' \left[ (\Delta')^{-it} U_{-a} (\Delta')^{it} (\Delta_{\Omega}')^{-it} U_a (\Delta_{\Omega}')^{it} \right] \left| \Omega \right>
\end{align}
where it is not hard to see that the object in square brackets is an operator in $\mathcal{A}_{C_a} \subset \mathcal{A}_C$ for $a>0$. Thus we can commute $c'$ through to act on $\left| \Omega \right>$ and reproduce \eqref{g32}.

And finally for:
\begin{equation}
G_{IV}(s,\eta) =  \left( (\Delta')^{  i s^\star} \left| \Omega \right>, U_{-a} (\Delta')^{is}S'_{\Omega|\psi}( c' \left|\Omega \right> ) \right)
 \end{equation}
 we need to examine:
\begin{align}
\label{G4}
G_{IV}(t-i/2,\eta) &=  \left(J \left| \psi \right>, (\Delta')^{-it} U_{-a} (\Delta')^{it} J ( c' \left| \Omega \right>)  \right) \,,\quad %0 \leq {\rm Im} (\eta) \leq 1/2 \\
\quad \eta \in \bar{S}(0,1/2) \\
G_{IV}(t,\eta) &=\left< \Omega \right| (c')^\dagger U_{a_t} %\Delta_\psi^{it} (\Delta')^{it}
(\Delta')^{-it} 
 U_{-a} (\Delta')^{it} \left| \psi \right> \,, \qquad a<0
 \label{clearfrom}
\end{align}
Following the same procedure as for $\Gamma_I$ we find the desired holomorphy, bound and continuity properties for all four cases. 
\end{proof}

We now use these vectors to give an analytic continuation of the structure function $g(s,\eta)$. That is:
\begin{lemma}
\label{lemma5}
The following function of two complex variables:
\begin{align}
\label{gall}
g(s,\eta) = \begin{cases}  \left< \psi \right| \left. \Gamma_I(s,\eta) \right> & (s,\eta) \in \overline{T}_I \\
\left< \Gamma_{II}(s^\star,\eta^\star) \right| \left. \psi \right> & (s,\eta) \in \overline{T}_{II} \\
\left< \Gamma_{III}(s^\star,\eta^\star) \right| U_{-a_s} \left| \psi \right> & (s,\eta) \in \overline{T}_{III} \\
  \left< \psi \right| U_{-a_s} \left| \Gamma_{IV}(s,\eta) \right> & (s,\eta) \in \overline{T}_{IV} 
\end{cases}
\end{align}
is holomorphic in the interior of $\mathcal{B} = \overline{T}_I \cup \overline{T}_{II} \cup \overline{T}_{III} \cup \overline{T}_{IV}= \bar{S}(-1/2,1/2) \times \bar{S}(0,1/2)$ and continuous everywhere in $\mathcal{B}$. We can give the following explicit form along complex sub-strips:
\begin{align}
g(s,\eta) = \begin{cases}
\left< \psi \right|U_{a_t} \Delta^{it} \Theta' U_{-a} (\Theta')^\dagger \Delta^{-it} \left| \psi \right>
& s =t - i/2   \\
\left< \psi \right| \Delta^{-it}  U_{-a_t}  \Delta^{it} \left| \psi \right>
& s =t  \\
\left< \psi \right| \Delta^{it} \Theta' U_{-a} (\Theta')^\dagger \Delta^{-it} U_{a_t} \left| \psi \right>
& s =t + i/2 
\end{cases}
\end{align}
where $\eta \in \bar{S}(0,1/2)$ and $t \in \mathbb{R}$ for all of these cases.
This then represents the desired analytic continuation of the original function \eqref{def:struct}.
\end{lemma}
\begin{proof}
In the respective triangular open tube regions we have holomorpy as well as continuity in the closure. This later fact follows since $U_{-a_s} \left| \psi \right>$ (being entire) is a strongly continuous vector and it is easy to show that the inner product of a bounded weakly continuous vector with a strongly continuous vector results in a continuous function. 

So to show that $g(s,\eta)$ is holomorphic inside $\mathcal{B}$ we need to compare, for consistency, the function on the overlapping regions where it has multiple definitions.  
Then we can apply the edge of the wedge theorem which then tells us that we can extend the region of holomorphy across these overlaps.\footnote{The multi-dimensional edge of the wedge theorem is overkill here, we can apply the one dimensional edge of the wedge theorem along $\eta$-strips at fixed $s$ or equivalently $s$-strips at fixed $\eta$. The result is a function $g(s,\eta)$ that is holomorphic in one of the variables when the other one is held fixed, which by Hartog's theorem is holomorphic in both variables. }

The strip $\eta \in \bar{S}(0,1/2)$ with $s= t$ real is obvious. For the complex strips defined by the lines at $45^0$ in Figure~\ref{fig:triangles} we just need to compare the different definitions at the boundaries of this diagonal strip. This is because we know both definitions are holomorphic along this diagonal strip, continuous and bounded in the closure implying both definitions in the bulk of the strip are determined by the boundary values. To see this, an explicit formula can be worked out by considering:
\begin{equation}
\label{gkern}
g(s,h +s) =  - i \int_C dt \frac{ (e^{2 \pi s}+i)}{ (e^{2\pi t}+i)\left(1- e^{2\pi (s-t)}\right)}  g(t,h+t) 
\end{equation}
%\begin{equation}
%g(s,s+i/2+h) = \int_{-\infty}^\infty dt \kappa_+(t,s) g(t,t+i/2+h) - \int_{-\infty}^{\infty} dt  \kappa_-(t,s) g(t-i/2,t+h)
%\end{equation}
where $C$ circles the pole at $t=s \in S(0,1/2)$ in the clockwise direction and for the complex strip under consideration
$h$ is fixed and real. We can deform the contour $C$ to the boundaries of the strip at ${\rm Im } t =0,1/2$ because the kernel decays exponentially at large $t \rightarrow \pm \infty$ and $|g(s,\eta)| \leq 1$ along this strip, so we can drop the ``vertical'' segments of $C$ at $\pm \infty$. We have also used the fact that we can take the limit of the ``horizontal'' segment to the boundary of the strips thanks to continuity of $g$ and the dominated convergence theorem where the dominating integrable function is obtained by replacing $|g(t,\eta)| \rightarrow 1$  and taking the magnitude of the kernel in \eqref{gkern}. 
Note that this kernel in \eqref{gkern} was constructed by mapping the strip in the $t$ plane to the unit disk via $z = 1/(1+ i e^{-2\pi t})$ and applying the Cauchy integral formula there.

So to reiterate we just need compare definitions on the edges of the diagonal strips.
For example, for the lower diagonal line in Fig~\ref{fig:triangles}, starting with the definition in terms of $\Gamma_{IV}$ at $s = -i/2 + t$ and $a > 0$:
\begin{equation}
(IV): g(t-i/2,\eta) = \left( J^\dagger (\Delta')^{-it} U_{-a} (\Delta')^{it}  J (U_{-a_t} \left| \psi \right>), \left| \psi \right> \right)
\end{equation}
where we used \eqref{G4}.
Now using,
\begin{equation}
J^\dagger (\Delta')^{-it} U_{-a} (\Delta')^{it} J = \Delta^{it} \Theta' U_{a} (\Theta')^\dagger \Delta^{-it}
\end{equation} 
which follows from \eqref{jimport}, \eqref{defconj} and \eqref{hsmi:algebra}, we find:
\begin{equation}
(IV): g(t-i/2,\eta)  = \left< \psi \right| U_{a_t} \Delta^{it} \Theta' U_{-a} (\Theta')^\dagger \Delta^{-it} \left| \psi \right>
\label{refup}
\end{equation}
This reproduces the definition with $\Gamma_I$ along this line, using \eqref{lastmany}:
\begin{align}
(I): g(t-i/2,\eta) &= \left< \psi \right| U_{a_t} \Delta^{it} \Theta' U_{-a} (\Delta_\Omega')^{it} (\Delta')^{-it} (\Theta')^\dagger \left| \psi \right> \\ &= \left< \psi \right| U_{a_t} \Delta^{it} \Theta' U_{-a} (\Theta')^\dagger (\Delta_\Omega')^{it} (\Delta')^{-it}  \left| \psi \right>  \\
&= \left< \psi \right| U_{a_t} \Delta^{it} \Theta' U_{-a} (\Theta')^\dagger 
\Delta^{-it} \Delta_\psi^{it} \pi'(\psi) \left| \psi \right> = \eqref{refup}
\end{align}
where in the second line we used the fact that $ (\Delta_\Omega')^{it} (\Delta')^{-it} \in \mathcal{A}_C$
so it commutes with $(\Theta')^\dagger$ and in the last line we used the co-cycle relation \eqref{def:cocycle}.

Agreement for $s = t$ real and $a<0$ is clear from \eqref{clearfrom}. Thus we must have agreement along the strip based on the diagonal line from $({\rm Im} s , {\rm Im} \eta ) = \{ (0,1/2),(-1/2,0) \}$.  Agreement along the upper diagonal line in Figure~\ref{fig:triangles} follows similar reasoning. 

\end{proof}
Notice that the function $g(s,\eta)$ is not uniformly bounded over the full domain where it is defined. 
Instead we have:
\begin{align}
\label{bound12}
|g(s,\eta)| & \leq 1 \qquad &(s,\eta) \in \overline{T}_{I} \cup \overline{T}_{II} \\
\label{bound34}
|g(s,\eta)| & \leq || U_{-a_s} \left| \psi \right> || = \left( \left< \psi \right| \exp( - 2 P {\rm Im}(a_s) ) \left|\psi \right>\right)^{1/2} &(s,\eta) \in \overline{T}_{III} \cup \overline{T}_{IV}
\end{align}
where we have used the fact that these states are entire states in order to compute the norm in the second equation above. If the entire state $\psi$ has bounded spectral support for $P$ up to $\Lambda$ then this later estimate becomes:
\begin{equation}
\label{bound34new}
|g(s,\eta)| \leq e^{- {\rm Im} a_s \Lambda}\, \qquad (s,\eta) \in \overline{T}_{III} \cup \overline{T}_{IV}
\end{equation}

Let us define the following function:
\begin{equation}
\label{geps}
\hat{g}_\epsilon(s,\eta) \equiv \frac{ 1-g(s, \eta)}{ \epsilon} \,, \quad \epsilon > 0 \,, \, (s,\eta) \in \bar{S}(-1/2,1/2) \times \bar{S}(0,1/2)
\end{equation}
where we will still use $a =\epsilon e^{2\pi \eta}$. We are now in a position to show:

\begin{lemma}
\label{lemma6}
The following limit exists:
\begin{equation}
\gamma(s,\eta) = \lim_{\epsilon \rightarrow 0^+} \hat{g}_\epsilon(s,\eta)  < \infty
\end{equation}
uniformly (and holomorphically) for compact subsets of 
\begin{equation}
S(-1/2,1/2) \times S(0,1/2) \ni (s,\eta) 
\end{equation}
where $\gamma(s,\eta)$ is a holomorphic function on this same region.
\end{lemma}
\begin{proof}
We start by examining the complex strip defined by $\eta =  i/4 + h$ for fixed real $h$
and $s \in S(-1/2,1/2)$. We consider a bounded open subregion $C \subset S(-1/2,1/2)$  that intersects the real axis and $s=0$. Then, in light of the bound \eqref{bound12}-\eqref{bound34new}, we can show that:
\begin{align}
\label{ghatestimate}
{\rm Re} \hat{g}_\epsilon(s,\eta) & \geq \min \left\{ \frac{ 1 - || U_{-a_s} \left| \psi \right> || }{\epsilon},0 \right\}  \geq  \frac{ 1 - \sqrt{\left< \psi \right| \exp( 2 P \epsilon m_{C}) \left| \psi \right>} }{\epsilon} \\
&  \geq \frac{ 1  - \exp( \Lambda \epsilon m_{C}) }{\epsilon} \geq  \frac{ 1  - \exp( \Lambda \epsilon_{0} m_C) }{\epsilon_0}  \qquad \epsilon < \epsilon_0
\label{ghatestimate2}
\end{align}
where $m_C =\max\{0, \sup \{  - {\rm Re} \exp(2\pi (h-s)); s \in C \} \}$. It follows from the bound on the real part that
\begin{equation}
y_\epsilon(s) = \exp(- \hat{g}_\epsilon(s,i/4+h) )
\end{equation}
 are all uniformly bounded holomorphic function of $s$ in $C$ and $0< \epsilon < \epsilon_0$. 

Along the real $s$ axis with $s = t$ for $t \in \mathbb{R}$ we can evaluate the function explicility:
\begin{equation}
\hat{g}_\epsilon(t, i/4+ h) = \left< \psi \right| \Delta^{-it} \left( \frac{1 - \exp( -\epsilon e^{2\pi (h-t)} P) }{\epsilon} \right)  \Delta^{it}\left| \psi \right> 
\end{equation}
from which we have the following monotonicity result:
\begin{equation}
\label{monot}
\hat{g}_\epsilon (t, i/4 +h) \geq \hat{g}_{\epsilon'} (t, i/4 +h) \,, \qquad \epsilon < \epsilon'
\end{equation}
implying that $y_{\epsilon}(t) \leq y_\epsilon'(t)$ for $\epsilon < \epsilon'$. Thus for fixed real $s$ the set of functions $y_\epsilon$ labelled by $\epsilon$ are real and positive, monotonically decreasing as a function of $\epsilon \rightarrow 0$, and bounded. This guarantees pointwise convergence along the real $s$-axis (possibly to zero - which would mean that $\hat{g}_\epsilon$ diverges, a possibility we will shortly rule out.)

We now consider a sequence of the functions $y_{\epsilon_n}$ labelled by an arbitrary sequence of real numbers $\epsilon_n > 0$ that converges to zero. If we can show that any such sequence of functions converges uniformly in $s$ to the same function then this function is the uniform limit of $y_\epsilon(s)$ for $\lim_{\epsilon \rightarrow 0}$.

Bounded sequences of holomorphic functions behave well under limits due to Montel's theorem \cite{schiff2013normal} which guarantees the existence of a sub-sequence $\epsilon_{n_m}$ that converges uniformly on compact subsets of $C$ to some holomorphic function.  This in turn can be used to prove the Vitali-Porter theorem \cite{schiff2013normal}, which we will apply repeatedly: \emph{Sequences of holomorphic functions, uniformly bounded in some region $C$, that converge pointwise on a non-discrete subset of $C$ necessarily converge uniformly on compact subset of $C$ to a fixed holomorphic function.} This later holomorphic function is the one appearing in Montel's theorem. 

We know from the discussion around \eqref{monot} that any such sequence $y_{\epsilon_{n}}$ converges along $s=t$ real 
and this is a non-discrete subset of $C$, which then guarantees uniform convergence for compact subsets of $C$ via Vitali's theorem. We can then extend this in the obvious way by again applying Vitali's theorem, to compact subsets of the strip $s \in S(-1/2,1/2)$ all at fixed $\eta = i/4 +h$. 

We must show that the limit function is the same for any two sequences $\epsilon_n$ and $\epsilon'_n$ positive and converging to $0$. We can combine these two sequences to $\epsilon''_k = \epsilon_{k/2}$ for $k$ even and  $\epsilon''_k = \epsilon'_{(k+1)/2}$ for $k$ odd. The combined sequence of functions $y_{\epsilon_k''}$ must converge to some holomorphic function on compacts and any subsequence must converge to the same holomorphic function.
We thus have:
\begin{equation}
\lim_{\epsilon \rightarrow 0^+} y_\epsilon(s) \equiv y_0(s) \big(=\exp( - \gamma(s, i/4+ h)) \big)
\, \qquad s \in S(-1/2,1/2)
\end{equation}
uniformly on compacts. 

We also need that the limiting function $y_0(s)$ is nowhere vanishing in $S(-1/2,1/2)$ (otherwise $\gamma$ would be infinite at that point). Indeed this follows since for a sequence of holomorphic functions $y_{\epsilon_n}(s)$ that are nowhere vanishing and that converge uniformly to a holomorphic function, the limit $y_{0}(s)$ either vanishes everywhere or nowhere.\footnote{The proof of this follows by  considering the topological invariant that counts the zeros enclosed in some subset $\Gamma \subset C$:  $2\pi i N_\epsilon = \int_{\partial \Gamma} ds \left( y_\epsilon'(s)/y_\epsilon(s)\right)$ of a holomorphic function. If the limiting function vanishes at some discrete point inside $\Gamma$ then for small enough $\Gamma$ we have, by uniform convergence applied to the integral, $\lim_{\epsilon \rightarrow 0} N_\epsilon =1$. This contradicts the original assumption which gives $N_\epsilon = 0$ for all $\epsilon$.}
The former possibility is discounted since we know for a state $\psi$ with finite null energy then the limit at $s=0$ exists:
\begin{equation}
\lim_{\epsilon \rightarrow 0^+} \hat{g}_\epsilon(0,i/4+h) =  \lim_{\epsilon \rightarrow 0^+}  \left< \psi \right| \left( \frac{1 - \exp( -\epsilon e^{2\pi h} P) }{\epsilon} \right)  \left| \psi \right>  
= P_\psi e^{2\pi h} <  \infty
\end{equation}

It is now easy to extend the existence of this limit to compact subsets of the two complex dimensional tube region. Consider a bounded open subregion $D \subset S(-1/2,1/2) \times S(0,1/2)$ then we have the same estimate as in \eqref{ghatestimate} and \eqref{ghatestimate2}  with $m_C \rightarrow m_D$ and 
\begin{equation}
m_D =\max\{0, \sup \{  -{\rm Im} \exp(2\pi (\eta-s)); (s,\eta) \in D \} \}
\end{equation}
So we can continue to apply Vitali's theorem. 

Consider the $\eta$-strips with $s$ fixed and $\eta \in S(0,1/2)$. 
Above we have shown pointwise convergence of $\hat{g}_\epsilon(s,\eta)$ along the line $\eta = i/4+h$ in this new complex strip. Again by Vitali's theorem, since $ e^{- \hat{g}_\epsilon}$ is uniformly bounded on $D$  (for $\epsilon < \epsilon_0$) we must have uniform convergence on compact subsets of $D$ intersected with the $\eta$-strip. \footnote{We should again take an arbitrary sequence $\epsilon_n$ converging to zero and show that for any such sequence we have convergence to the same limit function. This follows the same logic as above so we do not repeat it. }
This extends to the full $\eta$ strip at fixed $s$. Criss crossing the two dimensional complex domain $S(-1/2,1/2) \times S(0,1/2)$ in this way and by using either complex strips in the $s$ plane or the $\eta$ plane we arrive at uniform convergence on compacts to a limit function $e^{- \gamma(s,\eta)}$ that is nowhere vanishing and holomorphic in each variable $s$, $\eta$ separately. Hartog's Theorem then guarantees this is a holomorphic function of both variables.
This proves the assertion.

\end{proof}

\begin{corollary} \label{corr3} The ANE of the flowed state, $P_s$,  is finite and can be extracted from the limit function which has the following properties: 
\begin{equation}
\gamma(s,\eta) = e^{2\pi \eta} \rho(s)\,,
\end{equation}
where $\rho(s)$ is analytic in the strip $s\in S(-1/2,1/2)$ and where:
\begin{equation}
P_{s} =i \rho(s) \,, \quad s\in \mathbb{R}
\end{equation}
It also satisfies the bound:
\begin{equation}
\label{limbound}
{\rm Re} \gamma(s,\eta) \geq \begin{cases} 0\,, & (\eta,s)\in \overline{T}_{I} \cup \overline{T}_{II} \\
{\rm Im}(e^{2\pi (\eta-s)}) P_\psi\,, &  (\eta,s) \in \overline{T}_{III} \cup \overline{T}_{IV}
\end{cases}
\end{equation}

\end{corollary}
\begin{proof}
Since we know from Lemma~\ref{lemma6} that the limit $\epsilon \rightarrow 0$ is finite at $\eta= i/4+h$ for real $t$ this already tells us that the ANE of the flowed state is finite:
\begin{align}
\lim_{\epsilon \rightarrow 0} \hat{g}_\epsilon(t,i/4+h) &= \lim_{\epsilon \rightarrow 0^+} \int_0^\infty \frac{1 - e^{- \epsilon e^{2\pi h} \lambda}}{\epsilon} d \left< \psi_{t} \right| E_\lambda(P)\left|  \psi_{t} \right>  \\ &=   e^{2\pi h} P_{t}  < \infty
\end{align}
by the monotone convergence theorem for the Lebesgue integral. Knowing this is finite we can then show that the limit exists for all complex $0 \leq {\rm Im} \eta \leq 1/2$ (inclusive!) and real $t$ by applying the estimate  $\left| \frac{1-e^{i \zeta  \lambda}}{i \zeta}\right| \leq  \lambda$ for $0 \leq \arg \zeta\leq  \pi$\footnote{To show this set $\lambda \zeta =t $ and consider $\left|  (1-e^{i \zeta  \lambda})/(i\zeta \lambda)\right| =  \left| (1-e^{i   t})/t \right| $ which we have to show is bounded by $1$. Now for real $t$ we have: $\left| (1-e^{it})/t\right| = 2 |\sin(t/2)/t| \leq 1$. This later inequality extends, via the  Phragmén–Lindelöf principle, into the $t$-upper half plane since $(1-e^{it})/it$ is holomorphic and bounded in the upper half plane (UHP).
%So we have $\left| 1 + (1-e^{i \zeta  \lambda})/(i\zeta \lambda)\right| \leq 2$ for $\zeta$ in the UHP as required. 
 There is probably a much easier way to show this.} in conjunction with the dominated convergence theorem: 
\begin{align}
\lim_{\epsilon \rightarrow 0} \hat{g}_\epsilon(t,\eta) &= \lim_{\epsilon \rightarrow 0^+} \int_0^\infty \frac{1 - e^{i \epsilon e^{2\pi \eta} \lambda}}{\epsilon} d \left< \psi_{t} \right| E_\lambda(P)\left|  \psi_{t} \right>  \\ &= - i e^{2\pi \eta} P_{t}  \qquad \eta \in \bar{S}(-1/2,0)
\end{align}
Now we move to the full two dimensional complex plane. Define $\rho(s) = \gamma(s,0)$ then by Lemma~\ref{lemma6} we know that $\rho(s)$ is holomorphic in the strip $s\in S(-1/2,1/2)$.  For real $s=t$ we have $\rho(t) = - i P_{t}$ from the above considerations. So the following function
\begin{equation}
\gamma(s,\eta) - e^{2\pi \eta} \rho(s)
\end{equation}
is holomorphic in the two dimensional complex tube region $S(-1/2,1/2) \times S(-1/2,0)$ vanishing along the three dimensional plane ${\rm Im s} =0$. One can then show that this function must vanish in $S(-1/2,1/2) \times S(-1/2,0)$ - by considering complex sub-strips at fixed $\eta$ in the $s$-plane where we have a holomorphic function in $s$ vanishing along the real axis, which by the identity theorem must then vanish in the entire strip. Thus:
\begin{equation}
\gamma(s,\eta) = e^{2\pi \eta} \rho(s) \qquad (s,\eta) \in S(-1/2,1/2) \times S(-1/2,0)
\end{equation}

For the bound in \eqref{limbound} we can simply take the limit $\epsilon \rightarrow 0$ on the bound in \eqref{bound12} and \eqref{bound34}.

\end{proof}

Our next task is to show that $\gamma$ and hence $\rho$ are entire functions of $ z = e^{-2\pi s}$. We do this by showing that this function is periodic in $s \rightarrow s + 2\pi i$ and also by checking its behavior as $s \rightarrow \infty$. 

\begin{lemma}
\label{lemma7}
The limit function $\gamma(s,\eta)$ is a periodic function in the $s$-strip $s\in S(-1/2,1/2)$ for fixed $\eta$. Together with the bound \eqref{limbound} it follows that $\rho$ is an entire function of $ e^{-2\pi s}$ and that:
\begin{equation}
\label{known}
\rho(s) = -i( R+ Q  e^{- 2\pi s }  )
\end{equation}
where $R,Q$ are real constants satisfying $R+Q = P_\psi$ and $0 \leq R,Q \leq P_\psi$.
\end{lemma}
\begin{proof}
To show periodicity we define a new function:
\begin{equation}
\label{qeps}
\hat{q}_\epsilon(s,\eta) =\frac{1}{\epsilon}  \begin{cases} 
\left( 1 - \left< \Gamma_{III}(s^\star,\eta^\star) \right| \left. \psi \right> \right) & (s,\eta) \in \overline{T}_{III} \\ 
\left( 1 - \left< \psi \right| \left. \Gamma_{IV}(s,\eta) \right>\right) & (s,\eta) \in \overline{T}_{IV} 
\end{cases}
\end{equation}
This function is continuous under the identification of the top and bottom lines in Figure~\ref{fig:triangles} - see the right figure. That is if we identify $s = i/2 + t$ with $s= - i/2+t$ for fixed $-\infty< t < \infty$ and $0 \leq {\rm Im} \eta \leq 1/2$. This can be easily seen from \eqref{G3} and \eqref{G4} after replacing $c' \left| \Omega \right> \rightarrow \left| \psi \right>$.
By the edge of the wedge theorem $\hat{q}_\epsilon$ is holomorphic across this identified line. We can also extend the definition of $\hat{q}_\epsilon$ into $T_{I}$ and $T_{II}$ but we do not need this here. 

Now consider, for $(s,\eta) \in T_{III}$:
\begin{align}
\left| \hat{q}_\epsilon - \hat{g}_\epsilon  + \frac{(1- \left< \psi \right| U_{-a_s} \left| \psi \right> )}{\epsilon} \right| 
&= \frac{1}{\epsilon} \left|\left( \left| \psi \right> -  \left| \Gamma_{III} \right>, \left| \psi \right> - U_{-a_s} \left| \psi \right> \right)\right| \\
& \leq || \left| \psi \right> - \left| \Gamma_{III} \right>|| \left( \frac{|| \left| \psi \right> - U_{-a_s} \left| \psi \right>  ||}{\epsilon} \right) \\
& \leq \sqrt{2 \epsilon {\rm Re} \hat{q}_\epsilon } \left( \frac{|| \left| \psi \right> - U_{-a_s} \left| \psi \right>  ||}{\epsilon} \right)
\label{lastb}
\end{align}
where in the last inequality we used the fact that $|| \left| \Gamma_{III} \right> || \leq 1$ for $(s,\eta) \in T_{III}$. 
%We know that $\lim_{\epsilon \rightarrow 0^+} \left| \Gamma_{III} \right> \rightarrow \left| \psi \right>$\footnote{This follows from what we have already shown about the limit of $g(s,\eta)$ in this triangle and the fact that $\lim_{\epsilon \rightarrow 0} U_{-a_s} \left| \psi \right>= \left| \psi \right>$. {\bf Fill in details} }  
We can compute:
\begin{align}
 \frac{|| \left| \psi \right> - U_{-a_s} \left| \psi \right>  ||^2}{\epsilon^2}
 &= \frac{ 1 + \left< \psi \right| e^{-2 P {\rm Im} a_s } \left| \psi \right> - 2 {\rm Re} \left< \psi \right| e^{i a_s P} \left| \psi \right> }{\epsilon^2} \\ &\rightarrow  \left< \psi \right| P^2 \left| \psi \right> | e^{2\pi (\eta-s)}|^2
\end{align}
where in the last equation we have taken the limit $\epsilon \rightarrow 0^+$ which is indeed finite because $\psi$ is entire and so has finite fluctuations for $P$. We conclude from the bound in \eqref{lastb} that the only possible behavior for $\lim_{\epsilon \rightarrow 0^+} {\rm Re} \hat{q}_\epsilon < \infty$ (if it diverged then the left hand side would diverge quicker and violate the bound.) We have used the fact that all other quantities except $\sqrt{\epsilon}$ in \eqref{lastb} are finite in the limit $\epsilon \rightarrow 0^+$.

So the right hand side of \eqref{lastb} vanishes and we can thus compute the limit of this new function:
\begin{equation}
\label{rho2}
\lim_{\epsilon \rightarrow 0^+} \hat{q}_\epsilon(s,\eta) = e^{2\pi \eta} \left(\rho(s) + i \left< P \right>_\psi e^{-2\pi s}\right)
\end{equation}
uniformly for compact sub-regions of $\overline{T}_{III}$ and 
where a similar analysis in $\overline{T}_{IV}$ yields the same limit. Since we know that $\hat{q}_\epsilon$ is holomorphic across the identification $t + i/2 \equiv t - i/2$, it must be true that the limit is also holomorphic across this line (we have not shown uniform convergence of the limit everywhere across this line, but one can consider $e^{-\hat{q}_\epsilon}$ which is bounded on compacts and apply Vitali's theorem to guarantee the existence of this limit for compact sets crossing this line.) Our conclusion from \eqref{rho2} is that $\rho(s)$ can be extended to a function that shares this same holomorphy across the identification $t + i/2 \equiv t - i/2$. Since $\rho(s)$ is also holomorphic in the original $s$-strip $-1/2 < {\rm Im} s < 1/2$ we  get the desired periodicity. 

In light of the bound on $\gamma$ in \eqref{limbound} we write:
\begin{equation}
\gamma_\eta(z) \equiv \gamma(s,\eta) = e^{2\pi \eta} \rho(s)\,, \qquad z = e^{2\pi (\eta-s)}
\end{equation}
The bound translates to:
\begin{align}
\left|\exp(-\gamma_\eta(z)) \right| & \leq  \begin{cases} 1\,, & {\rm Im} z \geq 0 \\  \exp(- {\rm Im}(z) \left< P \right>_\psi ) \leq \exp(|z|\left< P \right>_\psi )  \, & {\rm Im} z \leq 0 
\end{cases}  \\
& \leq  \exp(|z|\left< P \right>_\psi )  
\label{entbound}
%\leq % \exp( {\rm Im} z \left< P \right>_\psi ) \leq \exp( |z| \left< P \right>_\psi )
\end{align}
Periodicity under $s \rightarrow s + i$ tells us that this is an analytic function everywhere in the complex $z$ plane except maybe at $z=0$. However the above bound in a neighborhood of $z=0$ means that this is at worst a ``removable singularity'', so we can extend $\gamma_\eta(z)$ to an analytic function there. We don't actually need $\gamma_\eta(0)$ in the sequel and here it is sufficient to know that $\gamma_\eta$ can be extended to an analytic function.

Furthermore the bound \eqref{entbound} tells us that $e^{-\gamma_\eta}$ is an entire function of finite order\footnote{The order of an entire function  $f(z)$ is $\limsup_{r \rightarrow \infty} (\log \log {\max_{|z|=r} |f(z)|)/\log r}$. } at most $1$ and which has no zeros. By the Weierstrass/Hadamard Factorization theorem\footnote{This is the well known fact that an entire function can be determined by its zeros  up to an overall exponential of a polynomial who's degree is the same as the order of the entire function. } applied to $e^{-\gamma_\eta}$ we see that $\gamma_\eta(z)$ is a degree $1$ polynomial in $z$ at fixed $\eta$.
The $\eta$ dependence is also fixed by the form:
\begin{equation}
\gamma_\eta(z) = e^{2\pi \eta} \rho(s) =- i( Q  z + R e^{2\pi \eta}) = -  ie^{2\pi \eta}(R +  Q e^{-2\pi s})
\end{equation}
for unknown $R,Q$. At $s=0$ we have $R+Q = P_\psi$. The bound in \eqref{limbound} applied in the upper and lower $z$ planes then gives $0\leq Q \leq P_\psi$. 
\end{proof}

In terms of the ANE of the flowed state we have:
\begin{equation}
P_s = R + (P_\psi -R) e^{-2\pi s}
\end{equation}
which proves Lemma~\ref{lemma2} for entire states. 

\subsection{General states}

We now aim to finish the proof of Lemma~\ref{lemma2} which pertains to general states (with finite ANE) in the Hilbert space. 

\begin{proof}[Proof of Lemma~\ref{lemma2}]
For a general state $\psi$ we can approximate it by a sequence of entire states:
\begin{equation}
\label{defpsin}
\left| \psi_n \right> =Z_n\Pi_{\Lambda_n} \left| \psi \right>\,,\quad  Z_n =\left(\left< \psi \right| \Pi_{\Lambda_n} \left| \psi \right> \right)^{-1/2} 
\end{equation}
Consider the flowed null energy:
\begin{align}
P_s^{(n)} &\equiv   \big< \psi_s^{(n)} \big| P \big| \psi_s^{(n)} \big> \,, \quad  \big| \psi_s^{(n)} \big> = u_s'^{(n)} \left| \psi_n \right> \,, \,\, s\in \mathbb{R} \\
& u_s'^{(n)} = (D \Omega:D\psi_n )'_s = (\Delta_{\Omega}')^{is} (\Delta_{\Omega|\psi_n}')^{-is}  \in \mathcal{A}'_C 
\end{align}
we know, that for these states the ANE is finite and takes the form:
\begin{equation}
P_s^{(n)} = R_n + Q_n e^{-2\pi s} \qquad 0 \leq R_n, Q_n \leq P_\psi^{(n)} 
\end{equation}
where  $P_\psi^{(n)}$ is the ANE of the projected state. It satisfies:
\begin{equation}
P_\psi^{(n)} \leq Z_n P_\psi
\end{equation}
and as discussed in \eqref{pcont} limits to $P_\psi$ as $n \rightarrow \infty$. This means that $R_n, Q_n$ are bounded sequences of positive numbers. (Note that $Z_n$ converges to $1$.)

We need continuity of the modular operators and in particular the co-cycle. It was shown by Araki (actually combining two of his results; Lemma~4.1 in \cite{araki1977relative} and Theorem~10 in \cite{araki1974some}) that:
\begin{equation}
\lim_{n\rightarrow \infty}  (\Delta'_{\Omega|\psi_n})^{is} =  (\Delta_{\Omega|\psi}')^{is} 
\end{equation}
in the strong operator topology and uniformly on compact subsets of $\mathbb{R} \ni s$.\footnote{ Actually Araki showed this for states $\widehat{\Omega}, \widehat{\psi},\widehat{\psi}_n$ that are representatives of $\Omega,\psi,\psi_n$ in a natural positive cone associated to $\mathcal{A}_C'$ and some cyclic separating vector.  Taking this vector to be $\Omega$ and using the fact that $\Delta'_{\Omega|\widehat{\psi}_n} = \Delta'_{\Omega|\psi_n}$, where $\widehat{\psi}_n$ is related to $\psi_n$ by some unitary in $\mathcal{A}_C$, we have the desired continuity.}
%More explicitly the states $\widehat{\psi}_n$ are related to $\psi_n$ by partial isometries that live in $W
We thus have the following continuity properties of the cocycle:
\begin{equation}
\lim_n u_s'^{(n)} = \lim_n (\Delta_{\Omega}')^{is} (\Delta_{\Omega|\psi_n}')^{-is} = u_s'
\end{equation}
strongly and uniformly on compacts.

We now show that $P_s^{(n)}$ is lower semi-continuous. 
Consider the spectral representation:
\begin{equation}
P_s^{(n)} = \int_0^\infty \lambda d \big< \psi_s^{(n)} \big| E_\lambda(P) \big| \psi_s^{(n)} \big> 
\geq  \int_0^\Lambda \lambda d \big< \psi_s^{(n)} \big| E_\lambda(P) \big| \psi_s^{(n)} \big> 
\end{equation}
for all $\Lambda > 0$. Since the projected operator $P \Pi_\Lambda$ is bounded we can now take the following limit on $n$ knowing the right hand side converges:
\begin{align}
\liminf_n P_s^{(n)} & \geq \int_0^\Lambda  \lambda d \big< \psi_s \big| E_\lambda(P) \big| \psi_s \big> 
\qquad \forall \Lambda \\ \implies \liminf_n P_s^{(n)}  &\geq \lim_{\Lambda \rightarrow \infty}  \int_0^\Lambda \lambda d \big< \psi_s \big| E_\lambda(P) \big| \psi_s \big>  = P_s
\end{align}
Since we know the left hand side divided by $Z_n$ is bounded for all $n$, and $\lim_n Z_n = 1$, then $\liminf$ must be finite. This implies that $P_s$ is finite. 

We will show that it saturates this bound. We return to our trusty structure function $g(s,\eta)$ from the previous sub-section. This function can still be studied for non entire states, although we have less control over its analytic properties. However here we only need study it for real $s,\eta$ or real
$s, \eta-i/2$.  

By using the algebra of modular inclusions one can derive the following identity for $a>0$:
\begin{align}
g(s,\eta) &= \left< \psi \right| (\Delta')^{-is} U_{-a} (\Delta')^{is} \left| \psi \right>
= \left< \psi \right| (D\psi:D\Omega)'_{s} U_{-a} (D\psi:D\Omega)_{-s} \Delta_{\Omega}^{-is}  \left| \psi \right>  \\
%&= \left< \psi \right| (D\psi:D\Omega)'_{s} U_{-a} (D\psi:D\Omega)_{-s} U_{a} U_{-a} \Delta_{\Omega}^{-is}  \left| \psi \right> \\
\label{commute}
& = \left< \psi \right|  U_{-a} (D\psi:D\Omega)_{-s} U_{a}   (D\psi:D\Omega)_{s}' U_{-a} \Delta_{\Omega}^{-is}  \left| \psi \right> \\
& = \left< \psi \right| U_{-a} X_a U_{-a_{s}} \left| \psi \right>\,, \qquad
X_a = (D\psi:D\Omega)_{-s} U_{a} \Delta^{-is}
\end{align}
where in \eqref{commute} we used the fact that $  U_{-a} (D\psi:D\Omega)_{-s} U_{a}  \in \mathcal{A}_{C_{a}}$ and so it commutes with $ (D\psi:D\Omega)'_{s}$.  We can thus approximate this in the limit $\epsilon \rightarrow 0$ (equivilently $a \rightarrow 0$) using:
\begin{align}
& | (1-g(s,\eta)) - (1- \left< \psi \right| U_{-a} \left| \psi \right> ) -  (1- \left< \psi \right| U_{-a_s} \left| \psi \right> ) -  (1- \left< \psi \right| X_a \left| \psi \right> ) |  \nonumber  \\ \nonumber
& \qquad \leq | \left< \psi \right| (1-U_{-a}) X_a (1-U_{-a_s} ) \left| \psi \right> |
 \\ & \qquad + | \left< \psi \right| (1-U_{-a} ) (1-X_a ) \left| \psi \right> | + | \left< \psi \right| (1-X_a ) (1-U_{-a_s} ) \left|  \psi \right> | 
 \end{align}
and:
\begin{align}
 || (1- U_{-a} ) \left| \psi \right> ||^2/a & = 2 (1 - {\rm Re} \left< \psi \right| U_{-a} \left| \psi \right> )/a \rightarrow 0  \\
 || (1- U_{-a_s} ) \left| \psi \right> ||^2/a &= 2 (1 - {\rm Re} \left< \psi \right| U_{-a_s} \left| \psi \right> )/a  
 \rightarrow 0 \\
 || (1- X_{a} ) \left| \psi \right> ||^2/a & \leq 2 (1 - {\rm Re} \left< \psi \right| X_a \left| \psi \right>)/a 
 = {\rm Re} (1-g(-s,\eta-s+i/2))/a \rightarrow 0
\end{align}
where in the later equation we have used $X_a^\dagger X_a \leq 1$ by dropping various support projectors. Using $|| X_a || \leq 1$ we arrive at the limit:
\begin{equation}
\gamma(s,\eta) =-  i e^{2\pi \eta}(1 + e^{-2\pi s}) P_\psi + \gamma(-s,\eta-s+i/2)
\end{equation}
where finiteness of $P_s$ for $\psi$ implies that we still have the following limit\footnote{See Corollary~\ref{corr3}, where we don't actually need entire states to prove the existence of the following limit at real $s,\eta$ given $P_s$ is finite}:
\begin{equation}
\lim_{\epsilon \rightarrow 0} \frac{1- g(s,\eta)}{\epsilon} = \gamma(s,\eta) =-  i e^{2\pi \eta} P_{s}
\end{equation}
and similarly for $\gamma(s,\eta+i/2)$. This gives the important constraint:
\begin{equation}
\label{sumc}
P_{s} + e^{-2\pi s} P_{-s} = (1 + e^{-2\pi s}) P_\psi
\end{equation}
Note that this equation is satisfied by the known form of $P_s$ for entire states  \eqref{known}. Since the left hand side of \eqref{sumc} is a sum of two positive quantities that are lower semi-continuous (with respect to the limit on $n$) and the right hand side is just continuous this implies that $P_s^{(n)}$ actually varies continuously with $n \rightarrow \infty$. 

Let us spell this out explicitly. Defining:
\begin{equation}
A_n =  (P_s^{(n)} - P_s )  \qquad B_n =(e^{-2\pi s} P_{-s}^{(n)} - e^{-2\pi s} P_{-s} ) 
\end{equation} 
By lower semi-continuity we have:
\begin{equation}
\liminf_n A_n \geq 0 \qquad \liminf_n B_n \geq 0
\end{equation}
But from the constraint \eqref{sumc} applied to $\psi$ and $\psi_n$ we have:
\begin{equation}
0 = \lim_{n} ( A_n + B_n  ) \geq \liminf_{n} A_n + \liminf_n B_n  \geq 0
\end{equation}
using continuity of $P_\psi^{(n)}$ and super-additivity of $\liminf$. This implies that separately:
\begin{equation}
\liminf_n A_n = 0 \implies \liminf_n P_s^{(n)} = P_s
\end{equation}
Using the known form of $P_s^{(n)}$ we find:
\begin{equation}
P_s =  (1-e^{-2\pi s}) \liminf_n (R_n) + P_\psi \qquad s \geq 0
\end{equation}
Now we know this equation is also true for any sub-sequence and since $R_n$ are a bounded sequence of numbers  we can pick a subsequence for which $\lim_k R_{n_k} = \limsup_n R_{n}$ which then implies that $\liminf_n R_n  =  \limsup_n R_{n} = \lim_n R_n$ which we now know must exist. So we have:
\begin{equation}\label{limR}
P_s = R + (P_\psi -R) e^{-2\pi s} \qquad R = \lim_n R_n
\end{equation}
which completes the proof.  \end{proof}

\section{Null Energy in the Natural Self-Dual Cone}
\label{sec:nullcone}

We aim to prove Lemma~\ref{lemma3} in this section.
We will firstly establish the statements about the ANE using entire states. Lower semicontinuity of the expectation values of $P$ is then sufficient to prove the claim.

\begin{proof}
Consider states $\psi_n$ defined in \eqref{defpsin} that limit to $\lim \psi_n \rightarrow \psi$. Consider the structure function $q^{(n)}_\epsilon(s,\eta)$ defined in \eqref{qeps}. Now with these entire states we know that the following limit exists on compact subsets in $\bar{T}_{III} \cup \bar{T}_{IV}$
\begin{equation}
\lim_{\epsilon \rightarrow 0} q_{\epsilon}^{(n)}(s,\eta) = - i e^{2\pi \eta} R_n (1-e^{-2\pi s})
\end{equation}
We can write out this structure function explicitly at $s=  i/2$ for $0 \leq {\rm Im} \eta \leq 1/2$:
\begin{equation}
q_\epsilon^{(n)}(i/2,\eta) = \left< \psi_n \right| \Theta'_n \frac{ 1 - U_{-a}}{\epsilon} (\Theta_n')^\dagger \left| \psi_n \right> = \big< \widehat{\psi}_n \big|  \frac{ 1 - U_{-a}}{\epsilon} \big| \widehat{\psi}_n \big>
\end{equation}
where we have used the relation between conjugation cocycles $\Theta$ and the states in the natural self-dual cone \eqref{def:natcone}. 

That is we can extract the null energy of the state in the natural self-dual cone using the $q_\epsilon$ structure function.
Again using monotone convergence for $\eta = i/4$ we can establish that indeed the ANE is finite here:
\begin{equation}
\widehat{P}_n =  \big< \widehat{\psi}_n \big| P \big| \widehat{\psi}_n \big> = 2 R_n \leq 2 P_n
\end{equation}
Now since $P$ is an unbounded positive operator we know that it varies lower semi-continuously if $\widehat{\psi}_n$ varies continuously. This can be established more explicitly as we did in the subsection above.
We have the result due to Araki \cite{araki1974some}:
\begin{equation}
\label{araki}
|| \big| \widehat{\psi}_n \big> - \big| \widehat{\psi} \big>  ||^2 \leq || \rho_{\psi_n} - \rho_\psi ||
\leq 2 || \left| \psi_n \right> - \left| \psi \right> ||
\end{equation}
where $\rho_{\psi}(\gamma) = \left< \psi \right| \gamma \left| \psi \right>$ is the linear functional on  operators in $\mathcal{A}_C$. The norm of the linear functional is:
\begin{equation}
|| \rho_{\psi_n} - \rho_\psi || = \sup_{\gamma \in \mathcal{A}_C; || \gamma||\leq 1}  | \rho_{\psi_n}(\gamma) - \rho_{\psi}(\gamma) |
\end{equation}
and we have applied the following inequality to get the last inequality in \eqref{araki}
\begin{equation}
| \left< \psi_n \right| \gamma \left| \psi_n \right> -\left< \psi \right| \gamma \left| \psi \right> |
\leq | ( \left| \psi_n\right> -\left| \psi \right>, \gamma \left| \psi \right> )| + | ( \left| \psi_n\right>, \gamma ( \left| \psi_n \right>- \left| \psi \right>) )|  \leq 2 || \gamma || \,\, 
|| \left| \psi_n \right> - \left| \psi \right>||
\end{equation}
for normalized vectors. So we have:
\begin{equation}
| \big< \widehat{\psi}_n \big| P \Pi_\Lambda \big| \widehat{\psi}_n \big> -\big< \widehat{\psi} \big| P \Pi_\Lambda \big| \widehat{\psi} \big> | \leq 2 C_\Lambda || \big| \widehat{\psi}_n \big>- \big| \widehat{\psi} \big>|| \leq  2 \sqrt{2} C_\Lambda \sqrt{|| \left| \psi_n \right>- \left| \psi \right>||} \rightarrow 0
\end{equation}
where $C_\Lambda = || P \Pi_\Lambda || < \infty$. So at any fixed $\Lambda$ we know the projected null energy $P \Pi_\Lambda$ of the states in the natural self-dual cone vary continuously. 

This implies that the ANE is lower semi-continuous. And since the numbers $\widehat{P}_n$ are bounded the limit state must have finite ANE:
\begin{equation}
2 P_\psi \geq \lim_n 2 R_n \geq \widehat{P}
\end{equation}
where we have used the fact that the limit on $R_n$ exists which is something we showed around \eqref{limR}. 
We are done.
 \end{proof}

\section{From Null Energy to Relative Entropy}

\label{sec:rel}

In this section we aim to prove Lemma~\ref{lemma1}. We tried many different approaches 
to this. The main difficulty relates to questions of the finiteness of relative entropy.
If we were to assume all the relative entropies are finite in \eqref{dsdsp} then the simplest way to proceed is to use the so called differentiate formula for relative entropy that involves the Connes cocycle, see for example Eq.~2 of \cite{Longo:2017mbg}, in combination with the hsmi-algebra. This differentiate formula is only correct if one knows a priori that the relative entropy is finite \cite{ohya2004quantum}. We however only want to make the assumptions in \eqref{assumpt} so we need a slightly different approach - in particular we will rely on \eqref{reltheta} to define the relative entropy. This later formula gives $\infty$ iff the relative entropy is infinite and so does not need any further assumptions. 

\begin{proof}
Firstly we note that by monotonicity $S(b_2)$ is finite for all $b_2 \geq c$ and $\bar{S}(b_1)$ is finite for all $b_1\leq c$. Consider the following function:
\begin{equation}
\label{firstfs}
f(s) = \left< \psi \right| \Delta_{c}^{-is} (\Delta_{b_2}')^{-is} \left| \psi \right>
\end{equation}
where $b_2 > c$ and
where we use the shorthand $\Delta_b = \Delta_{\psi|\Omega;A_b}$ and $\Delta_b' = \Delta_{\psi|\Omega;A_b'}$ etc. We aim to use properties of this function to show that $\bar{S}(b_2)$ is finite which we have \emph{not} assumed. 
Writing:
\begin{equation}
f(s) =\left(\Delta_{c}^{is^\star} \left| \psi \right>,  (\Delta_{b_2}')^{-is} \left| \psi \right>\right)
\end{equation}
By Tomita-Takesaki theory these vectors vary holomorphicaly in the strip $s \in S(0,1/2)$ and are strongly continuous on the closure of the strip such that $f(s)$ is holomorphic in the strip and continuous on the closure. It also satisfies the bound:
\begin{equation}
|f(t+ i \theta)| \leq || \Delta_c^\theta \left| \psi \right> || \,\,  || (\Delta'_{b_2})^\theta \left| \psi \right> ||
%\leq  || \pi_c(\psi) \left| \Omega \right> ||^{2\theta} || \pi_{b_2}'(\psi) \left| \Omega \right> ||^{2\theta} 
  \leq 1 
\end{equation}
for $0 \leq \theta \leq 1/2$. % using the definition of the Tomita operator \eqref{tom}.
For real $s$ we can derive the following relationship:
\begin{equation}
f(s) = \left< \psi \right| (\Delta_c')^{-is} U_{a(1-e^{2\pi s} )} \Delta_{b_2}^{-is} \left| \psi \right> \,, \qquad s\in \mathbb{R}
\end{equation} 
using \eqref{def:cocycle}, \eqref{compmod} and the algebra of half sided modular inclusions \eqref{hsmi:algebra}.  We have set $a = b_2-c > 0$ and in this section we will \emph{not} analytically continue $a$. Note that this later expression for $f(s)$ also has a continuation to the strip $S(0,1/2)$ which can be seen by writing:
\begin{equation}
\label{secondfs}
f_{II}(s) = \left( (\Delta_c')^{is^\star}   \left| \psi \right>, U_{a(1-e^{2\pi s} )} \Delta_{b_2}^{-is} \left| \psi \right> \right)
\end{equation} 
and the fact that ${\rm Im}(1-e^{2\pi s} ) < 0$ for $s\in S(0,1/2)$ so that the translation operator is a bounded operator there. In fact these two analytic continuations must be the same. One way to see this is to show that they have the same values on the top edge of the strip also. Then via an equivalent integral equation to that discussed in \eqref{gkern} these must be the same analytic functions. 

The first expression in \eqref{firstfs} evaluates at $s=t+i/2$ to:
\begin{equation}
f(t+i/2) = \left< \Omega \right| (\Theta_{c}')^\dagger \Delta_{c}^{-it} (\Delta_{b_2}')^{it} \Theta_{b_2} \left| \Omega \right>
\end{equation}
and the second expression \eqref{secondfs} gives:
\begin{equation}
f_{II}(t+i/2) = \left< \Omega \right| \Theta_{c}^\dagger (\Delta_{c}')^{-it} U_{a(1+e^{2\pi t})} \Delta_{b_2}^{it}\Theta_{b_2}' \left| \Omega \right>
\end{equation}
Using $U_{a(1+e^{2\pi t})} = \Delta_{\Omega;c}^{-it} J_{\Omega;c} J_{\Omega;b_2} \Delta_{\Omega;b_2}^{it}$ which can be derived from \eqref{hsmi:algebra} we can write 
\begin{equation}
f_{II}(t+i/2) = \left< \Omega \right| J_{\Omega;c'} \mathcal{O}_{c}'(-t)^\dagger \mathcal{O}_{b_2}(t)   J_{\Omega;b_2} \left| \Omega \right> \,, \qquad \mathcal{O}_{b_2}(t) \equiv  J_{\Omega;b_2} \Delta_{\Omega;b_2}^{it} \Delta_{b_2}^{-it} J_{\Omega|\psi;b_2} 
\end{equation}
and the equivalent expression for $\mathcal{O}_{c}'(t)$. The two operators commute: $[\mathcal{O}_{b_2}(t), \mathcal{O}_c'(-t)]$ since $b_2>c$ which gives:
\begin{equation}
f_{II}(t+i/2)  = \left< \Omega \right| J_{\Omega|\psi;b_2} (\Delta'_{b_2})^{it} \Delta_c^{it} J_{\psi|\Omega;c'} \left| \Omega \right>
\end{equation}
which is equivalent to $f(t+i/2)$.  We conclude that $f(s) = f_{II}(s)$ for $s\in \bar{S}(0,1/2)$. 
There is probably a much simpler way to show this.

Now consider the estimate applied to the second continuation of $f$ (\eqref{secondfs})
\begin{align}
&\left| (1- f(i \theta)) - \left< \psi \right| (1-(\Delta_c')^\theta) \left| \psi \right> -  \left<\psi \right| (1- V_\theta)\left| \psi \right> -\left< \psi \right| (1-\Delta_{b_2}^\theta)  \left| \psi \right>  \right| \\
&\qquad = \nonumber \left| \left< \psi \right| ( 1 - (\Delta_c')^\theta)  V_\theta(1-\Delta_{b_2}^\theta) \left| \psi \right> + \left< \psi \right| ( 1 - (\Delta_c')^\theta)  (1- V_\theta) \left| \psi \right> \right. \\ & \qquad \qquad  \left. + \left< \psi \right| (1- V_\theta)(1-\Delta_{b_2}^\theta)   \left| \psi \right> \right| \\
& \qquad \leq \left| \left< \psi \right| ( 1 - (\Delta_c')^\theta)  V_\theta(1-\Delta_{b_2}^\theta)  \left| \psi \right> \right| + \left| \left< \psi \right|  ( 1 - (\Delta_c')^\theta)  (1- V_\theta) \  \left| \psi \right> \right|\nonumber \\ & \qquad \qquad + \left| \left< \psi \right|  (1- V_\theta)(1-\Delta_{b_2}^\theta)   \left| \psi \right> \right|
\label{eachterm}
\end{align}
where
\begin{equation}
V_{\theta} = U_{a ( 1- e^{2\pi i \theta}) }  = \exp( - a \sin\theta P) \exp( i a(1- \cos\theta)P)
\end{equation}
Now we can use the following limits:
\begin{align}
\label{limissb}
\lim_{\theta \rightarrow 0^+} || (1 - \Delta_{b_2}^\theta) \left| \psi \right> ||^2 /\theta
=\left(1 - 2 \left< \psi \right| \Delta_{b_2}^\theta \left| \psi \right>+\left< \psi \right| \Delta_{b_2}^{2\theta} \left| \psi \right> \right)/\theta \rightarrow 0 \\
\label{limissc}
\lim_{\theta \rightarrow 0^+} || (1 - (\Delta_{b}')^\theta) \left| \psi \right> ||^2 /\theta
=\left(1 - 2 \left< \psi \right| (\Delta_{c}')^\theta \left| \psi \right>+\left< \psi \right| (\Delta_{c}')^{2\theta} \left| \psi \right> \right)/\theta \rightarrow 0 \\
\lim_{\theta \rightarrow 0^+} || (1 - V_\theta) \left| \psi \right> ||^2 /\theta
=\left(1 - 2 {\rm Re} \left< \psi \right| V_\theta \left| \psi \right>+\left< \psi \right| V_\theta^\dagger V_\theta \left| \psi \right> \right)/\theta \rightarrow 0
\label{acancel}
\end{align}
where we have used the assumed finiteness of $S(b_2)$ and $\bar{S}(c)$ and the resulting existence of the limit in \eqref{reltheta} to give the cancelation in \eqref{limissb}-\eqref{limissc}.  We have also used the assumed finiteness of the null energy to compute the limit:
\begin{align}
\lim_{\theta \rightarrow 0^+} \frac{\left< \psi \right|1- V_\theta \left| \psi \right>}{\theta} & = \lim_{\theta \rightarrow 0^+} \int_0^\infty  \frac{(1-e^{ - ia \lambda(1- e^{2\pi i \theta })})}{\theta} d\left< \psi \right| E_{\lambda}(P) \left| \psi \right>  \\ 
& = \lim_{\theta \rightarrow 0^+} \frac{\zeta}{\theta} \int_0^\infty  \frac{(1-e^{  ia \lambda\zeta})}{\zeta} d\left< \psi \right| E_{\lambda}(P) \left| \psi \right>  =  2\pi a P_\psi
\end{align}
where we used the estimate $\left| \frac{(1-e^{i \zeta  \lambda})}{i \zeta}\right| \leq  \lambda$ for $0\leq \arg \zeta  \leq \pi$, which is true for $\zeta = e^{2\pi i \theta}-1$,  and this then allows us to use the dominated convergence theorem in order to pass the limit inside the integral. 
A similar analysis yields:
\begin{equation}
\lim_{\theta \rightarrow 0^+} \left(1 - \left< \psi \right| V_\theta^\dagger V_\theta \left| \psi \right>  \right)/\theta = 4\pi a P_\psi
\end{equation}
which then gives the cancelation in \eqref{acancel}. 

We also need the bound on the norm of this operator:
\begin{equation}
|| V_\theta || \leq 1
\end{equation}
which is true for $0\leq \theta \leq 1/2$.
We can now show that each term in the right hand side \eqref{eachterm} divided by $\theta$ vanishes in the limit. For example:
\begin{align} \nonumber
& \frac{\left| \left< \psi \right| ( 1 - (\Delta_c')^\theta)  V_\theta (1-\Delta_{b_2}^\theta)  \left| \psi \right> \right| }{ \theta} \\  &\qquad \qquad \qquad \qquad 
\leq ||  V_\theta ||  \frac{|| (1 - \Delta_{b_2}^\theta) \left| \psi \right> ||}{\sqrt{\theta}}  \frac{|| (1 - (\Delta_{c}')^\theta) \left| \psi \right> ||}{\sqrt{\theta}} \rightarrow 0
\end{align}
and similarly for the other terms. 
We conclude that:
\begin{equation}
\lim_{\theta \rightarrow 0} \frac{( 1- f(i\theta))}{\theta} = \bar{S}(c) + S(b_2) + 2\pi a P_\psi
\end{equation}

Now we analyze $f(s)$ from the original definition in \eqref{firstfs}. We have the bound:
\begin{align}\nonumber
&\left| 1 - f(i\theta) - (1- \left< \psi \right| \Delta_c^\theta \left| \psi \right> )  -  (1- \left< \psi \right| (\Delta_{b_2}')^\theta) \left| \psi \right>   \right| 
= \left| \left( \left| \psi \right> - \Delta_c^{\theta} \left| \psi \right> , \left| \psi \right> - (\Delta_{b_2}')^{\theta} \left| \psi \right> \right) \right| \\
\label{betterbound}
&\, \leq \left(1- 2\left< \psi \right|  \Delta_c^{\theta}  \left| \psi \right> + \left< \psi \right|  \Delta_c^{2\theta}  \left| \psi \right> \right)^{1/2} \left(1- 2\left< \psi \right| (\Delta'_{b_2})^\theta \left| \psi \right>+\left< \psi \right| (\Delta'_{b_2})^{2\theta} \left| \psi \right>\right)^{1/2}  \\
& \, \leq \left(2- 2\left< \psi \right|  \Delta_c^{\theta}  \left| \psi \right>  \right)^{1/2} \left(2- 2\left< \psi \right| (\Delta'_{b_2})^\theta \left| \psi \right>\right)^{1/2}  
\label{boundcb}
\end{align}
where we used the bound $ \left< \psi \right|  \Delta_c^{2\theta}  \left| \psi \right> \leq 1$ valid for $0 \leq \theta \leq 1/2$. Let us define:
\begin{align}
A = \frac{1- \left< \psi \right| \Delta_c^\theta \left| \psi \right> }{\theta}
\qquad B = \frac{1- \left< \psi \right| (\Delta_{b_2}')^\theta \left| \psi \right>}{\theta} \qquad
C = \frac{1 - {\rm Re} f(i\theta)}{\theta}
\end{align}
we need to show show that the limit $\theta \rightarrow 0$ is finite for $B$ knowing that the same limit for $A,C$ are finite. Note that all $A,B,C$ are non negative. The bound in \eqref{boundcb} translates to:
\begin{equation}
|C - A - B| \leq 2 \sqrt{A B}
\end{equation}
where it is already clear that the limit on $B$ has to be finite. Slightly more explicitly we can translate this bound into:
\begin{equation}
| \sqrt{C} - \sqrt{A} | \leq \sqrt{B} \leq \sqrt{C} + \sqrt{A}
\end{equation}
We know that $B$ behaves well under limits (the limit always exists but could be $\infty$), due to the monotonicity property as a function of $\theta$ of the spectral integral that defines $B$. 
This was discussed around \eqref{reltheta}. So the bound above and finiteness of the limits for $C,A$ imply that the limit on $B$ exists and is equal to:
\begin{equation}
\lim_{\theta \rightarrow 0^+} B = \bar{S}(b_2)  < \infty
\end{equation}
 Knowing this we can improve the bound using \eqref{betterbound} instead, where we now find that the right hand side of \eqref{betterbound} vanishes in the limit and thus $\lim_{\theta \rightarrow 0^+} (C-A-B) = 0$ such that:
\begin{equation}
\bar{S}(b_2) + S(c) = \lim_{\theta \rightarrow 0^+} C = \bar{S}(c) + S(b_2) + 2\pi a P_\psi
\end{equation}
This is the relation we wanted to establish, at least for $b_1 = c$. We have also established that $\bar{S}(b_2) < \infty$ for all $b_2$.  We can repeat the above analysis for $(c,b_2) \rightarrow (b_1,c)$ with $b_1 < c$ to conclude that $S(b_1) < \infty$ for all $b_1$ (switching the roles of $A \leftrightarrow B$ above.)  
Finally we can extend the relation \eqref{dsdsp} to any $b_1, b_2$ by using our newfound knowledge that all the relative entropies are finite which allows us to apply the above discussion more generally. 

The continuity property of $S(b)$ follows since $\bar{S}(b)$ is monotonic so:
\begin{equation}
|S(b_1) - S(b_2) |  \leq 2\pi P_\psi |b_1- b_2|
\end{equation}
which is the definition of Lipschitz continuous. 
\end{proof}

\section{Discussion}
\label{sec:disc}

We end with a discussion of some loose ends and also some possible directions for future work. 

While our discussion of relative entropy in the flowed and natural cone states is very general there was still one main assumption which was the requirement of finite $P_\psi$ for the input state (or for some possible purification.) Technically this assumption was necessary since it allowed us to extrapolating our results (Lemma~\ref{lemma2}-\ref{lemma3}) for entire states to more general states.
It is natural to ask if we can relax this condition. Since we are discussing bounds on relative entropy (for the outer/unprimed region) in the first place the relative entropy must always be finite - does this imply that the state in the natural cone has finite null energy? This is not obvious to us either way. The finiteness of relative entropy for the ``outer region'' might remove possible IR issues and it seems we have dealt with any possible UV issues at the entangling cut in this paper. It is possible that a more thorough study of the structure function $q_\epsilon$ could answer this question either way. It is also possible that the results in \cite{Lashkari:2018nsl} can be used for similar purposes - it would be interesting to explore this moving forwards.

\subsection{Possible Relations to Recovery Maps}

A starting point for this work was an attempt to apply recent results in quantum information theory \cite{wilde2015recoverability,junge2018universal} which, from a very simple minded perspective, give strengthening's of the monotonicity property of relative entropy or as it is known in this work - the data processing inequality.
This body of work aims to find a best guess for the inverse of the action of a noisy quantum channel $\mathcal{N}$ (completely positive trace preserving map) on some density matrices $\rho,\omega$. This best guess, or recovered state, is an important ingredient in the strengthened version of monotoncity.
In terms of density matrices  $\rho$ and $\omega$ the data processing inequality is:
\begin{equation}
S(\rho ||\omega) - S(\mathcal{N}[\rho] || \mathcal{N}[\omega]) \geq 0
\end{equation}
%There are numerous bounds in the literature that improve the right hand side above. 
%We will find that the version given in \cite{junge2018universal} is the most interesting for our purposes. 
We consider the special case in which the completely positive map is associated to an inclusion of algebras. 
We will model this for now with Hilbert spaces that are finite dimensional with the factorization $\mathcal{H}_A = \mathcal{H}_B \otimes  \mathcal{H}_C $ and where the quantum channel is simply a partial trace over $C$. 
This is not actually appropriate in QFT, since the algebra's have type-III factors, however the results in \cite{junge2018universal} were not worked out in this more general setting. Instead we will simply try to use these formulas to guess how it might translate into the QFT case. We consider the density matrices $\rho_A$ and $\omega_A$ that come from global pure states $\left| \psi \right>, \left| \Omega\right>$ respectively. 

Consider the improved monotonicity inequality proven in \cite{junge2018universal} which will be the main ingredient of our discussion:
\begin{equation}
\label{impbd}
S(\rho_A || \omega_A) - S(\rho_B || \omega_B) \geq -2 \pi \int_{-\infty}^{\infty} ds (\cosh(2\pi s) +1)^{-1} \log \left\{ \mathcal{F}(\rho_A ,\mathcal{R}^{s}_{\omega, \mathcal{N}} \circ \mathcal{N}[\rho_A]) \right\} 
\end{equation}
The noisy channel denoted by $\mathcal{N}$ (partial trace over $C$), the approximate recovery map $\mathcal{R}$, and the fidelity between the original and the recovered state $\mathcal{F}$  have the following form:
\begin{equation}
\begin{aligned}
\mathcal{N}(\rho_A) &= \Tr_C \rho_A = \rho_B\\
\mathcal{R}^{s}_{\omega, \mathcal{N}} \circ \mathcal{N}(\rho_A) &= \rho_R = \omega_A^{-is + 1/2} \left( \omega_B^{is - 1/2} \rho_B \omega_B^{-is - 1/2} \otimes 1_{C} \right) \omega_A^{is + 1/2}\\
\mathcal{F}(\rho_A , \rho_R) &= \left[\Tr \sqrt{\rho_A^{1/2} \rho_R \rho_A^{1/2}} \right]^2 
\end{aligned}
\end{equation}
The $s=0$ case above is known as the Petz map \cite{petz1986sufficient,petz1988sufficiency,ohya2004quantum}.
Consider the expectation values of operators in this recovered state:
\begin{equation}
 \Tr_A \mathcal{R}^s ( \Tr_C( \rho_A) O_A )= \Tr_B \rho_B (\mathcal{R}^{s})^\dagger(O_A) = \left< \psi \right| (\mathcal{R}^{s})^\dagger(O_A) \left| \psi \right>
\end{equation}
where the adjoint of the recovery map is defined with respect to the matrix inner product: $\Tr M_1^\dagger M_2$. 
To translate these results into QFT we firstly write them in terms of modular operators. We can represent these here as operators on matrices. 
One can write the adjoint recovery map at $s=0$ in the form:
\begin{equation}
L_{(\mathcal{R}^{0})^\dagger(O_A)} = J_B \mathcal{V}^\dagger J_A L_{O_A} J_A \mathcal{V} J_B
\end{equation}
where we have defined the following operators on matrices:
\begin{align}
L_X(M) = X M \qquad R_X(M) &= M X \qquad  J_A(M_A) = M_A^\dagger  \qquad J_B(M_B) = M_B^\dagger \\
& \mathcal{V}^\dagger = R_{\omega_B^{-1/2}} \mathcal{N} R_{\omega_A^{1/2}}  
\end{align}
For non-zero $s$ we simply include modular flow:
\begin{equation}
(\mathcal{R}^{s})^\dagger(O_A) = \sigma_{-s}^B((\mathcal{R}^{0})^\dagger(\sigma_s^{A} (O_A) ))
\end{equation}
where $\sigma_s^A(O_A) = \omega_A^{i s} O_A \omega_A^{-is}$ etc. Note that:
\begin{equation}
\mathcal{V}(L_{O_B} \omega_B^{1/2}) = L_{O_B \otimes 1_C} \omega_A^{1/2}
\end{equation}
and in this form represents an isometric embedding of the Hilbert spaces of matrices $\mathcal{V}: {\rm End}(\mathcal{H}_{B}) \rightarrow {\rm End}(\mathcal{H}_{A})$ commuting with the action of the algebra $L_{O_B}$. Note that $\omega_B^{1/2}, \omega_A^{1/2}$ should be interpreted as vectors in the Hilbert space of matrices represent certain purifications of the ``vacuum'' state. 
When we pass to QFT 
we should replace $\omega_B^{1/2} \rightarrow \left| \Omega \right>$ and $\omega_A^{1/2} \rightarrow \left| \Omega \right>$ both on the same Hilbert space and the isometric embedding $\mathcal{V}$ becomes trivial. 
The resulting recovery map at $s=0$ is well known \cite{accardi1982conditional,petz1986sufficient,petz1988sufficiency} and in fact von Neumann algebras was the original setting where Petz studied this. 

Putting this together we get the following adjoint recovery map appropriate for QFT:
\begin{equation}
(\mathcal{R}^s)^\dagger(O_A) = \sigma_{-s}^B(J_B J_A \sigma_s^A(O_A) J_A J_B) 
\end{equation}
We can now apply the theory of half-sided modular inclusions \eqref{hsmi:algebra}:
\begin{equation}
(\mathcal{R}^s)^\dagger(O_A) =U_{a(1+ e^{2\pi s})} O_A U_{-a(1+ e^{2\pi s})} 
\end{equation}
where $a$ is the null translation between null cuts $A$ and $B$.
Passing back to the Schrodinger picture the appropriate recovered state on the full QFT Hilbert space is simply:
\begin{equation}
U_{-a(1+ e^{2\pi s})} \ket{\psi}    
\end{equation}
Now the improved bound in \eqref{impbd} tells us to find the fidelity of this state with $\left| \psi \right>$ when restricted to the algebra $A$. The fidelity in the general von Neumann setting was defined by Uhlmann \cite{uhlmann1976transition} as:
\begin{equation}
\label{uhlfid}
\mathcal{F}_A( \psi , U_{-a(1+ e^{2\pi s})} \psi ) 
\equiv \sup_{u_A' } \left| \left< \psi| u_A'  U_{-a(1+ e^{2\pi s})} | \psi \right> \right|
\end{equation}
We would like to be able to compute this in the limit $a \rightarrow 0$. This is a hard task that we do not solve generally. 

However in the limit of large subsystem size for both $A$ and $B$ we can make progress.
In this case the difference of relative entropies approaches the ANE. 
This can be seen from the sum rule \eqref{dsdsp} since we should take the complement relative entropies to vanish when the complement regions are pushed far away.

We want to get a simplified expression for the recovery bound. In this approximation, the expression for fidelity in \eqref{uhlfid} becomes the square of the transition amplitude between the two pure states:
\begin{equation}
\mathcal{F}_A(\ket{\psi} ,{\mathcal{R}^{s}}({\ket{\psi}})) \rightarrow \abs{\bra{\psi} e^{-ia(1+e^{2\pi s}) P} \ket{\psi}}^2
\end{equation}

\noindent
The bound on monotonicity can be further simplified with the coordinate transformation $ y = a (1 + e^{2 \pi s}) $ and the identity $-2 \pi (\cosh(2\pi s) +1)^{-1} = 2\partial_s (\frac{1}{e^{2\pi s} + 1})$. This gives rise to the following inequality:

\begin{equation}
\label{tosat}
\bra{\psi} P \ket{\psi} \geq - \frac{2}{\pi}\int_{a}^{\infty} dy \frac{1}{y^2} \log{\abs{\bra{\psi} e^{-i y P} \ket{\psi}}}
\end{equation}

\noindent
We can do this later integral:
\begin{equation}
-\pi \mathcal{I} =\int_{a}^{\infty} dy \frac{1}{y^2} \log{\bra{\psi} e^{i y P} \ket{\psi}}
+ \int_{a}^{\infty} dy \frac{1}{{y}^2} \log{\bra{\psi} e^{-i y P} \ket{\psi}}
\end{equation}
by sending $y \rightarrow -y$ in the second integral after which we can combine the integrals and deform the contour to give:
\begin{equation}
\mathcal{I} = \int_{-a}^{a} dy \frac{1}{\pi{(y+i\epsilon)}^2} \log{\bra{\psi} e^{i y P} \ket{\psi}}
\end{equation}
and we have dropped a contribution from large $y$ in the UHP which is justified since the wavefunction overlap is bounded (by $1$) there.

The argument of the logarithm can be expanded. Up to first order in $a$, this integral is
\begin{equation}
\mathcal{I} =  \bra{\psi} P \ket{\psi} - \frac{a}{\pi} (\bra{\psi} P^2 \ket{\psi} - \bra{\psi} P \ket{\psi}^2) + O(a^3)
\end{equation} 
We see that ANE is saturated by the bound from recovery map if we take $A$ to be the whole region. The improved bound on monotonicity becomes the (trivial) statement that the variance of the ANE is positive:
\begin{equation}\label{pfluct}
\frac{a}{\pi} \left( \bra{\psi} {P^2} \ket{\psi} - \bra{\psi} P \ket{\psi}^2 \right) \geq 0
\end{equation}
We can compare this to other recovery maps, for example \cite{wilde2015recoverability} where we do not find saturation:
\begin{equation}
\begin{aligned}
S(\rho_A || \omega_A) - S(\rho_B || \omega_B) &\geq  -\log \left\{ \sup_{s \in \mathbb{R}} \mathcal{F}(\rho_A ,\mathcal{R}^{s}_{\omega, \Tr_C} \circ \Tr_C[\rho_A]) \right\} \\
&\geq -\log \left\{ \sup_{s\in \mathbb{R}} \abs{ \bra{\psi} U_{-a(1+ e^{2\pi s})} \ket{\psi}}^2 \right\} = -\log \left\{ \abs{ \bra{\psi} U_{-a} \ket{\psi}}^2 
\right\} \\
\end{aligned}
\end{equation}
and in the later step we have assumed the largest overlap comes from the smallest translation. 
Expanding this fidelity, we see that the lowest order term is quadratic:
\begin{equation}
S(\rho_A || \omega_A) - S(\rho_B || \omega_B) \geq a^2( \bra{\psi} {P^2} \ket{\psi} - \bra{\psi} P \ket{\psi}^2 )
\end{equation}
which is not as tight as the bound in \eqref{impbd}.

This result has led us to conjecture that this saturation of the bound \eqref{impbd} in the limit $a \rightarrow 0$ continues to hold if we do not make the large subsystem size approximation.
For example it might be that the purifications that we worked with for the most of this paper, which involve the more complicated state dependent relative modular flows, might play a role in estimating the fidelity in \eqref{uhlfid}. Perhaps the methods of non-commutative $L_p$ spaces \cite{Lashkari:2018nsl,kosaki1984applications,haagerup1979lp} will be important for this. We suspect this could be the case since the original proofs of the strengthened inequalities relied on the finite quantum system version of these $L_p$ spaces. 

If this conjecture is true then it would be fascinating to compute the leading quadratic correction as $a \rightarrow 0$ and see what replaces \eqref{pfluct}.

\subsection{Other future directions}

We suspect that these methods might lead to new and improved bounds compare to the QNEC. For example if the recovery map story above works out then moving to second order in the limit $a \rightarrow 0$ could result in new bounds. It would also be nice to try to work out a story away from the robust confines of Rindler cuts. In more general curved spacetimes there may be no natural vacuum state to work with in order compute relative entropies. However an approximate vacuum near the cut might do the job and this is especially interesting if we only go after the QNEC which is a somewhat local constraint. Perhaps it is local enough to not care about the details of the state one should compare to. 

Potential other targets for these results include a possible algebraic approach to the statement of QNEC saturation \cite{Leichenauer:2018obf}. This is a statement about the second functional variation of relative entropy and its vanishing for the diagonal/contact piece that appears in this variation. This was shown originally in holographic theories \cite{Leichenauer:2018obf} and then for theories with a twist gap in \cite{qnecsat}. The distinctive behavior of free theories, where saturation is absent, might make one suspect the algebraic approach is not suitable for this question. However we are still optimistic that there might be a story here.

It would also be interesting to put various applications of the ANEC \cite{Hofman:2008ar} and the QNEC \cite{Callebaut:2018nlq} through these modular flow ``filters'' and see what happens. For example it would be fascinating to see what becomes of the BMS algebra uncovered in \cite{Cordova:2018ygx} under the action of relative modular flow.

Finally it is important to uncover the AdS/CFT dual of these statements. Likely the methods studied in \cite{Engelhardt:2018kcs,Engelhardt:2017aux} and \cite{Neuenfeld:2018dim,Casini:2018kzx} would be useful here.

\appendix

\section{Relative modular operator}% - Support Projection Version}
\label{app:mod}

In this appendix we collect various formula related to the relative modular operator. We are particularly interested in defining these objects when the vector states are not necessarily cyclic and separating.  
These considerations are standard and can be found in the Appendix of \cite{araki1982positive}. {\emph We warn the reader that we have a different convention for labeling our $S$ and $\Delta$ relative modular operators - the state labels are switched. This convention was used in \cite{Witten:2018zxz} and we stick with this.}

Take $\psi$ to not be cyclic and separating. This means that there could be some $\alpha \in \mathcal{A}$ such that $\alpha \left| \psi \right> =0$ (not separating) and also that $\mathcal{A} \left| \psi\right> $ may generate a proper subspace of $\mathcal{H}$ instead of the full Hilbert space (not cyclic). To describe this situation we define support projections as the minimal projectors that satisfy:
\begin{align}
s^\mathcal{A}(\psi) \left| \psi \right> = \left| \psi \right>  \qquad s^\mathcal{A}(\psi) \in \mathcal{A} \\
s^\mathcal{A'}(\psi) \left| \psi \right> = \left| \psi \right>  \qquad s^\mathcal{A'}(\psi) \in \mathcal{A}' 
\end{align}
An equivalent definition follows from finding the projector onto the following subspaces:
\begin{align}
\left[ \mathcal{A}' \left| \psi \right> \right] &= \pi(\psi)\mathcal{H}  \subset \mathcal{H} \\
\left[ \mathcal{A} \left| \psi \right> \right]&= \pi'(\psi)\mathcal{H}  \subset \mathcal{H} 
\end{align}
These are seen to be equivalent as follows. Firstly the $\pi(\psi)$ commutes with $\mathcal{A}'$ since for arbitrary state $\left|\phi_i\right> \in \mathcal{H}$ and for all $\alpha' \in \mathcal{A}'$:
\begin{align}
\left< \phi_1 \right| \left[ \pi_\mathcal{A}(\psi), \alpha' \right] \left| \phi_2 \right>
& = \left< \chi_1 \right| \alpha' \left| \phi_2 \right> - \left< \phi_1 \right| \alpha' \left| \chi_2 \right>
\qquad  & \left( \left| \chi_i \right> = \pi(\psi) \left| \phi_i \right> \right) \\
& =  \left< \chi_1 \right| \alpha' \left| \chi_2 \right> - \left< \chi_1 \right| \alpha' \left| \chi_2 \right>= 0
\qquad &\left( \alpha' \left| \chi_i \right> =  \pi(\psi) \alpha' \left| \chi_i \right> \right)  
\end{align}
so $\pi$ is in $\mathcal{A}$.
Secondly it is the minimal such projector leaving $\psi$ invariant since if it were not there would be another projector $\pi_2 \in \mathcal{A}$ with $\pi_2 \mathcal{H} \subset \pi(\psi) \mathcal{H}$ which also leaves invariant the subspace:
\begin{equation}
\left[ A' \left| \psi \right> \right] = \left[ A' \pi_2 \left| \psi \right> \right] = \pi_2 \left[ A' \left| \psi \right> \right] 
\end{equation}
such that $\pi_2 \mathcal{H} \subset \pi(\psi) \mathcal{H} = \pi_2  \pi(\psi) \mathcal{H} \subset \pi_2 \mathcal{H}$ implying that $\pi_2 = \pi(\psi)$. Thus
$\pi = s^\mathcal{A}$ and similarly for the commutant. 

Note that if $\pi(\psi)$ is not the unit operator, then $\psi$ is not cyclic for $\mathcal{A}'$ and  $(1-\pi(\psi))$ annihilates $\psi$ which means that $\psi$ is not separating for $\mathcal{A}$. That is the lack of either of these two properties exchange under $\mathcal{A} \leftrightarrow \mathcal{A}'$.

We now move to the modular operators. We will consider two state $\psi,\phi$ neither of which needs to be cyclic and separating. 
We start with the definition of the Tomita operators:
\begin{align}
\label{defpp}
S_{\psi|\phi}\left( \alpha \left| \psi \right> + \left| \chi' \right> \right) &= \pi(\psi) \alpha^\dagger \left| \phi \right>  \qquad \forall \,\, \chi' \in (1-\pi'(\psi)) \mathcal{H} \\
S_{\phi|\psi} \left( \alpha \left| \phi \right> + \left|\xi'\right> \right)&=  \pi(\phi) \alpha^\dagger \left| \psi \right>  \qquad \forall \, \, \xi' \in (1-\pi'(\phi)) \mathcal{H}
\label{defpp2}
\end{align}
for $\alpha \in \mathcal{A}$. Note that, for the first equation above, if both $\alpha \left| \psi \right> =0$ and $\left| \chi' \right> =0$ vanish then $\alpha [ \mathcal{A}' \left| \psi \right> ] = 0 \implies \pi(\psi) \alpha^\dagger = 0$ so $0$ is mapped to $0$ as is necessary for a linear operator. The Tomita operators are closable as defined and we will use the same symbol for the closure as the original operator. The support of these operators is:
\begin{equation}
\label{suppS}
{\rm supp} ( S_{\psi|\phi}, \, S_{\phi|\psi}^\dagger ) = \pi'(\psi) \pi(\phi) \mathcal{H} \qquad  {\rm supp} ( S_{\phi|\psi}, \,  S_{\psi|\phi}^\dagger ) = \pi(\psi) \pi'(\phi) \mathcal{H}
 \end{equation}
Applying the definitions twice we have:
\begin{align}
\label{Ssquare}
S_{\psi|\phi} S_{\phi|\psi} &= \pi(\psi) \pi'(\phi) \\
S_{\phi|\psi}  S_{\psi|\phi} &= \pi'(\psi) \pi(\phi)
\end{align}
For the commutant algebra we have:
\begin{align}
S_{\psi|\phi}'\left( \alpha' \left| \psi \right> + \left| \chi \right> \right) &= \pi'(\psi)  (\alpha')^\dagger \left| \phi \right>  \qquad \forall \, \, \chi \in (1-\pi(\psi)) \mathcal{H} \\
S_{\phi|\psi}' ( \alpha' \left| \phi \right>  + \left| \xi \right>) &= \pi'(\phi) (\alpha')^\dagger \left| \psi \right>
\qquad \forall \, \, \chi \in (1-\pi(\phi)) \mathcal{H} 
\end{align}
for $\alpha' \in \mathcal{A}'$ with support that is complementary to \eqref{suppS}. And similar equations hold for the commutant as in \eqref{Ssquare}. Now consider:
\begin{align}
& \left( \beta' \left| \psi \right> + \left| \chi \right> , S_{\psi|\phi}\left( \alpha \left| \psi \right> + \left| \chi' \right> \right) \right)  = \left( \beta' \left| \psi \right> + \left| \chi \right> ,  \pi(\psi) \alpha^\dagger \left| \phi \right>  \right)  \\
& \quad  =  \left( \beta' \left| \psi \right>  ,  \alpha^\dagger \left| \phi \right>  \right)
= \left( \alpha \left| \psi \right>, (\beta')^\dagger \left| \phi \right> \right) = \left( \alpha \left| \psi \right> + \left| \chi'\right>, \pi'(\psi) (\beta')^\dagger \left| \phi \right> \right) \\
&\quad = \left( \alpha \left| \psi \right>  + \left| \chi'\right> , S_{\psi|\phi}'\left( \beta' \left| \psi \right> + \left| \chi \right> \right) \right)
\end{align} 
which means that (because the above states are dense on the appropriate support)
\begin{equation}
\label{stocomp}
S_{\psi|\phi}' = S_{\psi|\phi}^\dagger \,, \qquad  S_{\phi|\psi}' = S_{\phi|\psi}^\dagger
\end{equation}
where the later equation follows a similar analysis.
%Similarly we have:
%\begin{align}
%&\left( \beta' \left| \phi \right>, S_{\phi|\psi} \alpha \left| \phi \right>\right) =
%\left( \beta' \left| \phi \right>,  \alpha^\dagger \left| \psi \right>\right)  \\
%& \quad = \left( \alpha \left| \phi \right>, (\beta')^\dagger \left| \psi \right>\right) 
%=  \left( \alpha \left| \phi \right>, S_{\phi|\psi}' \beta' \left| \psi \right>\right) 
%\end{align}
%which implies that

We move now to the relative modular operators. Consider the positive self adjoint operators:
\begin{equation}
\label{delsq}
\Delta_{\psi|\phi} = S_{\psi|\phi}^\dagger S_{\psi|\phi}
\qquad \Delta_{\phi|\psi} = S_{\phi|\psi}^\dagger S_{\phi|\psi}
\end{equation}
with support 
\begin{equation}
{\rm supp} ( \Delta_{\psi|\phi} ) = \pi'(\psi) \pi(\phi) \mathcal{H} \qquad  {\rm supp} ( \Delta_{\phi|\psi}  ) = \pi(\psi) \pi'(\phi) \mathcal{H}
 \end{equation}
For the commutant algebra we learn that:
\begin{align}
\Delta_{\psi|\phi}' =  (S_{\psi|\phi}')^\dagger S_{\psi|\phi}' =  S_{\psi|\phi} S_{\psi|\phi}^\dagger \quad \implies \quad \Delta_{\psi|\phi}'  \Delta_{\phi|\psi} =  \pi(\psi) \pi'(\phi) \\
\Delta_{\phi|\psi}' =  S_{\phi|\psi} S_{\phi|\psi}^\dagger \quad \implies \quad \Delta_{\phi|\psi}'  \Delta_{\psi|\phi} =  \pi'(\psi) \pi(\phi)
\end{align}
We can define powers of the modular operators $\Delta_{\psi|\phi}^z$ etc. to be zero when acting on $(1-\pi'(\psi) \pi(\phi)) \mathcal{H}$ and to be the usual power when acting on the support of $\Delta_{\psi|\phi}$
which for example means that $\Delta_{\psi|\phi}^0 =\pi'(\psi) \pi(\phi) $. So for example we have:
\begin{equation}
(\Delta_{\psi|\phi}')^{z} \Delta_{\phi|\psi}^{-z} = \pi(\psi) \pi'(\phi)
\end{equation}
Furthermore we can apply polar decompositions to the Tomita operators, where the anti-linear part is not unitary, but rather a partial anti-linear isometry with the support and range of the Tomita operators:
\begin{equation}
S_{\psi|\phi} =  J_{\psi|\phi} \Delta_{\psi|\phi}^{1/2} \qquad \qquad {\rm etc}
\end{equation}
where:
\begin{align}
J_{\psi|\phi}^\dagger J_{\psi|\phi} &= \pi'(\psi) \pi(\phi) \qquad J_{\psi|\phi} J_{\psi|\phi}^\dagger = \pi(\psi) \pi'(\phi) \\
 J_{\phi|\psi}^\dagger J_{\phi|\psi} &= \pi(\psi)\pi'(\phi) \qquad J_{\phi|\psi} J_{\phi|\psi}^\dagger = \pi'(\psi) \pi(\phi) \qquad \qquad {\rm etc}
\end{align}
(with appropriate support and range.)
Plugging back into the Tomita operators and \eqref{delsq} we have:
\begin{equation}
J_{\psi|\phi} \Delta^{1/2}_{\psi|\phi} J_{\phi|\psi} \Delta^{1/2}_{\phi|\psi} = \pi(\psi) \pi'(\phi)
\quad \implies \quad J_{\psi|\phi} J_{\phi|\psi} ( J_{\phi|\psi}^\dagger \Delta_{\psi|\phi}^{1/2} J_{\phi|\psi} ) = \Delta_{\phi|\psi}^{-1/2} 
\end{equation}
which by the uniqueness of the polar decomposition implies that:
\begin{align}
J_{\psi|\phi} J_{\phi|\psi} =  \pi(\psi) \pi'(\phi) \qquad J_{\psi|\phi} = J_{\phi|\psi}^\dagger \qquad J_{\psi|\phi} \Delta_{\psi|\phi}^{1/2} J_{\phi|\psi} = \Delta_{\phi|\psi}^{-1/2} 
\end{align} 
where the last equation implies that:
\begin{equation}
\label{jimport}
J_{\psi|\phi} \Delta_{\psi|\phi}^{is} = \Delta_{\phi|\psi}^{is} J_{\psi|\phi} 
\end{equation}
by the anti-linearity of $J$. 

For the complement we use \eqref{stocomp} to derive:
\begin{equation}
\label{compmod}
J_{\psi|\phi}' = J_{\phi|\psi} \qquad (\Delta_{\psi|\phi}')^z = \Delta_{\phi|\psi}^{-z}
\end{equation}

In order to understand relative modular flow and co-cycles we have to apply the Connes $2\times 2$ or $3\times 3$ trick which we present here in a vector language. Consider the Hilbert space:
\begin{equation}
\mathcal{H}_{\rm tot} = \mathcal{H}_L \otimes  \mathcal{H}_R \otimes \mathcal{H}_{QFT}
\end{equation}
where $\mathcal{H}_{L,R}$ are both simple $n$-dimensional qunit Hilbert spaces with basis $\left|i \right>; i = 0,1, \ldots, n-1$. In this Hilbert space we consider the state:
\begin{equation}
\left| \Psi \right> = \sum_{i=1}^{n} \frac{1}{\sqrt{n}} \left| i_L i_R \right> \otimes \left| \phi_i \right>
\end{equation}  
where $\left| \phi_i \right>$ are vector states in the QFT Hilbert space that need not be cyclic and separating. We will actually consider the state as living in the subspace projected by the support projectors of the states $\phi_i$:
\begin{equation}
\widetilde{\mathcal{H}} = \left[ \left| i_L j_R \right> \otimes \pi'(\phi_i) \pi(\phi_j) \mathcal{H}_{QFT}: i,j = 0,1, \ldots n-1 \right] \subset \mathcal{H}_{\rm tot}
\end{equation}
where the square brackets means the linear span. We now consider the algebra of operators acting on this new Hilbert space:
\begin{equation}
\label{newalg}
\gamma \in \mathbb{A}\,: \qquad \gamma = \sum_{ij} 1_L \otimes \left( \left| i \right> \left< j \right| \right)_R \otimes c_{ij} \qquad c_{ij} \in \pi(\phi_i) \mathcal{A}\, \pi(\phi_j)  
\end{equation}
The commutant is:
\begin{equation}
\gamma' \in \mathbb{A}'\,: \qquad \gamma' = \sum_{ij} \left( \left| i \right> \left< j \right| \right)_L \otimes 1_R \otimes c'_{ij} \qquad c_{ij}' \in \pi'(\phi_i) \mathcal{A}'\, \pi'(\phi_j)  
\end{equation}
and $\Phi$ is cyclic and separating for these algebras after we project to the subspace $\tilde{\mathcal{H}}$. The generalized Tomita operator is:
\begin{equation}
\mathbb{S} \gamma \left| \Psi \right> = \gamma^\dagger \left| \Psi \right>
\end{equation}
from which one finds:
\begin{equation}
\mathbb{S} = \sum_{ij} \left| j_L i_R \right> \left< i_L j_R \right| \otimes \tilde{S}_{i|j}
\,, \qquad  \tilde{S}_{i|j} c_{ji} \left| \phi_i \right> = c_{ji}^\dagger \left| \phi_j \right>
\end{equation}
Note that this later Tomita operator acts between Hilbert spaces:
\begin{equation}
\tilde{S}_{i|j}\, : \pi'(\phi_i) \pi(\phi_j) \mathcal{H}_{QFT} \rightarrow \pi'(\phi_j) \pi(\phi_i) \mathcal{H}_{QFT} 
\end{equation}
We can relate this to our original definition of the Tomita operators by passing back to the original unprojected Hilbert spaces setting $c_{ji} = \pi(\phi_j) \alpha \pi(\phi_i)$: 
\begin{equation}
\tilde{S}_{i|j} \pi(\phi_j) \alpha \left| \phi_i \right> = \pi(\phi_i) \alpha^\dagger \left| \phi_j \right>
\end{equation}
We can lose the $\pi(\phi_j)$ on the left hand side since the orthogonal part is killed on the right hand side anyway. Or in other words we can extend the definition $\tilde{S}_{i|j}$ to the larger Hilbert space consistent with the right hand side by dropping this projector and also demanding:
\begin{align}
&S_{i|j}  \alpha \left| \phi_i \right> = \pi(\phi_i) \alpha^\dagger \left| \phi_j \right> \\
&S_{i|j} (1- \pi'(\phi_i) ) \mathcal{H}_{QFT} = 0 
\end{align}
which is the same definition we gave for $i=\psi$ and $j=\phi$ in \eqref{defpp}.

The modular operator for the $n \times n$ Hilbert space is:
\begin{equation}
\mathbf{\Delta} = \mathbb{S}^\dagger \mathbb{S}
= \sum_{ij} \left| i_L j_R \right> \left< i_L j_R \right| \otimes \tilde{\Delta}_{i|j}\,, \qquad \tilde{\Delta}_{i|j} = \tilde{S}_{i|j}^\dagger \tilde{S}_{i|j}
\end{equation}
and the modular conjugation operator is:
\begin{equation}
\mathbb{J}  = \sum_{ij} \left| j_L i_R \right> \left< i_L j_R \right| \otimes \tilde{J}_{i|j}
\qquad \tilde{S}_{i|j} = \tilde{J}_{i|j} \tilde{\Delta}_{i|j}^{1/2}
\end{equation}
where $\mathbb{J}^2 = 1,\, \mathbb{ J} \mathbf{\Delta}  \mathbb{J} = \mathbf{\Delta}^{-1}$.
The projected modular operators act between the following Hilbert spaces:
\begin{align}
\tilde{\Delta}_{i|j}\, : \pi'(\phi_i) \pi(\phi_j) \mathcal{H}_{QFT}& \rightarrow \pi'(\phi_i) \pi(\phi_j) \mathcal{H}_{QFT} \\
\tilde{J}_{i|j}\, : \pi'(\phi_i) \pi(\phi_j) \mathcal{H}_{QFT}& \rightarrow \pi'(\phi_j) \pi(\phi_i) \mathcal{H}_{QFT}  
\end{align}
and can be extended to $\mathcal{H}_{QFT}$ as we did with the Tomita operators.

We now apply the results of Tomita-Takesaki theory to these new modular operators.
That is we know that $\mathbb{A} = \mathbf{\Delta}^{is} \mathbb{A} \mathbf{\Delta}^{-is}$
and $\mathbb{A}' = \mathbb{J} \mathbb{A} \mathbb{J}$. Computing:
\begin{align}
\mathbf{\Delta}^{is} \gamma \mathbf{\Delta}^{-is}
&= \sum_k \sum_{ij}  ( \left| k \right> \left< k \right| )_L  \otimes ( \left| i \right> \left< j \right| )_R
\otimes \tilde{\Delta}_{k|i}^{is} c_{ij} \tilde{\Delta}_{k|j}^{-is} \\
\mathbb{J} \gamma \mathbb{J} &= \sum_k \sum_{ij}  ( \left| i \right> \left<j \right|)_L \otimes (\left| k \right> \left< k \right|)_R \otimes \tilde{J}_{k | i } c_{ij} \tilde{J}_{j | k} 
\end{align}

This is only consistent with the form of the algebra given in \eqref{newalg}  if we have:
\begin{equation}
\tilde{\Delta}_{k|i}^{is} c_{ij} \tilde{\Delta}_{k|j}^{-is} \, \in \pi(\phi_i) \mathcal{A} \pi(\phi_j)\,, \qquad 
c_{ij} \in \pi(\phi_i) \mathcal{A} \pi(\phi_j) 
\end{equation}
and the flowed operator is the same operator for all $k$. 
Extending these statements to the full Hilbert space by setting the flowed and conjugated operators to zero away from the support we find that:
\begin{equation}
\Delta_{k|i}^{is} \alpha \Delta_{k|j}^{-is} \, \in \mathcal{A} \pi'(\phi_k)\,, \qquad 
\alpha \in  \mathcal{A}
\end{equation}
where we are forced to add $\pi'(\phi_k)$ so that it vanishes away from this support. For example this is consistent with $s \rightarrow 0$ where we find $\pi(\phi_i) \alpha \pi(\phi_j) \pi'(\phi_k)$. Note that we can pick $\alpha$ in $\mathcal{A}$ rather than the projected algebra since the modular operators above anyway apply this projection. Note the flowed operator is now marginally not independent of $k$ due to $\pi'(\phi_k)$. 

If we set $i=j$ this defines the standard modular automorphism group but now for a non-cyclic and separating vector:
\begin{equation}
\Delta_{i}^{is} \alpha \Delta_i^{-is} = \sigma_s^{\phi_i}(\alpha) \pi'(\phi_i) \,,
\qquad \sigma_s^{\phi_i}(\alpha) = \Delta_{\Omega|i}^{is} \alpha \Delta_{\Omega|i}^{-is} 
\end{equation}
where in the later equation we have used a cyclic and separating vector $\Omega$ to define this flow. 

For $\alpha = 1$ we can define the operator in $\mathcal{A}$ that is independent of $k$ as the co-cycle:
\begin{equation}
\label{def:cocycle}
\Delta_{k|i}^{is} \Delta_{k|j}^{-is} \equiv (D \phi_i : D\phi_j )_s \pi'(\phi_k) \,, \qquad (D \phi_i : D\phi_j )_s \in \mathcal{A} 
\end{equation}
which can be extracted by picking $\phi_k$ to be a cyclic and separating vector.
The co-cycle satisfies:
\begin{align}
(D \phi_i : D\phi_j)^\dagger_s &= (D \phi_j : D\phi_i)_s  \\
(D \phi_i : D\phi_j)^\dagger_s  (D \phi_i : D\phi_j)_s &= \sigma_s^{\phi_i}( \pi(\phi_j))  \\
 (D \phi_i : D\phi_j)_s (D \phi_i : D\phi_j)^\dagger_s  &= \sigma_s^{\phi_j}(\pi(\phi_i))
\end{align}
where the right hand side of the later two equations are projection operators if $[ \pi(\phi_i), \pi(\phi_j)] =0$ which means that for states with commuting support projectors the co-cycle is a partial isometry.

%{\bf Put here commutant versions of the above?}

There is the following relation on triples of the co-cycle:
\begin{equation}
(D \phi_1 : D\phi_2)_s (D\phi_2:D\phi_3)_s = (D\phi_1:D\phi_3)_s 
\qquad  {\rm if} \,\,\,\, \begin{matrix} \pi(\phi_1)\pi(\phi_2) = \pi(\phi_1) \,\, {\rm or} \\ \pi(\phi_2)\pi(\phi_3) = \pi(\phi_3) \end{matrix}
\end{equation}
where the later conditions on the projectors can be guaranteed by demanding:
\begin{equation}
\pi(\phi_1)\mathcal{H} \subset \pi(\phi_2)\mathcal{H} \,\, {\rm or} \,\, \pi(\phi_3)\mathcal{H} \subset \pi(\phi_2)\mathcal{H}
\end{equation}

Similary for the modular conjugation operators we have:
\begin{equation}
\tilde{J}_{k|i} c_{ij} \tilde{J}_{j|k} \, \in \pi'(\phi_i) \mathcal{A}' \pi'(\phi_j)\,, \qquad 
c_{ij} \in \pi(\phi_i) \mathcal{A} \pi(\phi_j) 
\end{equation}
which extends to the $\mathcal{H}_{QFT}$ in the usual way with:
\begin{equation}
J_{k|i} \alpha J_{j|k} \in \mathcal{A}' \pi(\phi_k)\,, \qquad \alpha \in \mathcal{A}
\end{equation}
and where apart from the projector $\pi(\phi_k)$ this is the same operator independent of $k$.
We define the non relative modular conjugation action as:
\begin{equation}
J_i \alpha J_i = j^{\phi_i}(\alpha) \pi(\phi_i) \,,\qquad   j^{\phi_i}(\alpha) = J_{\Omega|i} \alpha J_{i|\Omega}
\end{equation}
The equivalent of the co-cycles are the following linear operators:
\begin{equation}
\label{defconj}
J_{k|i}  J_{j|k}  = \Theta_{i|j}' \pi(\phi_k)\,, \qquad \Theta_{i|j}' \in \mathcal{A}'
\end{equation}
which satisfies:
\begin{align}
(\Theta'_{i|j})^\dagger = (\Theta'_{j|i}) \,, \quad
(\Theta'_{i|j})^\dagger (\Theta'_{i|j}) = j^{\phi_j}(\pi(\phi_i))  \,, \quad
 (\Theta'_{i|j})  (\Theta'_{i|j})^\dagger =  j^{\phi_i}(\pi(\phi_j))
\end{align}
where this is then a partial isometry if $[\pi(\phi_i),\pi(\phi_j)] =0$. 

For example if we specify that $\phi_i = \Omega$, cyclic and separating, and $\phi_j = \psi$ we find that:
\begin{equation}
J_\Omega J_{\psi|\Omega} = \Theta'_{\Omega|\psi}\,, \qquad(\Theta'_{\Omega|\psi})^\dagger (\Theta'_{\Omega|\psi}) = \pi'(\psi)\,,  \qquad (\Theta'_{\Omega|\psi})  (\Theta'_{\Omega|\psi})^\dagger = J_{\Omega} \pi(\psi) J_\Omega
\end{equation}
where the support of this operator is: ${\rm supp}(\Theta_{\Omega|\psi}') = \pi'(\psi) \mathcal{H}$ and 
${\rm supp}(\Theta_{\Omega|\psi}')^\dagger =J_\Omega \pi(\psi) J_{\Omega} \mathcal{H}$. We also have the useful relations:
\begin{equation}
J_{\psi|\Omega} = J_\Omega \Theta'_{\Omega|\psi} \qquad J_{\Omega|\psi} = ( \Theta'_{\Omega|\psi})^\dagger J_\Omega \qquad J_\psi = (\Theta'_{\Omega|\psi})^\dagger J_\Omega \Theta'_{\Omega|\psi}
\end{equation}
There is also a triple relation:
\begin{equation}
\Theta'_{1|2} \Theta'_{2|3} = \Theta'_{1|3}
\qquad  {\rm if} \,\,\,\, \begin{matrix} \pi(\phi_1)\pi(\phi_2) = \pi(\phi_1) \,\, {\rm or} \\ \pi(\phi_2)\pi(\phi_3) = \pi(\phi_3) \end{matrix}
\end{equation}
where again this later condition can be achieved only under the conditions specified for the support projectors.

\acknowledgments

We especially thank Raphael Bousso, Ven Chandrasekaran, Netta Engelhardt, Ben Freivogel, Marius Junge,  Nima Lashkari, Juan Maldacena, Arvin Shahbazi-Moghaddam for discussions related to this work.  This work was supported by the DOE: award number DE-SC0019517.

\bibliography{entanglement}

\providecommand{\href}[2]{#2}\begingroup\raggedright\begin{thebibliography}{10}

\bibitem{Bousso:2015wca}
R.~Bousso, Z.~Fisher, J.~Koeller, S.~Leichenauer and A.~C. Wall, \emph{{Proof
  of the Quantum Null Energy Condition}},
  \href{https://doi.org/10.1103/PhysRevD.93.024017}{\emph{Phys. Rev.}
  {\bfseries D93} (2016) 024017}
  [\href{https://arxiv.org/abs/1509.02542}{{\ttfamily 1509.02542}}].

\bibitem{bousso2016quantum}
R.~Bousso, Z.~Fisher, S.~Leichenauer and A.~C. Wall, \emph{Quantum focusing
  conjecture}, {\emph{Physical Review D} {\bfseries 93} (2016) 064044}.

\bibitem{Bekenstein:1973ur}
J.~D. Bekenstein, \emph{{Black holes and entropy}},
  \href{https://doi.org/10.1103/PhysRevD.7.2333}{\emph{Phys. Rev.} {\bfseries
  D7} (1973) 2333}.

\bibitem{Casini:2008cr}
H.~Casini, \emph{{Relative entropy and the Bekenstein bound}},
  \href{https://doi.org/10.1088/0264-9381/25/20/205021}{\emph{Class. Quant.
  Grav.} {\bfseries 25} (2008) 205021}
  [\href{https://arxiv.org/abs/0804.2182}{{\ttfamily 0804.2182}}].

\bibitem{wall2012proof}
A.~C. Wall, \emph{Proof of the generalized second law for rapidly changing
  fields and arbitrary horizon slices}, {\emph{Physical Review D} {\bfseries
  85} (2012) 104049}.

\bibitem{Witten:2018lha}
E.~Witten, \emph{{APS Medal for Exceptional Achievement in Research: Invited
  article on entanglement properties of quantum field theory}},
  \href{https://doi.org/10.1103/RevModPhys.90.045003}{\emph{Rev. Mod. Phys.}
  {\bfseries 90} (2018) 045003}
  [\href{https://arxiv.org/abs/1803.04993}{{\ttfamily 1803.04993}}].

\bibitem{Balakrishnan:2017bjg}
S.~Balakrishnan, T.~Faulkner, Z.~U. Khandker and H.~Wang, \emph{{A General
  Proof of the Quantum Null Energy Condition}},
  \href{https://arxiv.org/abs/1706.09432}{{\ttfamily 1706.09432}}.

\bibitem{Holzhey:1994we}
C.~Holzhey, F.~Larsen and F.~Wilczek, \emph{{Geometric and renormalized entropy
  in conformal field theory}},
  \href{https://doi.org/10.1016/0550-3213(94)90402-2}{\emph{Nucl. Phys.}
  {\bfseries B424} (1994) 443}
  [\href{https://arxiv.org/abs/hep-th/9403108}{{\ttfamily hep-th/9403108}}].

\bibitem{Calabrese:2004eu}
P.~Calabrese and J.~L. Cardy, \emph{{Entanglement entropy and quantum field
  theory}}, \href{https://doi.org/10.1088/1742-5468/2004/06/P06002}{\emph{J.
  Stat. Mech.} {\bfseries 0406} (2004) P06002}
  [\href{https://arxiv.org/abs/hep-th/0405152}{{\ttfamily hep-th/0405152}}].

\bibitem{Lewkowycz:2013nqa}
A.~Lewkowycz and J.~Maldacena, \emph{{Generalized gravitational entropy}},
  \href{https://doi.org/10.1007/JHEP08(2013)090}{\emph{JHEP} {\bfseries 08}
  (2013) 090} [\href{https://arxiv.org/abs/1304.4926}{{\ttfamily 1304.4926}}].

\bibitem{Hollands:2017dov}
S.~Hollands and K.~Sanders, \emph{{Entanglement measures and their properties
  in quantum field theory}},
  \href{https://arxiv.org/abs/1702.04924}{{\ttfamily 1702.04924}}.

\bibitem{Longo:2018obd}
R.~Longo, \emph{{Entropy distribution of localised states}},
  \href{https://arxiv.org/abs/1809.03358}{{\ttfamily 1809.03358}}.

\bibitem{Longo:2017mbg}
R.~Longo and F.~Xu, \emph{{Relative Entropy in CFT}},
  \href{https://doi.org/10.1016/j.aim.2018.08.015}{\emph{Adv. Math.} {\bfseries
  337} (2018) 139} [\href{https://arxiv.org/abs/1712.07283}{{\ttfamily
  1712.07283}}].

\bibitem{Xu:2018fsv}
F.~Xu, \emph{{Some Results On Relative Entropy in Quantum Field Theory}},
  \href{https://arxiv.org/abs/1810.10642}{{\ttfamily 1810.10642}}.

\bibitem{Xu:2018uxc}
F.~Xu, \emph{{On Relative Entropy and Global Index}},
  \href{https://arxiv.org/abs/1812.01119}{{\ttfamily 1812.01119}}.

\bibitem{Kang:2018xqy}
M.~J. Kang and D.~K. Kolchmeyer, \emph{{Holographic Relative Entropy in
  Infinite-dimensional Hilbert Spaces}},
  \href{https://arxiv.org/abs/1811.05482}{{\ttfamily 1811.05482}}.

\bibitem{Wall:2017blw}
A.~C. Wall, \emph{{Lower Bound on the Energy Density in Classical and Quantum
  Field Theories}},
  \href{https://doi.org/10.1103/PhysRevLett.118.151601}{\emph{Phys. Rev. Lett.}
  {\bfseries 118} (2017) 151601}
  [\href{https://arxiv.org/abs/1701.03196}{{\ttfamily 1701.03196}}].

\bibitem{Klinkhammer:1991ki}
G.~Klinkhammer, \emph{{Averaged energy conditions for free scalar fields in
  flat space-times}},
  \href{https://doi.org/10.1103/PhysRevD.43.2542}{\emph{Phys. Rev.} {\bfseries
  D43} (1991) 2542}.

\bibitem{Kelly:2014mra}
W.~R. Kelly and A.~C. Wall, \emph{{Holographic proof of the averaged null
  energy condition}}, \href{https://doi.org/10.1103/PhysRevD.90.106003,
  10.1103/PhysRevD.91.069902}{\emph{Phys. Rev.} {\bfseries D90} (2014) 106003}
  [\href{https://arxiv.org/abs/1408.3566}{{\ttfamily 1408.3566}}].

\bibitem{Faulkner:2016mzt}
T.~Faulkner, R.~G. Leigh, O.~Parrikar and H.~Wang, \emph{{Modular Hamiltonians
  for Deformed Half-Spaces and the Averaged Null Energy Condition}},
  \href{https://doi.org/10.1007/JHEP09(2016)038}{\emph{JHEP} {\bfseries 09}
  (2016) 038} [\href{https://arxiv.org/abs/1605.08072}{{\ttfamily
  1605.08072}}].

\bibitem{Hartman:2016lgu}
T.~Hartman, S.~Kundu and A.~Tajdini, \emph{{Averaged Null Energy Condition from
  Causality}}, \href{https://doi.org/10.1007/JHEP07(2017)066}{\emph{JHEP}
  {\bfseries 07} (2017) 066}
  [\href{https://arxiv.org/abs/1610.05308}{{\ttfamily 1610.05308}}].

\bibitem{Kravchuk:2018htv}
P.~Kravchuk and D.~Simmons-Duffin, \emph{{Light-ray operators in conformal
  field theory}}, \href{https://doi.org/10.1007/JHEP11(2018)102}{\emph{JHEP}
  {\bfseries 11} (2018) 102}
  [\href{https://arxiv.org/abs/1805.00098}{{\ttfamily 1805.00098}}].

\bibitem{borchers1992cpt}
H.-J. Borchers, \emph{The cpt-theorem in two-dimensional theories of local
  observables}, {\emph{Communications in Mathematical Physics} {\bfseries 143}
  (1992) 315}.

\bibitem{wiesbrock1993half}
H.-W. Wiesbrock, \emph{Half-sided modular inclusions of von-neumann-algebras},
  {\emph{Communications in Mathematical Physics} {\bfseries 157} (1993) 83}.

\bibitem{borchers1996half}
H.~Borchers, \emph{Half-sided modular inclusion and the construction of the
  poincar{\'e} group}, {\emph{Communications in mathematical physics}
  {\bfseries 179} (1996) 703}.

\bibitem{araki2005extension}
H.~Araki and L.~Zsid{\'o}, \emph{Extension of the structure theorem of borchers
  and its application to half-sided modular inclusions}, {\emph{Reviews in
  Mathematical Physics} {\bfseries 17} (2005) 491}.

\bibitem{Maldacena:2015waa}
J.~Maldacena, S.~H. Shenker and D.~Stanford, \emph{{A bound on chaos}},
  \href{https://doi.org/10.1007/JHEP08(2016)106}{\emph{JHEP} {\bfseries 08}
  (2016) 106} [\href{https://arxiv.org/abs/1503.01409}{{\ttfamily
  1503.01409}}].

\bibitem{Hartman:2015lfa}
T.~Hartman, S.~Jain and S.~Kundu, \emph{{Causality Constraints in Conformal
  Field Theory}}, \href{https://doi.org/10.1007/JHEP05(2016)099}{\emph{JHEP}
  {\bfseries 05} (2016) 099}
  [\href{https://arxiv.org/abs/1509.00014}{{\ttfamily 1509.00014}}].

\bibitem{fawzi2015quantum}
O.~Fawzi and R.~Renner, \emph{Quantum conditional mutual information and
  approximate markov chains}, {\emph{Communications in Mathematical Physics}
  {\bfseries 340} (2015) 575}.

\bibitem{wilde2015recoverability}
M.~M. Wilde, \emph{Recoverability in quantum information theory}, {\emph{Proc.
  R. Soc. A} {\bfseries 471} (2015) 20150338}.

\bibitem{junge2018universal}
M.~Junge, R.~Renner, D.~Sutter, M.~M. Wilde and A.~Winter, \emph{Universal
  recovery maps and approximate sufficiency of quantum relative entropy},  in
  \emph{Annales Henri Poincar{\'e}}, vol.~19, pp.~2955--2978, Springer, 2018.

\bibitem{swingle2018recovery}
B.~Swingle and Y.~Wang, \emph{Recovery map for fermionic gaussian channels},
  {\emph{arXiv preprint arXiv:1811.04956} (2018) }.

\bibitem{wittenpitp}
E.~Witten, ``Black holes, singularity theorems, and all that.''
  \url{https://static.ias.edu/pitp/2018/sites/pitp/files/gr_lectures_edited.pdf},
  2018.

\bibitem{reeh1961bemerkungen}
H.~Reeh and S.~Schlieder, \emph{Bemerkungen zur unit{\"a}r{\"a}quivalenz von
  lorentzinvarianten feldern}, {\emph{Il Nuovo Cimento (1955-1965)} {\bfseries
  22} (1961) 1051}.

\bibitem{Bisognano:1976za}
J.~J. Bisognano and E.~H. Wichmann, \emph{{On the Duality Condition for Quantum
  Fields}}, \href{https://doi.org/10.1063/1.522898}{\emph{J. Math. Phys.}
  {\bfseries 17} (1976) 303}.

\bibitem{koeller2018local}
J.~Koeller, S.~Leichenauer, A.~Levine and A.~Shahbazi-Moghaddam, \emph{Local
  modular hamiltonians from the quantum null energy condition}, {\emph{Physical
  Review D} {\bfseries 97} (2018) 065011}.

\bibitem{Casini:2017roe}
H.~Casini, E.~Teste and G.~Torroba, \emph{{Modular Hamiltonians on the null
  plane and the Markov property of the vacuum state}},
  \href{https://doi.org/10.1088/1751-8121/aa7eaa}{\emph{J. Phys.} {\bfseries
  A50} (2017) 364001} [\href{https://arxiv.org/abs/1703.10656}{{\ttfamily
  1703.10656}}].

\bibitem{borchers1995use}
H.-J. Borchers, \emph{On the use of modular groups in quantum field theory},
  in \emph{Annales de l'Institut Henri Poincare-A Physique Theorique}, vol.~63,
  pp.~331--382, Paris: Gauthier-Villars, c1983-c1999., 1995.

\bibitem{buchholz1990nuclear}
D.~Buchholz, C.~D'Antoni and R.~Longo, \emph{Nuclear maps and modular
  structures. i. general properties}, {\emph{Journal of Functional Analysis}
  {\bfseries 88} (1990) 233}.

\bibitem{Jefferson:2018ksk}
R.~Jefferson, \emph{{Comments on black hole interiors and modular inclusions}},
   \href{https://arxiv.org/abs/1811.08900}{{\ttfamily 1811.08900}}.

\bibitem{araki1976relative}
H.~Araki, \emph{Relative entropy of states of von neumann algebras},
  {\emph{Publications of the Research Institute for Mathematical Sciences}
  {\bfseries 11} (1976) 809}.

\bibitem{araki1977relative}
H.~Araki, \emph{Relative entropy for states of von neumann algebras ii},
  {\emph{Publications of the Research Institute for Mathematical Sciences}
  {\bfseries 13} (1977) 173}.

\bibitem{araki1982positive}
H.~Araki and T.~Masuda, \emph{Positive cones and lp-spaces for von neumann
  algebras}, {\emph{Publications of the Research Institute for Mathematical
  Sciences} {\bfseries 18} (1982) 759}.

\bibitem{uhlmann1977relative}
A.~Uhlmann, \emph{Relative entropy and the wigner-yanase-dyson-lieb concavity
  in an interpolation theory}, {\emph{Communications in Mathematical Physics}
  {\bfseries 54} (1977) 21}.

\bibitem{ohya2004quantum}
M.~Ohya and D.~Petz, \emph{Quantum entropy and its use}. Springer Science \&
  Business Media, 2004.

\bibitem{lieb1973proof}
E.~H. Lieb and M.~B. Ruskai, \emph{Proof of the strong subadditivity of
  quantum-mechanical entropy}, {\emph{Journal of Mathematical Physics}
  {\bfseries 14} (1973) 1938}.

\bibitem{bousso2015entropy}
R.~Bousso, H.~Casini, Z.~Fisher and J.~Maldacena, \emph{Entropy on a null
  surface for interacting quantum field theories and the bousso bound},
  {\emph{Physical Review D} {\bfseries 91} (2015) 084030}.

\bibitem{blanco2013localization}
D.~D. Blanco and H.~Casini, \emph{Localization of negative energy and the
  bekenstein bound}, {\emph{Physical review letters} {\bfseries 111} (2013)
  221601}.

\bibitem{Jafferis:2014lza}
D.~L. Jafferis and S.~J. Suh, \emph{{The Gravity Duals of Modular
  Hamiltonians}}, \href{https://doi.org/10.1007/JHEP09(2016)068}{\emph{JHEP}
  {\bfseries 09} (2016) 068} [\href{https://arxiv.org/abs/1412.8465}{{\ttfamily
  1412.8465}}].

\bibitem{Faulkner:2018faa}
T.~Faulkner, M.~Li and H.~Wang, \emph{{A modular toolkit for bulk
  reconstruction}},  \href{https://arxiv.org/abs/1806.10560}{{\ttfamily
  1806.10560}}.

\bibitem{Shenker:2013pqa}
S.~H. Shenker and D.~Stanford, \emph{{Black holes and the butterfly effect}},
  \href{https://doi.org/10.1007/JHEP03(2014)067}{\emph{JHEP} {\bfseries 03}
  (2014) 067} [\href{https://arxiv.org/abs/1306.0622}{{\ttfamily 1306.0622}}].

\bibitem{Caron-Huot:2017vep}
S.~Caron-Huot, \emph{{Analyticity in Spin in Conformal Theories}},
  \href{https://doi.org/10.1007/JHEP09(2017)078}{\emph{JHEP} {\bfseries 09}
  (2017) 078} [\href{https://arxiv.org/abs/1703.00278}{{\ttfamily
  1703.00278}}].

\bibitem{araki1974some}
H.~Araki, \emph{Some properties of modular conjugation operator of von neumann
  algebras and a non-commutative radon-nikodym theorem with a chain rule},
  {\emph{Pacific Journal of Mathematics} {\bfseries 50} (1974) 309}.

\bibitem{araki1973relative}
H.~Araki, \emph{Relative hamiltonian for faithful normal states of a von
  neumann algebra}, {\emph{Publications of the Research Institute for
  Mathematical Sciences} {\bfseries 9} (1973) 165}.

\bibitem{Witten:2018zxz}
E.~Witten, \emph{{Notes on Some Entanglement Properties of Quantum Field
  Theory}},  \href{https://arxiv.org/abs/1803.04993}{{\ttfamily 1803.04993}}.

\bibitem{schiff2013normal}
J.~L. Schiff, \emph{Normal families}. Springer Science \& Business Media, 2013.

\bibitem{Lashkari:2018nsl}
N.~Lashkari, \emph{{Constraining Quantum Fields using Modular Theory}},
  \href{https://arxiv.org/abs/1810.09306}{{\ttfamily 1810.09306}}.

\bibitem{petz1986sufficient}
D.~Petz, \emph{Sufficient subalgebras and the relative entropy of states of a
  von neumann algebra}, {\emph{Communications in mathematical physics}
  {\bfseries 105} (1986) 123}.

\bibitem{petz1988sufficiency}
D.~Petz, \emph{Sufficiency of channels over von neumann algebras}, {\emph{The
  Quarterly Journal of Mathematics} {\bfseries 39} (1988) 97}.

\bibitem{accardi1982conditional}
L.~Accardi and C.~Cecchini, \emph{Conditional expectations in von neumann
  algebras and a theorem of takesaki}, {\emph{Journal of Functional Analysis}
  {\bfseries 45} (1982) 245}.

\bibitem{uhlmann1976transition}
A.~Uhlmann, \emph{The “transition probability” in the state space of
  a$^*$-algebra}, {\emph{Reports on Mathematical Physics} {\bfseries 9} (1976)
  273}.

\bibitem{kosaki1984applications}
H.~Kosaki, \emph{Applications of the complex interpolation method to a von
  neumann algebra: non-commutative lp-spaces}, {\emph{Journal of Functional
  Analysis} {\bfseries 56} (1984) 29}.

\bibitem{haagerup1979lp}
U.~Haagerup, \emph{Lp-spaces associated with an arbitrary von neumann algebra},
   in \emph{Algebres d’op{\'e}rateurs et leurs applications en physique
  math{\'e}matique (Proc. Colloq., Marseille, 1977)}, vol.~274, pp.~175--184,
  1979.

\bibitem{Leichenauer:2018obf}
S.~Leichenauer, A.~Levine and A.~Shahbazi-Moghaddam, \emph{{Energy density from
  second shape variations of the von Neumann entropy}},
  \href{https://doi.org/10.1103/PhysRevD.98.086013}{\emph{Phys. Rev.}
  {\bfseries D98} (2018) 086013}
  [\href{https://arxiv.org/abs/1802.02584}{{\ttfamily 1802.02584}}].

\bibitem{qnecsat}
S.~Balakrishnan, V.~Chandrasekaran, T.~Faulkner, A.~Levine and
  A.~Shahbazi-Moghaddam, \emph{{Entropy Variations in Defect Conformal Field
  Theory (to appear)}}, .

\bibitem{Hofman:2008ar}
D.~M. Hofman and J.~Maldacena, \emph{{Conformal collider physics: Energy and
  charge correlations}},
  \href{https://doi.org/10.1088/1126-6708/2008/05/012}{\emph{JHEP} {\bfseries
  05} (2008) 012} [\href{https://arxiv.org/abs/0803.1467}{{\ttfamily
  0803.1467}}].

\bibitem{Callebaut:2018nlq}
N.~Callebaut and H.~Verlinde, \emph{{Entanglement Dynamics in 2D CFT with
  Boundary: Entropic origin of JT gravity and Schwarzian QM}},
  \href{https://arxiv.org/abs/1808.05583}{{\ttfamily 1808.05583}}.

\bibitem{Cordova:2018ygx}
C.~Cordova and S.-H. Shao, \emph{{Light-ray Operators and the BMS Algebra}},
  \href{https://arxiv.org/abs/1810.05706}{{\ttfamily 1810.05706}}.

\bibitem{Engelhardt:2018kcs}
N.~Engelhardt and A.~C. Wall, \emph{{Coarse Graining Holographic Black Holes}},
   \href{https://arxiv.org/abs/1806.01281}{{\ttfamily 1806.01281}}.

\bibitem{Engelhardt:2017aux}
N.~Engelhardt and A.~C. Wall, \emph{{Decoding the Apparent Horizon: A
  Coarse-Grained Holographic Entropy}},
  \href{https://doi.org/10.1103/PhysRevLett.121.211301}{\emph{Phys. Rev. Lett.}
  {\bfseries 121} (2018) 211301}
  [\href{https://arxiv.org/abs/1706.02038}{{\ttfamily 1706.02038}}].

\bibitem{Neuenfeld:2018dim}
D.~Neuenfeld, K.~Saraswat and M.~Van~Raamsdonk, \emph{{Positive gravitational
  subsystem energies from CFT cone relative entropies}},
  \href{https://doi.org/10.1007/JHEP06(2018)050}{\emph{JHEP} {\bfseries 06}
  (2018) 050} [\href{https://arxiv.org/abs/1802.01585}{{\ttfamily
  1802.01585}}].

\bibitem{Casini:2018kzx}
H.~Casini, E.~Teste and G.~Torroba, \emph{{All the entropies on the
  light-cone}}, \href{https://doi.org/10.1007/JHEP05(2018)005}{\emph{JHEP}
  {\bfseries 05} (2018) 005}
  [\href{https://arxiv.org/abs/1802.04278}{{\ttfamily 1802.04278}}].

\end{thebibliography}\endgroup
%\input{atoq-v2.bbl}

% Please avoid comments such as "For a review'', "For some examples",
% "and references therein" or move them in the text. In general,
% please leave only references in the bibliography and move all
% accessory text in footnotes.

% Also, please have only one work for each \bibitem.

\end{document}